\documentclass[aps,twocolumn,pra,floatfix,groupedaddress,nofootinbib,notitlepage,showpacs,superscriptaddress]{revtex4-1}

\usepackage{amsthm}
\usepackage{latexsym}
\usepackage{float}
\usepackage{appendix}
\usepackage{epstopdf}
\usepackage{physics}
\usepackage[colorlinks=true,linkcolor=blue,citecolor=green,plainpages=false,pdfpagelabels]{hyperref}
\usepackage[normalem]{ulem}
\usepackage{newlfont}
\usepackage{amsfonts}

\usepackage{epsfig}
\usepackage{mathrsfs}
\usepackage[thinc]{esdiff}

\newtheorem{lemma}{Lemma}
\newtheorem{proposition}{Proposition}
\newtheorem{corollary}{Corollary}

\newtheorem{theorem}{Theorem}

\DeclareOldFontCommand{\rm}{\normalfont\rmfamily}{\mathrm}
\usepackage{bigints}

\usepackage{braket}

\usepackage{caption}
\usepackage{subcaption}
\usepackage{setspace}
\usepackage{amsmath}
\usepackage{times,bm,bbm,bbold,graphicx,graphics,amssymb,amsmath,amsfonts,dsfont,hyperref,color}
\usepackage{mathrsfs,cancel}
\usepackage[utf8]{inputenc}
\usepackage{soul,xcolor}
\setstcolor{red}
\definecolor{mycolor}{rgb}{0.1, 0.1, 0.7}
\usepackage{dsfont}
\usepackage{url}
\usepackage{soul}
\usepackage[vlined,ruled]{algorithm2e}
\usepackage{stackengine}

\DeclareMathAlphabet{\mathpzc}{OT1}{pzc}%
{m}{it}
\hypersetup{
	plainpages=true,
	breaklinks=true,
	hypertexnames=false,
	pageanchor=true,
	colorlinks=true,
	linkcolor={mycolor},
	citecolor={mycolor},
	urlcolor={mycolor},
	anchorcolor={black}
}

\begin{document}

\title{State-dependent and state-independent uncertainty relations for skew information and standard deviation}

\author{Sahil}
\email{sahilmd@imsc.res.in}
\affiliation{Optics and Quantum Information Group, The Institute of Mathematical Sciences,
CIT Campus, Taramani, Chennai 600113, India}
\affiliation{Homi Bhabha National Institute, Training School Complex, Anushakti Nagar, Mumbai 400085, India}

\begin{abstract}
  In this work, we derive state-dependent uncertainty relations (uncertainty equalities)  in which  commutators of incompatible operators (not necessarily Hermitian) are explicitly present  and state-independent  uncertainty relations based on the  Wigner-Yanase (-Dyson) skew information.   We derive uncertainty equality based on  standard deviation for incompatible operators with mixed states, a generalization of  previous works in which only pure states were considered.  We show that  for pure states,  the  Wigner-Yanase  skew information based  state-independent uncertainty relations become standard deviation based  state-independent uncertainty relations which turn out to be tighter  uncertainty relations for some cases than the ones  given in   previous works, and we generalize the previous works   for arbitrary operators.  As the Wigner-Yanase skew information of a quantum channel can be considered as a measure of quantum coherence of a density operator with respect to that channel, we  show that  there exists a state-independent uncertainty relation for the coherence measures of the density operator with respect to a collection of different channels.  We show that  state-dependent and state-independent uncertainty relations based on  a more general version of skew information called generalized skew information   which includes the Wigner-Yanase (-Dyson) skew information and the Fisher information as  special cases hold.  In qubits, we derive tighter state-independent uncertainty  inequalities for  different  form of generalized skew informations and standard deviations, and state-independent uncertainty equalities involving generalized skew informations and standard deviations of spin operators along three orthogonal directions. Finally, we provide a scheme to determine the Wigner-Yanase (-Dyson) skew information of an unknown observable  which can be performed in experiment using the notion of weak values. 
\end{abstract}

\maketitle

\section{Introduction}

Whenever two or more incompatible (noncommuting) observables are considered, we are always curious to know about the impact (in the form of  uncertainty relation) of the noncommutative nature of those incompatible observables on each others' measurement results.  Standard deviation based uncertainty relation \cite{robertson} known as the  Robertson-Heisenberg uncertainty relation (RHUR) is familiar  from the inception of modern quantum mechanics and provides  the  bound on the sharp preparation of a quantum state for two noncommuting observables. Thus, there can not exist a quantum state for which we can have both the standard deviations of the two noncommuting observables zero.  See Refs. \cite{Busch-2014,Bagchi-2016,Mondal-2017,Huang-2012,Guise-2018,Abbott-2016,Werner-2015,Werner-2017,Karol-2019} for improved and tighter state-dependent and state-independent  uncertainty relations based on standard deviation. Entropic uncertainty relation says that there is a tradeoff between uncertainties  (Shannon entropy) about two incompatible observables captured by the inequality known as the Maassen-Uffink entropic uncertainty relation \cite{maassen-uffink}.  Recently, it has been confirmed that there exists an uncertainty relation in a pre- and postselected  system describing the impossibility of joint sharp preparation of pre- and postselected states for  measuring incompatible observables \cite{sahil-sohail-sibasish}.  \par
In recent years, skew information, originally introduced by Wigner and Yanase \cite{wigner-yanase} to  quantify the quantum uncertainty in measurements under conservation laws, became an important information content in quantum information theory. For example, skew information has been regarded as quantum uncertainty of an observable  in a mixed state  (a  classical mixture of  quantum states) \cite{Luo-07-05}, a measure of the noncommutativity between the observable and the density operator \cite{Alain-1978}, a particular kind of quantum Fisher information in the theory of statistical estimation \cite{Luo-2003}, measures of coherence and  asymmetry of the density operator  with respect to  observables \cite{Girolami-Luo,Marvian-2016}. Moreover, the skew information of an observable becomes standard deviation of the same observable when the initially prepared state is pure \cite{Luo-Zhang}. All of these different manifestations of skew information demonstrate its importance in  information theory. \par
As the  skew information is observable dependent,  one may ask whether there exist  uncertainty relations for incompatible observables based on the Wigner-Yanase skew information.   Product \cite{Luo-2003,Luo-Zhang} and  sum \cite{Chen-2016} uncertainty relations based on skew information have been derived but particularly, the product uncertainty relations   are found to be incorrect (see  Sec. \ref{II}  for the actual reason) via counter examples which show the violation of the product uncertainty inequalities \cite{Yanagi-2005,Furuichi-2009,Yanagi-2010}.   As the Wigner-Yanase skew information can be expressed as the Frobenius norm of the commutator of square root of density operator and the observable, one can simply use Cauchy-Schwarz inequality based on Frobenius norm (see  section \ref{II}) to derive uncertainty relation for the Wigner-Yanase skew informations of two incompatible observables. But such uncertainty relation doesn't contain the commutator of the incompatible observables. The same problem  occurs in sum uncertainty relations based on the Wigner-Yanase skew informations. Uncertainty relation based on standard deviation (known as the RHUR) is not only fundamentally  important but also from application's perspective: to formulate  quantum mechanics \cite{Busch-2007,Lahti-1987}, for entanglement detection \cite{Hofmann-2003,Guhne-2004}, for the security analysis of quantum key distribution in quantum cryptography \cite{Fuchs-1996}, as a fundamental building block for quantum mechanics and quantum gravity \cite{mjwhall-2005}, etc. Thus one should  expect  uncertainty relations based on skew information  to play such important roles in application and reveal some new properties of quantum systems.\par
In this work, we derive  uncertainty equality  based on the Wigner-Yanase (-Dyson) skew information   where the commutator of two incompatible operators (not necessarily Hermitian) is explicitly present. Then we show that such uncertainty equality   can be used to solve the triviality issue of uncertainty relations based on the Wigner-Yanase skew information and  standard deviation.  It is also shown how one can derive  quantum speed limit for observable (state) using the commutator term in the  uncertainty equality.  We derive sum and product uncertainty equalities based on standard deviation with explicit presence of commutator of incompatible operators when the initially prepared state is an arbitrary mixed state, a generalization of the works of Ref. \cite{Maccone-Pati, Yao-2015} where only pure states were considered. We provide detail calculation of sum and product uncertainty equalities for three operators (can be extended for multiple operators also). It is shown that variance (skew information)  of any operator can alway be written as sum of variances (skew informations) of two Hermitian operators.  For the first time, we derive  state-independent uncertainty relation based on the Wigner-Yanase (-Dyson)  skew information for a collection  of arbitrary operators  using the  eigenvalue minimization procedure, a similar technique   first appeared in Ref. \cite{Giorda-2019} to derive the state-independent uncertainty relation based on  standard deviation. Also, we generalize  the work of  Ref. \cite{Giorda-2019}  for   a collection  of arbitrary operators. By utilizing  skew information based state-independent uncertainty relation, we derive standard deviation based state-independent uncertainty relation  with improvement of lower bounds  for some cases compared to  \cite{Giorda-2019}. \par
{Recently,  the authors in  Ref. \cite{Luo-Sun-2017} have  defined a coherence measure of a state  with respect to the L$\ddot{\emph{u}}$ders measurement (or equivalently  an orthonormal basis) involving the Wigner-Yanase  skew informations and later they have generalized the coherence measure of the state with respect to  a quantum channel  \cite{Luo-Sun-2018}.  Now as the measure of coherence depends on  orthonormal basis or more generally on a channel involving the Wigner-Yanase skew information,  it is natural to ask if there is a trade-off relation between the coherence measures of a state in different orthonormal bases (or with respect to  channels). Singh \emph{et al.} \cite{Singh-Pati-2016} derived such trade-off relation for relative entropy  coherence measures of a state in different orthonormal bases,  and Fu \emph{et al.} \cite{Fu-Luo-2019} derived  state-dependent trade-off relation for  coherence measures  with respect to two channels based on the Wigner-Yanase skew information. Here  we derive state-independent  uncertainty relation  for  coherence measures  of a state with respect to a collection of  channels based on the Wigner-Yanase skew information.}  \par
The Wigner-Yanase (-Dyson)  skew information and the Fisher information have been shown to be the  special cases of a quantity known as generalized skew information \cite{Yang-2022}. As generalized skew information is lower bounded by the Wigner-Yanase skew information \cite{Yang-2022},  we show that the state-dependent and state-independent uncertainty relations for generalized skew information can easily be derived using the state-dependent and state-independent uncertainty relations for the Wigner-Yanase  skew informations. In qubits,  we show how different form of generalized skew informations, for example, the Wigner-Yanase  skew information and the Fisher information   are connected via equalities. We study state-independent tighter uncertainty relations involving different form of generalized skew informations (\emph{e.g}.,  the Wigner-Yanase skew information and the Fisher information and so on). Also,  we study state-independent  tighter uncertainty relations involving  different form of generalized skew informations and standard deviations, and state-independent uncertainty equalities involving generalized skew informations and standard deviations of spin operators along three orthogonal directions.  Finally, we show  that the Wigner-Yanase (-Dyson) skew information can experimentally be obtained using the notion of weak values. \par
This paper is organized as follows. In Sec. \ref{II}, we discuss preliminaries of standard deviation and skew information, and related uncertainty relations. We derive our main results of this paper in Sec. \ref{III}. In particular,  in Sec. \ref{IIIA}, we derive state-dependent uncertainty equalities and inequalities based on skew information and standard deviation.  In Sec. \ref{IIIB}, we derive state-independent uncertainty relations for the Wigner-Yanase (-Dyson) skew informations and provide   several examples, and also analyse the situation when the density operator is a pure state.  We derive  state-independent uncertainty relation for the Wigner-Yanase skew informations of a collection of quantum channels in Sec. \ref{IIIC}.  We show      in Sec. \ref{IIID} that state-dependent and state-independent uncertainty relations based on a more general version of skew information called generalized skew information hold, and state-independent uncertainty inequalities and equalities for different form of generalized skew informations and standard deviations are derived. In Sec. \ref{IIIE}, we provide a weak value based scheme to determine the Wigner-Yanase (-Dyson) skew information in experiment. Finally we conclude our work  in Sec. \ref{IV}.

\section{Preliminaries}\label{II}
In this section, we provide preliminaries of standard deviation and skew information with associated properties. After that we discuss the shortcomings of existing uncertainty relations based on standard deviation and skew information. 
\subsection{Standard Deviation and Skew Information}
\subsubsection{Standard Deviation}
Let  $A\in\cal{L}(\cal{H})$ be an  operator, where $\cal{L}(\cal{H})$ is the  space of all  linear operators on $\cal{H}$. Then we define the standard deviation of  $A$  as 
\begin{align}
    \braket{\Delta A}_{\rho}=\sqrt{\Tr[\rho(A^{\dagger}A+AA^{\dagger})/2]-|Tr(A\rho)|^2},\label{IIA1-1}
\end{align}
where $\rho$ is the density operator of the system. The definition of standard deviation in Eq. (\ref{IIA1-1}) is symmetric under the exchange of $A$ and $A^{\dagger}$, and is  different from the one defined in \cite{Leblond-1976, Hall-Pati-2016} where the authors considered standard deviation for a normal operator. For Hermitian operators, $\braket{\Delta A}_{\rho}$ in  Eq. (\ref{IIA1-1}) becomes $\braket{\Delta A}_{\rho}=\sqrt{\Tr(\rho A^2)-[\Tr(\rho A)]^2}$. See Ref. \cite{sahil-sohail-sibasish}, where the authors  have given  meaningful geometric as well as information-theoretic interpretations of the standard deviation  of a Hermitian operator.
\subsubsection{Skew Information}
The  Wigner-Yanase skew information is defined as \cite{wigner-yanase}
\begin{align}
    I_{\rho}(A)=\frac{1}{2}Tr([\sqrt{\rho},A]^{\dagger}[\sqrt{\rho},A]).\label{IIA2-1}
\end{align}
   A more general version of the Wigner-Yanase skew information extended by Dyson, now known as the Wigner-Yanase-Dyson skew information, is defined as 
\begin{align}
    I^s_{\rho}(A)=\frac{1}{2}Tr([\rho^s,A]^{\dagger}[\rho^{1-s},A]),\label{IIA2-2}
\end{align}
where $0<s<1$ and obviously, $I^{s=\scriptsize{{1}/{2}}}_{\rho}(A)=I_{\rho}(A)$. An equivalent expression of $I^s_{\rho}(A)$ for the Hermitian operator $A$ is:
\begin{align}
    I^s_{\rho}(A)=Tr(A^2\rho)-Tr(\rho^{1-s}A\rho^{s}A).\label{IIA2-3}
\end{align}
 Several remarkable properties make the Wigner-Yanase-Dyson skew information interesting:\\
($i$) $I^s_{\rho}(A)$ is  convex under the classical mixing of density operators  \cite{Lieb-1973} \emph{i.e}., 
\begin{align}
    I^s_{\rho}(A)\leq \sum_{i}{p_iI^s_{\rho_i}(A)},\hspace{5mm} where \hspace{5mm}\rho=\sum_i{p_i\rho_i}.\label{IIA2-4}
\end{align}
($ii$) $I^s_{\rho}(A)$ is additive in the following sense,
\begin{align}
    I^s_{\rho_1\otimes\rho_2}(A)=I^s_{\rho_1}(A)+I^s_{\rho_2}(A).\label{IIA2-5}
\end{align}
Clearly, when $s=\frac{1}{2}$, the Wigner-Yanase-Dyson skew information $I^s_{\rho}(A)$ becomes the Frobenius  norm of the operator $\frac{1}{\sqrt{2}}[\sqrt{\rho},A]$, and hence, it is a non-negative quantity. Here, Frobenius  norm of an operator $X$ is defined as $||X||_{F}=\sqrt{Tr(X^{\dagger}X)}$. To see how $I^s_{\rho}(A)$ is non-negative as well as upper bounded by $I_{\rho}(A)$, one may use the following inequality which holds for any unitarily invariant norm (\emph{e.g}.,  Frobenius norm) \cite{Bhatia-Davis-1993},
\begin{align*}
    2||M^{\small{1/2}}XN^{\small{1/2}}||_F&\leq ||M^sXN^{1-s}+M^{1-s}XN^s||_F\nonumber\\
    &\leq ||MX+XN||_F,
\end{align*}
where $X$ is an arbitrary operator on $\mathcal{H}$, and $M$ and $N$ are positive semidefinite operators ($M,N\geq 0$). Now, by substituting $X=A$, and $M=N=\rho^{\small{1/2}}$ in the above inequality, and using the definition of $I^s_{\rho}(A)$ and $I_{\rho}(A)$, we have the following 
\begin{align*}
    0\leq I^s_{\rho}(A)\leq I_{\rho}(A).
\end{align*}
As the variance of $A$ is concave under the classical  mixing of density operators \emph{i.e}., $V_{\rho}(A)\geq\sum_ip_iV_{\rho_i}(A)$, where $V_{\rho}(A)=\braket{\Delta A}_{\rho}^2$, the above inequality becomes
\begin{align}
    0\leq I^s_{\rho}(A)\leq I_{\rho}(A)\leq V_{\rho}(A).\label{IIA2-6}
\end{align}

\subsection{Uncertainty relations}
\subsubsection{Based on standard deviation}
The RHUR and its tighter version are respectively  given by 
\begin{align}
\braket{\Delta{A}}_{\rho}^{{2}}\braket{\Delta{B}}_{\rho}^{{2}}&\geq \left[\frac{1}{2i}Tr([A,B]\rho)\right]^2, \label{IIB1-1}
\end{align}
\begin{align}
\braket{\Delta{A}}_{\rho}^{{2}}\braket{\Delta{B}}_{\rho}^{{2}}&\geq \left[\frac{1}{2i}Tr([A,B]\rho)\right]^2 \nonumber\\
&+ \left[\frac{1}{2}Tr(\{A,B\}\rho) - \braket{A}_{\rho}\braket{B}_{\rho}\right]^2.\label{IIB1-2}
\end{align}
The implication of the RHUR is that whenever $A$ and $B$ are noncommuting \emph{i.e}., $[A,B]\neq 0$, both the uncertainties $\braket{\Delta{A}}_{\rho}^{{2}}$ and $\braket{\Delta{B}}_{\rho}^{{2}}$ can not arbitrarily be small  in the same prepared state.\par
The RHUR has drawback in the sense that when $\Tr([A,B]\rho)=0$ even if $[A,B] \neq 0$, the lower bound becomes zero in Eq. (\ref{IIB1-1}) and thus the RHUR turns in to a  trivial uncertainty relation.  Eq. (\ref{IIB1-2}) can be used to have a non-zero lower bound when $\Tr([A,B]\rho)=0$, but the commutator $[A,B]$  term disappears from the relation (\ref{IIB1-2}). Note that if $\rho$ is a pure state and is an eigenstate of either $A$ or $B$, the relation (\ref{IIB1-2}) also becomes trivial. To solve this triviality issue of Eq. (\ref{IIB1-1}) and to have the commutator term $[A,B]$ explicitly  in the lower bound   of the uncertainty relation based on standard deviation, we can use the uncertainty relation based on the Wigner-Yanase-Dyson skew information   (see \emph{Corollary \ref{Corollary 1}}).

\subsubsection{Based on skew information}
One can derive an uncertainty  relation for two incompatible observables $A$ and $B$ involving the Wigner-Yanase skew informations $I_{\rho}(A)$ and $I_{\rho}(B)$, respectively using the Cauchy-Schwarz inequality based on Frobenius norm as
\begin{align}
I_{\rho}(A) I_{\rho}(B) \geq \left|{\frac{1}{2}Tr\left(\{A,B\}\rho\right)-Tr\left(\sqrt{\rho}A\sqrt{\rho}B\right)}\right|^2.\label{IIB2-1}
\end{align}
The sum uncertainty relation reads \cite{Chen-2016} 
\begin{align}
    I_{\rho}(A)+I_{\rho}(B)\geq \frac{1}{2}max\{I_{\rho}(A\pm B)\}.\label{IIB2-2}
\end{align}
Notice that the commutator of $A$ and $B$ is not explicitly present in both the inequalities  (\ref{IIB2-1}) and (\ref{IIB2-2}), and hence a direct impact of the noncommutativity of the observables is not captured.  See Ref. \cite{Cai-2021} where the author derived sum uncertainty relations for metric-adjusted skew information in which the commutator term is absent. Note that a particular case of the metric-adjusted skew information is the Wigner-Yanase-Dyson skew information and thus uncertainty relations based on  the Wigner-Yanase-Dyson skew information also do not contain the commutator of $A$ and $B$.\par
 An initial  attempt was made in the Refs. \cite{Luo-2003,Luo-Zhang} to include the  commutator in the uncertainty relation. But it was pointed out in Ref. \cite{Yanagi-2005} by a counter example  that the uncertainty relation in Refs. \cite{Luo-2003,Luo-Zhang}  is violated. The basic problem in the inner product $(A,B)=Tr(\rho AB)-Tr(\sqrt{\rho}A\sqrt{\rho}B)$ between two operators $A$ and $B$ which Luo defined in Ref. \cite{Luo-2003}  is that the inner product $(A,B)$ doesn't satisfy the semi-definite property \emph{i.e}.,  $(A,A)=0$ if and only if $A=0$. It is easy to see that $(A,A)=I_{\rho}(A)=0$ holds also for  those $A\neq 0$ for which $[\sqrt{\rho},A]=0$. \par
Later, another attempt was made  by  Luo \cite{Luo-2005} to include the commutator of $A$ and $B$ in the following uncertainty relation
\begin{align}
    U_{\rho}(A)U_{\rho}(B)\geq \frac{1}{4}|Tr(\rho [A,B])|^2,\label{IIB2-3}
\end{align}
where $U_{\rho}(X)=\sqrt{V_{\rho}(A)^2-[V_{\rho}(A)-I_{\rho}(A)]^2}=\sqrt{I_{\rho}(A)J_{\rho}(A)}$ with $J_{\rho}(A)=2V_{\rho}(A)-I_{\rho}(A)$. Although the Eq. (\ref{IIB2-3}) contains the commutator of $A$ and $B$, the left hand side  is not the product of the Wigner-Yanase skew informations and hence it doesn't reflect the true uncertainty relation for the Wigner-Yanase skew informations with the explicit form of the commutator.  Another issue with  Eq. (\ref{IIB2-3}) is that it is not a purely quantum uncertainty relation as $U_{\rho}(A)$ is the product of $I_{\rho}(A)$ [a convex function under classical mixing of density operators, see Eq. (\ref{IIA2-4})] and $J_{\rho}(A)$ [a concave function under the same classical mixing of density operators]. Note that,  a purely quantum uncertainty (or variance, or variance-like) should not increase  under the classical mixing of density operators. The expected uncertainty which one intends to define [instead of $I_{\rho}(A)$] should be non-increasing under the classical mixing of density matrices as $I_{\rho}(A)$ has the same property in this regard. The authors of  Ref. \cite{sahil-sohail-sibasish} have shown that   purely quantum uncertainty relations exist in pre- and postselected quantum systems.

\section{Main results}\label{III}

\subsection{Uncertainty equalities} \label{IIIA}
\subsubsection{Based on standard deviation for mixed states}
The sum and product uncertainty equalities based on standard deviation containing  the commutator of two incompatible observables explicitly for all pure states have been derived in Ref. \cite{Maccone-Pati, Yao-2015}.   Uncertainty equalities are particularly important to obtain  tighter uncertainty inequalities with hierarchical structure.  Moreover, recently in Ref. \cite{Bagchi-Pati-2023}, stronger quantum speed limit for mixed quantum states has been derived.  As our product uncertainty relation (\ref{IIIA1-2}) is an equality, we can derive  stronger quantum speed limit for mixed quantum states using the same technique used in Ref. \cite{Bagchi-Pati-2023}. Note that exact quantum speed limits for  pure states have recently been derived \cite{pati-brij-sahil-SLB} using the exact uncertainty relation  \cite{Hall-2001}.  As non-Hermitian operators have recently become important in the context of uncertainty relations  \cite{Massar-Spindel,Shrobona-Pati,Bong-2018,Yu-2019,Zhao-2024}, here we provide sum and  product  uncertainty equalities for any two operators when the initially prepared state  is an arbitrary mixed state. Such uncertainty equalities contain the commutator  of the  two concerned operators. See Theorem \ref{Theorem 3} where we explain how one can measure variance of a non-Hermitian operator.
\begin{theorem}\label{Theorem 1}
 For any two operators  $A$, $B\in\mathcal{L}(\mathcal{H})$, the following equalities hold:
     \begin{align}
       \braket{\Delta A}^2_{\rho}+\braket{\Delta B}^2_{\rho}=&\pm \frac{1}{2}\braket{i([A^{\dagger},B]+[A,B^{\dagger}])}_{\rho}\nonumber\\
       &+\frac{1}{2}\braket{(M_{\mp}^{\dagger}M_{\mp}+N_{\pm}N_{\pm}^{\dagger})}_{\rho},\label{IIIA1-1}
        \end{align}
        where $M_{\mp}=A\mp iB-\braket{A\mp iB}_{\rho}I$ and $N_{\pm}=A\pm iB-\braket{A\pm iB}_{\rho}I$. The `$\pm$' sign is chosen  such that the first  term in the right hand side  remains positive, and 
         \begin{align}
       \braket{\Delta A}_{\rho}\braket{\Delta B}_{\rho}=\frac{\pm \frac{1}{4}\braket{i([A^{\dagger},B]+[A,B^{\dagger}])}_{\rho}}{1-\frac{1}{4}\braket{(R_{\mp}^{\dagger}R_{\mp}+S_{\pm}S_{\pm}^{\dagger})}_{\rho}},\label{IIIA1-2}
        \end{align}
where $R_{\mp}=\frac{A}{\braket{\Delta A}_{\rho}}\mp i\frac{B}{\braket{\Delta B}_{\rho}}-\braket{\frac{A}{\braket{\Delta A}_{\rho}}\mp i\frac{B}{\braket{\Delta B}_{\rho}}}_{\rho}I$ and $S_{\pm}=\frac{A}{\braket{\Delta A}_{\rho}}\pm i\frac{B}{\braket{\Delta B}_{\rho}}-\braket{\frac{A}{\braket{\Delta A}_{\rho}}\pm i\frac{B}{\braket{\Delta B}_{\rho}}}_{\rho}I$. The `$\pm$' sign is chosen  such  that the numerator in the right hand side  remains positive. We assume that $\braket{\Delta A}_{\rho}\neq 0$, $\braket{\Delta B}_{\rho}\neq 0$.
\end{theorem}
\begin{proof}
The proof of Eqs. (\ref{IIIA1-1}) and (\ref{IIIA1-2}) are given in Appendix \ref{A}.
\end{proof}
 Recently,  tighter uncertainty relations for unitary operators have been derived in \cite{Yu-2019} which are inequalities.  In contrast, the relations (\ref{IIIA1-1}) and (\ref{IIIA1-2}) are  uncertainty equalities with  the commutator  for any two operators when the initially prepared state is a mixed state.  As earlier mentioned that  from the uncertainty equalities, a series of inequalities with hierarchical structure can be obtained.  One can do that by noting, for example, in Eq. (\ref{IIIA1-1}) that  $\braket{(M_{\mp}^{\dagger}M_{\mp}+N_{\pm}N_{\pm}^{\dagger})}_{\rho}=\sum_i[\braket{\phi_i|M_{\mp}\rho M_{\mp}^{\dagger}|\phi_i}+\braket{\phi_i|N_{\pm}^{\dagger}\rho N_{\pm}|\phi_i}]$, where $\sum_i\ketbra{\phi_i}{\phi_i}=I$. Now as each term in the summation is positive number, one or more terms can be discarded suitably and hence a series of inequalities with hierarchical structure can be obtained from Eq. (\ref{IIIA1-1}). One can easily  verify that when $\rho$ is a pure state, and $A$ and $B$ are Hermitian operators,  the equalities  (\ref{IIIA1-1}) and (\ref{IIIA1-2}) become the  same uncertainty equalities first derived in Ref. \cite{Yao-2015}. \par
 To achieve the RHUR uncertainty relations in sum and product form, respectively for Hermitian operators, we have to consider the constraint $\braket{M_{\mp}^{\dagger}M_{\mp}}_{\rho}=0$ and $\braket{R_{\mp}^{\dagger}R_{\mp}}_{\rho}=0$ in Eqs. (\ref{IIIA1-1}) and   (\ref{IIIA1-2}), respectively.  Here $N_{\pm}=M_{\mp}$ and $S_{\pm}=R_{\mp}$ for Hermitian operators $A$ and $B$.  \par
The states for which the terms $\braket{(M_{\mp}^{\dagger}M_{\mp}+N_{\pm}^{\dagger}N_{\pm})}_{\rho}$ and $\braket{(R_{\mp}^{\dagger}R_{\mp}+S_{\pm}^{\dagger}S_{\pm})}_{\rho}$ in Eqs. (\ref{IIIA1-1}) and   (\ref{IIIA1-2}), respectively vanish are called \emph{intelligent states}.  It can be shown that intelligent states must satisfy \{$\sqrt{\rho}M_{\mp}^{\dagger}\ket{\phi_i}=0=\sqrt{\rho}N_{\pm}\ket{\phi_i}\}^d_{i=1}$ in Eq. (\ref{IIIA1-1}) and  \{$\sqrt{\rho}R_{\mp}^{\dagger}\ket{\phi_i}=0=\sqrt{\rho}S_{\pm}\ket{\phi_i}\}^d_{i=1}$ in Eq. (\ref{IIIA1-2}), where \{$\ket{\phi_i}\}^d_{i=1}$ is a basis.\par
The product uncertainty equality (\ref{IIIA1-2}) may become trivial for the case where average values of the commutators are zero. To get rid of the triviality issue of the product uncertainty equality (\ref{IIIA1-2}), one can derive the following product uncertainty equality by making the substitutions $A\rightarrow A \braket{\Delta B}_{\rho}$ and $B\rightarrow B \braket{\Delta A}_{\rho}$ in Eq. (\ref{IIIA1-1}) (after that, divide both the sides by $\braket{\Delta A}_{\rho}\braket{\Delta B}_{\rho}$) as 
\begin{align}
 \braket{\Delta A}_{\rho}\braket{\Delta B}_{\rho}=&\pm \frac{1}{4}\braket{i([A^{\dagger},B]+[A,B^{\dagger}])}_{\rho}\nonumber\\
       &+\frac{1}{4}\braket{(\tilde{M}_{\mp}^{\dagger}\tilde{M}_{\mp}+\tilde{N}_{\pm}\tilde{N}_{\pm}^{\dagger})}_{\rho},\label{IIIA1-3}
\end{align}
 where $\tilde{M}_{\mp}=A\sqrt{\frac{ \braket{\Delta B}_{\rho}}{ \braket{\Delta A}_{\rho}}}\mp iB\sqrt{\frac{ \braket{\Delta A}_{\rho}}{ \braket{\Delta B}_{\rho}}}-\braket{A\sqrt{\frac{ \braket{\Delta B}_{\rho}}{\braket{\Delta A}_{\rho}}}\mp iB\sqrt{\frac{ \braket{\Delta A}_{\rho}}{ \braket{\Delta B}_{\rho}}}}_{\rho}I$ and  $\tilde{N}_{\pm}=A\sqrt{\frac{ \braket{\Delta B}_{\rho}}{ \braket{\Delta A}_{\rho}}}\pm iB\sqrt{\frac{ \braket{\Delta A}_{\rho}}{\braket{\Delta B}_{\rho}}}-\braket{A\sqrt{\frac{ \braket{\Delta B}_{\rho}}{ \braket{\Delta A}_{\rho}}}\pm iB\sqrt{\frac{ \braket{\Delta A}_{\rho}}{ \braket{\Delta B}_{\rho}}}}_{\rho}I$. Thus, even if the average value of the commutator in Eq. (\ref{IIIA1-3}) is zero, the  Eq. (\ref{IIIA1-3}) remains nontrivial  and useful. \par
  The sum uncertainty  equality (\ref{IIIA1-1})  might  also be useful (depending upon the context) to solve the  triviality issue of the RHUR (\ref{IIB1-1}) and (\ref{IIB1-2}). Note that the RHUR can become trivial  if for example, $\rho$ being a pure state and   is an eigenstate of either $A$ or $B$, then  both the sides of Eq. (\ref{IIB1-1}) and (\ref{IIB1-2})  becomes zero. Now even if the initially prepared state is an eigenstate of either $A$ or $B$, both the sides of  the sum uncertainty  equality (\ref{IIIA1-1}) does not become zero and thus the triviality issue the RHUR is solved. See Ref. \cite{sahil-sohail-sibasish,Maccone-Pati} for detail discussion for the solution to the trivial uncertainty relations. \par
  Uncertainty equality for multiple operators can be derived easily using Eq. (\ref{IIIA1-1}). See Appendix \ref{A} for detail discussion.

 \subsubsection{Based on skew information}

Similar to Theorem \ref{Theorem 1}, we now derive uncertainty equality for the Wigner-Yanase-Dyson skew informations of two arbitrary operators. 

\begin{theorem}\label{Theorem 2}
   For two arbitrary operators $A$, $B\in\mathcal{L}(\mathcal{H})$, the following equality holds:
    \begin{align}
       \sqrt{I_{\rho}^s(A)I_{\rho}^s(B)}=\frac{\pm \frac{i}{4}[\Tr([A^{\dagger},B]+[A,B^{\dagger}])\rho^s+\mathcal{E}_{AB}^s]}{1+\Omega_{AB}^s-\frac{1}{4}\Tr[(\xi^{\pm}+\eta^{\mp})(I-\rho^{1-s})]},\label{IIIA2-1}
        \end{align}
    where $\xi^{\pm}=(\frac{A}{\sqrt{I^s_{\rho}(A)}}\pm i\frac{B}{\sqrt{I^s_{\rho}(B)}})^{\dagger}\rho^s(\frac{A}{\sqrt{I^s_{\rho}(A)}}\pm i\frac{B}{\sqrt{I^s_{\rho}(B)}})$ and   $\eta^{\mp}=(\frac{A}{\sqrt{I^s_{\rho}(A)}}\mp i\frac{B}{\sqrt{I^s_{\rho}(B)}})\rho^s(\frac{A}{\sqrt{I^s_{\rho}(A)}}\mp i\frac{B}{\sqrt{I^s_{\rho}(B)}})^{\dagger}$, and $\mathcal{E}^s_{AB}=\Tr(\rho^{1-s}B^{\dagger}\rho^sA)+\Tr(\rho^{1-s}B\rho^sA^{\dagger})-\Tr(\rho^{1-s}A\rho^sB^{\dagger})-\Tr(\rho^{1-s}A^{\dagger}\rho^sB)$ with $\mathcal{E}^{s=1/2}_{AB}=0$, $\Omega_{AB}^s=\frac{\Tr(\{A^{\dagger},A\}\sigma)}{4I_{\rho}^s(A)}+\frac{\Tr(\{B^{\dagger},B\}\sigma)}{4I_{\rho}^s(B)}$ with $\sigma=\rho^s-\rho\geq 0$. We assume here that   $I^s_{\rho}(A)\neq 0$, $I^s_{\rho}(B)\neq 0$.  The sign ``$\pm$" is chosen such that  the numerator always remains positive.
\end{theorem}
\begin{proof}
See Appendix \ref{B} for the proof.    
\end{proof}
The uncertainty equality (\ref{IIIA2-1}) is very similar to the uncertainty equality (\ref{IIIA1-2}). In fact, if the density operator $\rho=\ketbra{\psi}{\psi}$ is a pure  state then both of  them reduce to the same uncertainty equality  which was first derived in \cite{Yao-2015}.  Note that the equality (\ref{IIIA2-1}) contains the commutator of the two operators $A$ and $B$, and thus a direct impact of the  noncommutativity of $A$ and $B$ can explicitly be observed. \par
Theorem \ref{Theorem 2} can be useful to obtain a series of lower bounds for the product of the Wigner-Yanase-Dyson skew informations of two operators by discarding some of the positive terms from $\frac{1}{4}\Tr[(\xi^{\pm}+\eta^{\mp})(I-\rho^{1-s})]$ in the Eq. (\ref{IIIA2-1}) as it is a sum of positive terms: $\frac{1}{4}\Tr[(\xi^{\pm}+\eta^{\mp})(I-\rho^{1-s})]=\frac{1}{4}\sum_i\braket{\phi_i|(\xi^{\pm}+\eta^{\mp})(I-\rho^{1-s})|\phi_i}$, where \{$\ket{\phi_i}\}$ forms a basis and $(I-\rho^{1-s})\geq 0$.  \par
The advantage of having a commutator term in an uncertainty relation is that we can obtain an upper bound for the rate of change of the system's state with respect to time (or other parameters). The quantum speed limits for states and observables are the best examples of it \cite{Mandelstam-Tamm,Deffner-Campbell,Brij-Pati-observable}.  For instance, if we consider $A_t$ to be a time dependent  observable  and $H$ is the system's time independent Hamiltonian, then under the unitary evolution, the Heisenberg equation of motion for the observable is given by $i\hbar\frac{d}{dt}A_t=[A_t,H]$. 
Now we want to find the minimum time taken by the observable $A_t$  to reach  the observable $A_T$ starting  from the observable $A_0$ at time $t=0$ when the skew informations of the system Hamiltonian $H$ and the time dependent observable $A_t$ are given. This can be done in the following way. 
By substituting $A=A_t$, $B=H$ and $s=\frac{1}{2}$ in Eq. (\ref{IIIA2-1}), we have 
\begin{align}
 \left|\frac{d}{dt}\braket{A_t}_{\rho_{\frac{1}{2}}}\right|=\frac{2}{\hbar\Tr\sqrt{\rho}}\sqrt{I_{\rho}(A_t)I_{\rho}(H)}\cal{R}_t,\label{IIIA2-1-1}
\end{align}
where $\rho_{\frac{1}{2}}=\frac{\sqrt{\rho}}{\Tr\sqrt{\rho}}$ is a normalized density operator and $\cal{R}_t=1+\Omega-\frac{1}{2}\Tr[\xi^{\pm}(I-\sqrt{\rho})]$. Here, $\Omega$ and $\xi^{\pm}$ are defined in Eq. (\ref{IIIA2-1}). By integrating the above equality, we obtain
\begin{align}
T&= \frac{\hbar\Tr\sqrt{\rho}}{2\sqrt{I_{\rho}(H)}}\int^{T}_{0}\frac{|d\braket{A_t}_{\rho_{\frac{1}{2}}}|}{\sqrt{I_{\rho}(A_t)}\cal{R}_t}\nonumber\\
&\geq \frac{\hbar\Tr\sqrt{\rho}}{2\sqrt{I_{\rho}(H)}}\left|\int^{T}_{0}\frac{d\braket{A_t}_{\rho_{\frac{1}{2}}}}{\sqrt{I_{\rho}(A_t)}\cal{R}_t}\right|\equiv T^{O}_{QSL}\label{IIIA2-1-2}
\end{align}
where we have used the triangle inequality for the integral to obtain the last inequality. Now if we parametrize the observable by time `$t$', then the right hand side of the inequality (\ref{IIIA2-1-2}) will become an integration with respect to time only. Due to the complex structure of the integral, it might not be an easy task to solve analytically but one can do that by imposing several conditions and restrictions on the observable and the density operator.  Thus the minimum time taken by  the observable $A_t$ to reach  the observable $A_T$ starting  from the observable $A_0$ at time $t=0$  is given by  $T^{O}_{QSL}$.  Quantum speed limit  for states can also be derived from  Eq. (\ref{IIIA2-1-1}).  {The meaning of the inequality (\ref{IIIA2-1-2}) is that it tells us how the lower bound of the  required time $T$ is controlled by the value of the Wigner-Yanase skew information of the Hamiltonian of the system \emph{i.e.,} $I_{\rho}(H)$. In the original quantum speed limit (see \cite{Mandelstam-Tamm,Deffner-Campbell}), it is the standard deviation of the Hamiltonian which plays the  role of controlling lower bound of the required time $T$.  The inequality (\ref{IIIA2-1-2}) reflects that the more the Wigner-Yanase skew information of the system Hamiltonian, the less the  lower bound on the required time $T$ (in general). As skew information is interpreted as measures of  quantum uncertainty, asymmetry, coherence, etc, the quantum speed limits can be improved and understood better while having the knowledge of the skew information in different contexts.} \par

Below we show that  the uncertainty equality (\ref{IIIA2-1}) is also useful to provide the commutator term $[A,B]$ in the lower bound of uncertainty relations for the product of standard deviations of  the observables (Hermitian operators) $A$ and $B$ even when $\Tr([A,B]\rho)=0$ in the RHUR  (\ref{IIB1-1}). Similarly, if $\Tr([A,B]\sqrt{\rho})=0$ for the observables in the  relation (\ref{IIIA2-1}), then one is able to obtain uncertainty inequality for the Wigner-Yanase skew information ($s=\frac{1}{2}$) where the lower bound will contain the commutator $[A,B]$.

\begin{corollary}\label{Corollary 1}
   For two incompatible observables $A$, $B\in\mathcal{L}(\mathcal{H})$,  product of the standard deviations and  product of the Wigner-Yanase skew informations satisfy the following inequalities:
    \begin{align}
        \braket{\Delta A}_{\rho}\hspace{-1mm}\braket{\Delta B}_{\rho}&\geq  \sqrt{I_{\rho}(A)I_{\rho}(B)}\nonumber\\
        &\geq \frac{\pm \frac{i}{2}[Tr([A,B]\rho^s)+\mathcal{E}_{AB}^s]}{1+\Omega_{AB}^s-\frac{1}{2}\Tr[\xi^{\pm}(I-\rho^{1-s})]}.\label{IIIA2-2}
    \end{align}
\end{corollary}
\begin{proof}
    As the inequalities in  (\ref{IIA2-6}) hold for observable $B$ also, it is easy to see by using Eq. (\ref{IIIA2-1}) for Hermitian operators $A$ and $B$ that Eq. (\ref{IIIA2-2})  holds automatically.
\end{proof}
Thus, by taking a suitable value of the variable  `$s$' in Eq. (\ref{IIIA2-2}), the commutator term can not be vanished, and  the commutator    of $A$ and $B$ remains explicit in  standard deviation and the Wigner-Yanase skew information based uncertainty relations in any situation. \par
Before we proceed towards state-independent sum uncertainty relations for  a collection of arbitrary operators based on skew informations, below we show that each skew information  (variance) of an arbitrary operator  can be expressed as the sum of skew informations (variances) of two Hermitian operators. This will be useful to derive  state-independent sum uncertainty relations for  a collection of arbitrary operators based on skew informations.\par
 It can be shown for any operator $A$ that 
 \begin{align}
 A=A_{1}\pm iA_{2},\label{IIIA2-3}
 \end{align}
 where $A_{1}=\frac{1}{2}(A+A^{\dagger})$ and $A_{2}=\mp\frac{i}{2}(A-A^{\dagger})$ are both Hermitian operators.  Then we have the following theorem:
\begin{theorem}\label{Theorem 3}
   The Wigner-Yanase-Dyson skew information (standard deviation)  of any operator $A\in\mathcal{L}(\mathcal{H})$ can always be written as
     \begin{align}
        I_{\rho}^s(A)&= I_{\rho}^s(A_1)+I_{\rho}^s(A_2),\label{IIIA2-4}\\
        \braket{\Delta A}_{\rho}^2&=\braket{\Delta A_1}^2_{\rho}+\braket{\Delta A_2}^2_{\rho},\label{IIIA2-5}
     \end{align}
  where $A_1$ and $A_2$ are defined in Eq. (\ref{IIIA2-3}).
\end{theorem}
\begin{proof}
    Consider Eq. (\ref{IIA2-2}) for a non-Hermitian operator $D_{\pm}$ as 
    \begin{align*}
        I_{\rho}^s(D_{\pm})&=\frac{1}{2}Tr([\sqrt{\rho},D_{\pm}]^{\dagger}[\sqrt{\rho},D_{\pm}])\nonumber\\
        &=\frac{1}{2}\big[Tr(\{D_{\pm},D_{\pm}^{\dagger}\}\rho)-Tr(\rho^{1-s}D_{\pm}^{\dagger}\rho^sD_{\pm})\nonumber\\
        &\hspace{2.8cm}-Tr(\rho^{1-s}D_{\pm}\rho^sD_{\pm}^{\dagger})\big],
    \end{align*}
    where $\{\cdot,\cdot\}$ is the anti-commutator. By substituting  $D_{\pm}=A\pm iB$ in the above equation, we obtain  Eq. (\ref{IIIA2-4}).\par
    To obtain Eq. (\ref{IIIA2-5}), expand the standard deviation of the non-Hermitian operator $A=A_1\pm iA_2$ defined in Eq. (\ref{IIA1-1}) and at the end consider the definition of standard deviations of  Hermitian operators $A_1$ and $A_2$.
\end{proof}
  Eqs. (\ref{IIIA2-4}) and (\ref{IIIA2-5}) suggest that the skew information (or variance) of a non-Hermitian operator can always  be obtained as sum of the skew informations (or variances) of two Hermitian operators.\par

To obtain state-independent uncertainty relation for a collection of  arbitrary operators, we need to impose some condition on those collection of operators. Let us, for example,  consider \{$A_k\}_{k=1}^N$, where $A_k$'s are any operators and by using  Eqs. (\ref{IIIA2-4}) and (\ref{IIIA2-5}), we have
\begin{align}
\sum_{k=1}^N\braket{\Delta A_k}_{\rho}^2&=\sum_{k=1}^N\braket{\Delta A_{k,1}}_{\rho}^2+\braket{\Delta A_{k,2}}_{\rho}^2,\label{IIIA2-6}\\
\sum_{k=1}^N I^s_{\rho}(A_k)&=\sum_{k=1}^N I^s_{\rho}(A_{k,1})+I^s_{\rho}(A_{k,2}).\label{IIIA2-7}
\end{align}
 The condition which we want to impose here from now on  is that there is \emph{no common eigenstate}  between \emph{at least two Hermitian operators}  in   \{$A_{k,1},A_{k,2}\}_{k=1}^N$. This condition  is  necessary to have  non-zero state-independent  lower bounds  for  sum of variances or skew informations of \{$A_k\}_{k=1}^N$.

\subsection{State-independent uncertainty relations}\label{IIIB}
Here,  we provide a detail derivation and analysis with examples of the state-independent uncertainty relations  based on the Wigner-Yanase (-Dyson) skew information. After that we show that  if the density operator is considered to be a pure state, then the state-independent uncertainty relations  based on the Wigner-Yanase (-Dyson) skew information become the state-independent uncertainty relations  based on standard deviation.
\subsubsection{Based on skew information}
 In Ref. \cite{Giorda-2019}, the authors have shown that the  variance of an observable can equivalently  be expressed as an average value of a positive semidefinite operator in a bipartite system and based on that, they have derived state-independent uncertainty relation for arbitrary number of Hermitian operators.  Motivated by their strategy,  we first show that the Wigner-Yanase-Dyson skew information  of any operator $A$ defined in Eq. (\ref{IIA2-2}) can  be realized in the similar way. \par
For a density operator  $\rho=\sum_i\lambda_i\ket{\lambda_i}\bra{{\lambda_i}}$ (spectral decomposition), we  define  unnormalized states $\ket{\Tilde{\Phi}^s}:=\sum_{i}\lambda^s_i\ket{\lambda_i}\ket{\lambda^*_i}$ and $\ket{\Tilde{\Phi}^{1-s}}:=\sum_{i}\lambda^{1-s}_i\ket{\lambda_i}\ket{\lambda^*_i}$, where $\ket{\Tilde{\Phi}^{s}}, \ket{\Tilde{\Phi}^{1-s}}\in \mathcal{H}\otimes\mathcal{H}$ and the state  $\ket{\psi^*}$ is defined in such a way that if $\ket{\psi}=\sum_k\alpha_k\ket{k}$ in the computational basis \{$\ket{k}\}$, then $\ket{\psi^*}=\sum_k\alpha_k^*\ket{k}$.  Also define $H_A:=\frac{1}{\sqrt{2}}(A\otimes I-I\otimes A^T)$, where `$T$'  is the transpose operation with respect to the computational basis.  Now let us calculate 
\begin{align}
    &\braket{\Tilde{\Phi}^s|H_A^{\dagger}H_A|\tilde{\Phi}^{1-s}}\nonumber\\
    &=\frac{1}{2}\sum_{i,j}\lambda_i^s\lambda_j^{1-s}\bra{\lambda_i}\bra{\lambda^*_i}\Big[A^{\dagger}A\otimes I+I\otimes (A^T)^{\dagger}A^T\nonumber\\
    &\hspace{3.15cm}-A^{\dagger}\otimes A^T-A\otimes (A^T)^{\dagger}\Big]\ket{\lambda_j}\ket{\lambda^*_j}\nonumber\\
     &=\frac{1}{2}\sum_{i,j}\lambda_i^s\lambda_j^{1-s}\Big[\braket{\lambda_i|A^{\dagger}A|\lambda_j}\delta_{ij}+\delta_{ij}\braket{\lambda_i^*|(AA^{\dagger})^T|\lambda_j^*}\nonumber\\
    &\hspace{0.4cm}-\braket{\lambda_i|A^{\dagger}|\lambda_j}\braket{\lambda_i^*|A^{T}|\lambda_j^*}-\braket{\lambda_i|A|\lambda_j}\braket{\lambda_i^*|(A^{\dagger})^T|\lambda_j^*}\Big],\nonumber\\\label{IIIB-1}
\end{align}
where we have used $(A^T)^{\dagger}=(A^{\dagger})^T$ and $\braket{\lambda_i^*|\lambda_j^*}=\delta_{ij}$ \cite{proof-0}.\par
By using  $\braket{\lambda^*_i|X^T|\lambda^*_j}=\braket{\lambda_j|X|\lambda_i}$ \cite{proof-1} in Eq. (\ref{IIIB-1}),   we have 
\begin{align}
   \braket{\Tilde{\Phi}^s|H_A^{\dagger}H_A|\tilde{\Phi}^{1-s}}&=\frac{1}{2}\Big[\Tr(A^{\dagger}A\rho)+\Tr(AA^{\dagger}\rho)\nonumber\\
   &-\Tr(\rho^{s}A^{\dagger}\rho^{1-s}A)-\Tr(\rho^{s}A\rho^{1-s}A^{\dagger})\Big]\nonumber\\
    &=I^s_{\rho}(A).\label{IIIB-2}
\end{align}
From Eq. (\ref{IIIB-2}), it is  seen that  if $s=1/2$, then the Wigner-Yanase skew information $I_{\rho}(A)$ can be viewed as the average value of the positive semidefinite operator $H^{\dagger}_AH_A$ in the state $\ket{\Phi}=\sum_{i}\sqrt{\lambda_i}\ket{\lambda_i}\ket{\lambda^*_i}$.\par

 For a non-Hermitian operator $A_k$ defined in Eq. (\ref{IIIA2-3}), the Wigner-Yanase-Dyson skew information using Eq. (\ref{IIIA2-4})  becomes
 \begin{align}
 I^s_{\rho}(A_k)&=I^s_{\rho}(A_{k,1})+I^s_{\rho}(A_{k,2})\nonumber\\
 &= \braket{\Tilde{\Phi}^s|(H_{k,1}^2+ H_{k,2}^2)|\tilde{\Phi}^{1-s}},\label{IIIB-4}
\end{align}
where  we have used Eq. (\ref{IIIB-2}) for the Hermitian operators $A_{k,1}$ and $A_{k,2}$. Here $H_{k,1}=\frac{1}{\sqrt{2}}(A_{k,1}\otimes I-I\otimes A_{k,1}^T)$ and similarly for $H_{k,2}$ also.
 Thus Eq. (\ref{IIIA2-7}) becomes
  \begin{align}
        \sum_{k=1}^NI_{\rho}^s(A_k)= \braket{\Tilde{\Phi}^s|H_{tot}|\tilde{\Phi}^{1-s}},\label{IIIB-5}
    \end{align}
    where $H_{tot}=\sum_{k=1}^NH_{k,1}^2+H_{k,2}^2$ is a positive semidefinite operator.\par
    In the following, we provide state-independent uncertainty relations for  the sum of \{$I^s_{\rho}(A_k)\}_{k=1}^N$ using Eq. (\ref{IIIB-5}) from eigenvalue minimization method \cite{Giorda-2019}. First we study the case $s\neq \frac{1}{2}$ and then the special case $s=\frac{1}{2}$. \par
    For $s\neq \frac{1}{2}$,  the lower bound of $\sum_{k=1}^NI_{\rho}^s(A_k)$ can not be obtained by finding the minimum eigenvalue of the positive semidefinite  operator $H_{tot}$ as $\braket{\Tilde{\Phi}^s|H_{tot}|\tilde{\Phi}^{1-s}}$ doesn't represent the average value of $H_{tot}$. To find the lower bound of the positive number $\braket{\Tilde{\Phi}^s|H_{tot}|\tilde{\Phi}^{1-s}}$, we apply the reverse Cauchy-Schwarz inequality \cite{Otachel-2018}.
    \begin{proposition}\label{Proposition 1}
  For the positive number $|\braket{\Tilde{\Phi}^s|H_{tot}|\tilde{\Phi}^{1-s}}|$, the following inequalities hold
  \\
  (i)
      \begin{align}
        |\braket{\Tilde{\Phi}^s|H_{tot}|\tilde{\Phi}^{1-s}}| \geq  f(\tau_1,\tau_2)\Theta_s\sqrt{\braket{\Phi^{1-s}|H_{tot}^2|\Phi^{1-s}}}.\label{IIIB1-1}
    \end{align}
    Here $\Theta_s=\sqrt{\Tr[\rho^{2s}]\Tr[\rho^{2(1-s)}]}$, 
     \begin{align*}
    f(\tau_1,\tau_2)=\underset{\tau_1,\tau_2}{max}\{\frac{1-\tau_1\tau_2}{(1+\tau_1^2)(1+\tau_2^2)}\}\leq 1,\\
   where,\hspace{2mm} \tau_1\geq \frac{1}{|\braket{\chi|\Phi^s}|^2}-1,\hspace{2mm} \tau_2\geq \frac{1}{|\braket{\chi|\Phi^{1-s}_{H_{tot}}}|^2}-1,
    \end{align*}
    \begin{align*}
       s.\hspace{0.5mm}t,\hspace{2mm}\tau_1\tau_2<1,
    \end{align*}
and $\ket{\Phi^s}=\ket{\tilde{\Phi}^s}/\sqrt{\braket{\tilde{\Phi}^s|\tilde{\Phi}^s}}$,  $\ket{\Phi^{1-s}_{H_{tot}}}=H_{tot}\ket{\tilde{\Phi}^{1-s}}/\sqrt{\braket{\tilde{\Phi}^{1-s}|H_{tot}^2|\tilde{\Phi}^{1-s}}}$, and $\ket{\chi}\in \mathcal{H}\otimes\mathcal{H}$ is an arbitrary bipartite state.\\
 (ii)
      \begin{align}
        |\braket{\Tilde{\Phi}^s|H_{tot}|\tilde{\Phi}^{1-s}}| \geq  f(\mu_1,\mu_2)\Theta_s\sqrt{\braket{\Phi^{s}|H_{tot}^2|\Phi^{s}}}.\label{IIIB1-2}
    \end{align}
Here 
       \begin{align*}
    f(\mu_1,\mu_2)=\underset{\mu_1,\mu_2}{max}\{\frac{1-\mu_1\mu_2}{(1+\mu_1^2)(1+\mu_2^2)}\}\leq 1,\\
  where,\hspace{2mm}\mu_1\geq \frac{1}{|\braket{\xi|\Phi^{1-s}}|^2}-1,\hspace{2mm}\ \mu_2\geq \frac{1}{|\braket{\xi|\Phi^{s}_{H_{tot}}}|^2}-1,
    \end{align*}
    \begin{align*}
     s.\hspace{0.5mm}t,\hspace{2mm}\mu_1\mu_2<1,
    \end{align*}
and $\ket{\Phi^{1-s}}=\ket{\tilde{\Phi}^{1-s}}/\sqrt{\braket{\tilde{\Phi}^{1-s}|\tilde{\Phi}^{1-s}}}$,  $\ket{\Phi^{s}_{H_{tot}}}=H_{tot}\ket{\tilde{\Phi}^s}/\sqrt{\braket{\tilde{\Phi}^s|H_{tot}^2|\tilde{\Phi}^s}}$, and $\ket{\xi}\in \mathcal{H}\otimes\mathcal{H}$ is an arbitrary bipartite state.
\end{proposition}
\begin{proof}
See Appendix \ref{C} for the proofs.
\end{proof}

    \begin{theorem}\label{Theorem 4}
    Let \{$A_k\}_{k=1}^N\subset\cal{L}(\cal{H})$ and at least two Hermitian operators in  \{$A_{k,1},A_{k,2}\}_{k=1}^N$  defined in Eq. (\ref{IIIA2-7}) do not share any common eigenstate. Then the sum of \{$I^s_{\rho}(A_k)\}_{k=1}^N$ with $s\neq\frac{1}{2}$ is lower bounded by 
    \begin{align}
        \sum_{k=1}^NI^s_{\rho}(A_k)\geq max\{f(\tau_1,\tau_2), f(\mu_1,\mu_2)\}\Theta_s\epsilon_0,\label{IIIB2-1}
    \end{align}
    and if $\epsilon_0=0$, then 
   \begin{align}
        \sum_{k=1}^NI^s_{\rho}(A_k)\geq  max\{f_{1-s}(\tau_1,\tau_2), f_{s}(\mu_1,\mu_2)\}\Theta_s\epsilon_1, \label{IIIB2-2}
    \end{align}
   where $\epsilon_0$ and $\epsilon_1$ are the ground state and first excited state eigenvalues, respectively of the operator $H_{tot}$  defined in Eq. (\ref{IIIB-5}), and   
        \begin{align*}
   f_{1-s}(\tau_1,\tau_2)&=f(\tau_1,\tau_2)\left[1-\frac{1}{d}\frac{(\Tr\rho^{1-s})^2}{{\Tr\rho^{2(1-s)}}}\right]^{\frac{1}{2}},\\
    f_{s}(\mu_1,\mu_2)&=f(\mu_1,\mu_2)\left[1-\frac{1}{d}\frac{(\Tr\rho^{s})^2}{{\Tr\rho^{2s}}}\right]^{\frac{1}{2}}.
    \end{align*}
    \end{theorem}
    \begin{proof}
       There are two lower bounds of  $\sum_{k=1}^NI_{\rho}^s(A_k)=\braket{\Tilde{\Phi}^s|H_{tot}|\tilde{\Phi}^{1-s}}$  given in \emph{Proposition \ref{Proposition 1}}, namely Eqs. (\ref{IIIB1-1}) and (\ref{IIIB1-2}).  Now the minimum value of both $\braket{\Phi^{1-s}|H_{tot}^2|\Phi^{1-s}}$ and $\braket{\Phi^{s}|H_{tot}^2|\Phi^{s}}$ is  $\epsilon^2_0$, where $\epsilon_0$ is  the ground state eigenvalue  of $H_{tot}$. By taking the maximum of the lower bounds of $\braket{\Tilde{\Phi}^s|H_{tot}|\tilde{\Phi}^{1-s}}$ given in  Eqs. (\ref{IIIB1-1}) and (\ref{IIIB1-2}), we have inequality (\ref{IIIB2-1}).\par
       To prove inequality (\ref{IIIB2-2}), we have to find the minimum values of $\braket{\Phi^{1-s}|H_{tot}^2|\Phi^{1-s}}$ and $\braket{\Phi^{s}|H_{tot}^2|\Phi^{s}}$ when $\epsilon_0=0$. In Appendix \ref{D}, we show that $\braket{\Phi^{1-s}|H_{tot}^2|\Phi^{1-s}}$ and $\braket{\Phi^{s}|H_{tot}^2|\Phi^{s}}$ can have zero value if and only if both $\ket{\Phi^s}$ and $\ket{\Phi^{1-s}}$ are maximally entangled state.  Note that $\ket{\Phi^s}$ and $\ket{\Phi^{1-s}}$ being maximally entangled state, lead to $\sum_{k=1}^NI_{\rho}^s(A_k)=\braket{\Tilde{\Phi}^s|H_{tot}|\tilde{\Phi}^{1-s}}=0$ as each skew information becomes zero. Hence, we discard the possibility of zero lower bound of $\sum_{k=1}^NI_{\rho}^s(A_k)$ when $\ket{\Phi^s}$ and $\ket{\Phi^{1-s}}$ are not maximally entangled state or equivalently the density operator $\rho$ is not maximally mixed state. It is shown  in Appendix \ref{D} that  the lower bounds of  $\braket{\Phi^{1-s}|H_{tot}^2|\Phi^{1-s}}$ and $\braket{\Phi^{s}|H_{tot}^2|\Phi^{s}}$ are $\epsilon^2_1\left[1-\frac{1}{d}\frac{(\Tr\rho^{1-s})^2}{{\Tr\rho^{2(1-s)}}}\right]$ and $\epsilon^2_1\left[1-\frac{1}{d}\frac{(\Tr\rho^{s})^2}{{\Tr\rho^{2s}}}\right]$, respectively whenever $\epsilon_0=0$. Now by putting these values in Eqs. (\ref{IIIB1-1}) and (\ref{IIIB1-2}), and taking maximum of the lower bounds of $\braket{\Tilde{\Phi}^s|H_{tot}|\tilde{\Phi}^{1-s}}$ in Eqs. (\ref{IIIB1-1}) and (\ref{IIIB1-2}), we finally obtain inequality (\ref{IIIB2-2}).
    \end{proof}

\begin{theorem}\label{Theorem 5}
     Let \{$A_k\}_{k=1}^N\subset\cal{L}(\cal{H})$ and at least two Hermitian operators in  \{$A_{k,1},A_{k,2}\}_{k=1}^N$  defined in Eq. (\ref{IIIA2-7}) do not share any common eigenstate. Then the sum of \{$I_{\rho}(A_k)\}_{k=1}^N$ is lower bounded by
    \begin{align}
        \sum_{k=1}^NI_{\rho}(A_k)\geq \epsilon_0,\label{IIIB3-1}
    \end{align}
    and if $\epsilon_0=0$, then 
    \begin{align}
        \sum_{k=1}^NI_{\rho}(A_k)\geq \epsilon_1\left[1-\frac{1}{d}(\Tr\sqrt{\rho})^2\right],\label{IIIB3-2}
    \end{align}
    where $\epsilon_0$ and $\epsilon_1$ are the ground state and first excited state eigenvalues, respectively of the operator $H_{tot}$ defined in Eq. (\ref{IIIB-5}).
\end{theorem}
\begin{proof}
    From Eq. (\ref{IIIB-5}) with $s=\frac{1}{2}$, we have
    \begin{align}
        \sum_{k=1}^NI_{\rho}(A_k)=\braket{{\Phi}|H_{tot}|{\Phi}},\label{IIIB3-3}
    \end{align}
    where $\ket{{\Phi}}=\sum_{i}\sqrt{\lambda_i}\ket{\lambda_i}\ket{\lambda_i^*}$ is a normalized state vector. As $\sum_{k=1}^NI_{\rho}(A_k)$ is the average value of the operator $H_{tot}$ in the state $\ket{\Phi}$, the minimum value of  $\sum_{k=1}^NI_{\rho}(A_k)$ is the ground state eigenvalue of $H_{tot}$ \emph{i.e}., $\epsilon_0$. Hence Eq. (\ref{IIIB3-1}) is proved.\par
   In Appendix \ref{D}, we show that if $\epsilon_0=0$, then  the minimum value of $\braket{{\Phi}|H_{tot}|{\Phi}}$ is given by  $\epsilon_1\left[1-\frac{1}{d}(\Tr\sqrt{\rho})^2\right]$ and thus Eq. (\ref{IIIB3-2}) is proved.
    \end{proof}
    As the method of obtaining the lower bound of sum of skew informations is not unique,  the lower bound in  the inequality (\ref{IIIB3-1}) may not be the actual one  or even near to the actual  lower bound. The reason is that the ground state  of ${H_{tot}}$ may not be an entangled state (not maximally as in that case the ground state eigenvalue becomes zero) and thus for product states, the lower bound is not achievable. In the reference \cite{Giorda-2019}, the authors have shown that it is possible to make the lower bound tighter for the sum of variances of arbitrary number of observables. Unfortunately, the similar strategy can not be applied here to make the lower bound in Eq. (\ref{IIIB3-1}) tighter.  \par 
\textbf{States that saturate the lower bound.} Now let us discuss the saturation condition and the states which saturate the inequalities (\ref{IIIB3-1}) and (\ref{IIIB3-2}).  We assume $\epsilon_{d^2-1}\geq\epsilon_{d^2-2}\geq\cdots\geq\epsilon_{1}\geq\epsilon_{0}$ without any loss of generality, where  $\{\epsilon_i\}$ are the eigenvalues of  the positive semidefinite operator $H_{tot}$. Then from Eq. (\ref{IIIB3-3}), we have
     \begin{align}
        \sum_{k=1}^NI_{\rho}(A_k)=\epsilon_0|\braket{\Phi|\epsilon_0}|^2+\sum_{i=1}^{d^2-1}\epsilon_i|\braket{\Phi|\epsilon_i}|^2,\nonumber
    \end{align}
    where $\ket{\epsilon_i}$ is the eigenvector corresponding to the eigenvalue $\epsilon_i$. Clearly, $\sum_{k=1}^NI_{\rho}(A_k)=\epsilon_0$ only when $\ket{\Phi}=\ket{\epsilon_0}$ \emph{i.e}.,  $\rho=\Tr_B(\ketbra{\epsilon_0}{\epsilon_0})$. If $\epsilon_0=0$, then we have 
         \begin{align}
        \sum_{k=1}^NI_{\rho}(A_k)&=\sum_{i=1}^{d^2-1}\epsilon_i|\braket{\Phi|\epsilon_i}|^2,\nonumber\\
        &=\sum_{i=1}^{d^2-1}(\epsilon_1+\gamma_i)|\braket{\Phi|\epsilon_i}|^2\nonumber\\
        &=\epsilon_1(1-|\braket{\Phi|\epsilon_0}|^2)+\sum_{i=2}^{d^2-1}\gamma_i|\braket{\Phi|\epsilon_i}|^2,\label{IIIB3-4}
    \end{align}
    where $\gamma_i$ is defined by the relation $\epsilon_i=\epsilon_1+\gamma_i$ with $\gamma_1=0$, and we have used $\sum_{i=1}^{d^2-1}|\braket{\Phi|\epsilon_i}|^2=1-|\braket{\Phi|\epsilon_0}|^2$. To have $\sum_{i=2}^{d^2-1}\gamma_i|\braket{\Phi|\epsilon_i}|^2=0$ in Eq. (\ref{IIIB3-4}),  $\ket{\Phi}$ must belong to  $span\{\ket{\epsilon_0},\ket{\epsilon_1}\}$. As we already know that if $\epsilon_0=0$, then $\ket{\epsilon_0}$ is a maximally entangled state and hence $|\braket{\Phi|\epsilon_0}|^2=\frac{1}{d}(\Tr\sqrt{\rho})^2$. Then we have  $\sum_{k=1}^NI_{\rho}(A_k)=\epsilon_1\left[1-\frac{1}{d}(\Tr\sqrt{\rho})^2\right]$.\par
   \textbf{Goodness of the  lower bound.}  Now, let us  understand how good the  lower bound $\tilde{l}_B:=\epsilon_{0}$ obtained in Eq. (\ref{IIIB3-1}) is  from the actual unknown  lower bound  $l_B$ for sum of  \{$I_{\rho}(A_k)\}_{k=1}^N$. The  obtained lower bound $\tilde{l}_B$ using our method can not be guaranteed to be the actual unknown lower bound $l_B$. The reason is simply because the ground state of $H_{tot}$ may not be an entangled state. On the one hand,  if the  ground state eigenvalue  $\epsilon_0$ corresponds to an entangled eigenstate of $H_{tot}$, then $\tilde{l}_B=\epsilon_0$ is the actual lower bound $l_B$  which is achievable by just taking $\ket{\Phi}$ to be that entangled  ground state of $H_{tot}$.  On the other hand, if  the ground state eigenvalue $\epsilon_0$ corresponds to product eigenstate of $H_{tot}$, then $\tilde{l}_B=\epsilon_0$ is not  the actual unknown lower bound $l_B$ and thus there will be a gap between these two. Our task is now to find the possible range of lower bounds for sum of  \{$I_{\rho}(A_k)\}_{k=1}^N$  in which  $l_B$ lies.
   We assume that  $H_{tot}$ has spectral decomposition \{$\{\epsilon_0,\ket{\epsilon_0}\}, \{\epsilon_i,\ket{\epsilon_i}\}_{i=1}^{K}\}$, where $K=d^2-1$.  
   Now, we start from Eq. (\ref{IIIB3-3}):
 \begin{align}
 \sum_{k=1}^NI_{\rho}(A_k)&=\braket{\Phi|H_{tot}|\Phi}\nonumber\\
 &=\epsilon_0|\braket{\Phi|\epsilon_0}|^2+\sum_{i=1}^{K}\epsilon_i|\braket{\Phi|\epsilon_i}|^2\nonumber\\
 &\leq \epsilon_0|\braket{\Phi|\epsilon_0}|^2+\epsilon_K\sum_{i=1}^{K}|\braket{\Phi|\epsilon_i}|^2\nonumber\\
 &=\epsilon_0|\braket{\Phi|\epsilon_0}|^2+\epsilon_K(1-|\braket{\Phi|\epsilon_0}|^2)\nonumber\\
 &=\epsilon_K-(\epsilon_K-\epsilon_0)|\braket{\Phi|\epsilon_0}|^2,\label{IIIB3-5}
  \end{align}
 where we have used the fact that $\epsilon_i\leq\epsilon_K$ for $i=1,\cdots,K-1$ and $\ketbra{\epsilon_0}{\epsilon_0}+\sum_{i=1}^K\ketbra{\epsilon_i}{\epsilon_i}=I$.   Note that, the upper bound in Eq. (\ref{IIIB3-5}) depends on $\ket{\epsilon_0}$ which indicates that the upper bound can have multiple values depending upon the nature of $\ket{\epsilon_0}$.  Now from Eq. (\ref{IIIB3-1}) and (\ref{IIIB3-5}), we have 
   \begin{align}
 \epsilon_0\leq\sum_{k=1}^NI_{\rho}(A_k)\leq\epsilon_K-(\epsilon_K-\epsilon_0)|\braket{\Phi|\epsilon_0}|^2.\label{IIIB3-6}
  \end{align}
 Eq. (\ref{IIIB3-6}) suggests that the actual unknown lower bound $l_B$ of $\sum_{k=1}^NI_{\rho}(A_k)$ must lie in the interval $[\epsilon_0,\epsilon_K-(\epsilon_K-\epsilon_0)|\braket{\Phi|\epsilon_0}|^2]$.  Let us now study two cases in Eq. (\ref{IIIB3-6}):\\
 ($i$) If  $\ket{\epsilon_0}$ is an  entangled state, then by taking  $\ket{\Phi}=\ket{\epsilon_0}$, the interval becomes $[\epsilon_0,\epsilon_0]$ \emph{i.e}., width of the interval is zero  and we conclude that the lower bound in Eq. (\ref{IIIB3-1}) is achievable for the density operator  $\rho=\Tr_B(\ketbra{\epsilon_0}{\epsilon_0})$.\\
 ($ii$) If  $\ket{\epsilon_0}=\ket{\tilde{l}_B}\ket{\tilde{l}^*_B}$, a product state, then the minimum value  of the upper bound in Eq. (\ref{IIIB3-6}) is given by $\epsilon_K-(\epsilon_K-\epsilon_0)\lambda_d$ as $max|\braket{\Phi|\epsilon_0}|^2=\lambda_d$ \cite{Wolf-2012}, where we considered  $\ket{\Phi}=\sum_{i=1}^{d-1}\sqrt{\lambda_i}\ket{\lambda_i}\ket{\lambda_i^*}+\sqrt{\lambda_d}\ket{\lambda_d}\ket{\lambda_d^*}$ (in Schmidt decomposition form) and $\lambda_{d}$ being  the largest   Schmidt coefficient.  Thus the interval in which the actual unknown lower bound $l_B$ lies is  $[\epsilon_0,\epsilon_0\lambda_d+\epsilon_{K} (1-\lambda_d)]$.\par
   The goodness of the  lower bound $\tilde{l}_B$ now depends on the width of the interval $[\epsilon_0,\epsilon_0\lambda_d+\epsilon_{K} (1-\lambda_d)]$. The smaller the interval, the better the approximation of the lower bound $\tilde{l}_B$ . The condition for sufficiently  smaller  interval is $\epsilon_{K} (1-\lambda_d)\ll \epsilon_0\lambda_d$.
    \par
 \textbf{Example 1.}  We provide here analytical results for   spin-$\frac{1}{2}$ and spin-1 systems and  remark on the lower bound of sum of the Wigner-Yanase skew informations   in a spin-$j$ system based on numerical results at the end of this example. \par
 ($i$) \emph{Spin-$\frac{1}{2}$:} For  the observables $S_x=\frac{1}{2}\sigma_x$, $S_y=\frac{1}{2}\sigma_y$ and $S_z=\frac{1}{2}\sigma_z$,  the ground state eigenvalue of $H_{tot}$ defined in \emph{Theorem \ref{Theorem 5}} is found to be $\epsilon_0=0$. Thus we have to look for the first excited state eigenvalue which is given by $\epsilon_1=1$  and   according to  \emph{Theorem \ref{Theorem 5}},  we have
 \begin{align*}
  \sum_{k=x,y,z}I_{\rho}(S_k)\geq [1-\frac{1}{d}(\Tr\sqrt{\rho})^2],
 \end{align*}
    where $d=2$. \par
    ($ii$) \emph{Spin-$1$:}  The spin-$1$ operators are defined as
     \begin{equation*}
S_x=
 \frac{1}{\sqrt{2}}\begin{pmatrix}
 0 & 1 & 0\\
1 & 0 & 1\\
0& 1 & 0
\end{pmatrix},
\hspace{6mm}
S_y= \frac{1}{\sqrt{2}}
\begin{pmatrix}
 0 & -i & 0\\
i& 0 & -i\\
0 & i & 0
\end{pmatrix},
\end{equation*}
  \begin{equation*}
\hspace{6mm}
S_z=\begin{pmatrix}
 1 & 0 & 0\\
0 & 0 & 0\\
0& 0 & -1
\end{pmatrix}.
\end{equation*}
The ground state eigenvalue of $H_{tot}$ defined in \emph{Theorem \ref{Theorem 5}} is found to be $\epsilon_0=0$. Thus we  look for the first excited state eigenvalue which is given by $\epsilon_1=1$  and   according to \emph{Theorem \ref{Theorem 5}},  we have
 \begin{align*}
  \sum_{k=x,y,z}I_{\rho}(S_k)\geq [1-\frac{1}{d}(\Tr\sqrt{\rho})^2],
 \end{align*}
    where $d=3$. \par
($iii$) \emph{Systems with large spin:} For systems with large spin number \emph{e.g}.,  $j=3/2 , 2, \cdots$, the numerical evidences show that the ground state eigenvalue of $H_{tot}$ defined in \emph{Theorem \ref{Theorem 5}} for  spin operators in three orthogonal directions $x$, $y$ and $z$  with arbitrary   spin number is zero \emph{i.e}., $\epsilon_0=0$.  The first excited state eigenvalue is found to be $\epsilon_1=1$. Thus, based on numerical evidences, we conjucture that 
\begin{align*}
  \sum_{k=x,y,z}I_{\rho}(S_k)\geq [1-\frac{1}{d}(\Tr\sqrt{\rho})^2],
 \end{align*}
    where $d=2j+1$. The density operator  which  saturates the lower bound of the above inequality has already been discussed [see Eq. (\ref{IIIB3-4})]. \par
    
    \subsubsection{Based on standard deviation}
As mentioned in the reference \cite{Giorda-2019} that there is no unique way to achieve the lower bound of the sum of variances of $N$ number of Hermitian operators from eigenvalue minimization.   Below we show another  similar method using  the idea of  eigenvalue minimization to achieve the lower bound of the sum of variances of $N$ number of arbitrary  operators. This is the generalization of the work  \cite{Giorda-2019} for  $N$ number of arbitrary operators.  The result  follows from \emph{Theorem \ref{Theorem 5}}.
\begin{corollary}\label{Corollary 2}
 Let \{$A_k\}_{k=1}^N\subset\cal{L}(\cal{H})$ and at least two Hermitian operators in  \{$A_{k,1},A_{k,2}\}_{k=1}^N$  defined in Eq. (\ref{IIIA2-7}) do not share any common eigenstate. Then the sum of \{$\braket{\Delta A_k}_{\psi}^2\}_{k=1}^N$ is lower bounded by 
    \begin{align}
        \sum_{k=1}^N\braket{\Delta A_k}_{\psi}^2\geq \epsilon_0,\label{IIIB4-1}
    \end{align}
    and if $\epsilon_0=0$, then 
    \begin{align}
        \sum_{k=1}^N\braket{\Delta A_k}_{\psi}^2\geq \epsilon_1\left[1-\frac{1}{d}\right],\label{IIIB4-2}
    \end{align}
    where $\epsilon_0$ and $\epsilon_1$ are ground state and first excited state eigenvalues, respectively of the operator $H_{tot}$  defined in Eq. (\ref{IIIB-5}).
\end{corollary}
\begin{proof}
For a pure state $\rho=\ketbra{\psi}{\psi}$,  skew information defined in Eq. (\ref{IIA2-1}) becomes variance \emph{i.e}.,  $I_{\psi}(A_k)=\braket{\Delta A_k}_{\psi}^2$ and thus the proofs of Eqs. (\ref{IIIB4-1}) and (\ref{IIIB4-2}) automatically follow from the results of \emph{Theorem \ref{Theorem 5}}.
\end{proof}
Later, we will show that our  result of  \emph{Corollary \ref{Corollary 2}}  can sometimes provide tighter state-independent lower bound than the method given in  Ref. \cite{Giorda-2019}.\par
Now we provide the strategy similar to  Ref. \cite{Giorda-2019} to have tighter lower bound for $\sum_{k=1}^N\braket{\Delta A_k}_{\psi}^2$ compared to the one given in Eq. (\ref{IIIB4-1}).
 \begin{proposition}\label{Proposition 2}
 For each $k=1,2,\cdots,N$ and $n=1,2$, and  $\alpha\in[a^{(1)}_{k,n}, a^{(d)}_{k,n}]$, define $H_{tot,k,n}^{\alpha}:=H_{tot}+A_{k,n}^{\alpha}\otimes A_{k,n}^{\alpha}$, where $H_{tot}$ is defined in Eq. (\ref{IIIB-5}) and ${A^{\alpha}_{k,n}}=A_{k,n}-\alpha I$. Here,    \{$a_{k,n}^{(i)}\}^{d}_{i=1}$  are the eigenvalues of the Hermitian operator $A_{k,n}$,  and $a_{k,n}^{(1)}$ and $a_{k,n}^{(d)}$ are minimum and maximum eigenvalues of  $A_{k,n}$, respectively. Then the best lower bound for $\sum_{k^{\prime}=1}^N\braket{\Delta A_{k^{\prime}}}_{\psi}^2$ is given by 
\begin{align}
 \sum_{k^{\prime}=1}^N\braket{\Delta A_{k^{\prime}}}_{\psi}^2\geq \underset{k,n}{max}\hspace{2mm}\underset{\alpha}{min}\hspace{1mm}\epsilon_{0,k,n}^{\alpha},\label{IIIB5-1}
\end{align}
where $\epsilon_{0,k,n}^{\alpha}$ is the ground state eigenvalue of $H_{tot,k,n}^{\alpha}$.
\end{proposition}  
\begin{proof}
First we have to show that  $H_{tot,k,n}^{\alpha}\geq 0$ \emph{i.e}., positive semidefinite. Note  that  $H^{\alpha}_{k,n}=\frac{1}{\sqrt{2}}(A^{\alpha}_{k,n}\otimes I-I\otimes {A^{\alpha}_{k,n}}^T)=H_{k,n}$ and  hence 
\begin{align}
H_{tot}&=\sum_{k^{\prime}=1}^N\sum_{n^{\prime}=1}^2H^2_{k^{\prime},n^{\prime}}\nonumber\\
&=\sum_{k^{\prime}=1}^N\sum_{n^{\prime}=1}^2(H^{\alpha}_{k^{\prime},n^{\prime}})^2\nonumber\\
&=\sum_{{k^{\prime}}\neq k}^N\sum_{{n^{\prime}}\neq n}^2(H^{\alpha}_{k^{\prime},n^{\prime}})^2+\frac{1}{2}[(A^{\alpha}_{k,n})^2\otimes I+I\otimes ({A^{\alpha}_{k,n}}^T)^2]\nonumber\\
&-A^{\alpha}_{k,n}\otimes {A^{\alpha}_{k,n}}^T.\nonumber
\end{align}
Now, we have $H_{tot,k,n}^{\alpha}=H_{tot}+A^{\alpha}_{k,n}\otimes {A^{\alpha}_{k,n}}^T=\sum_{{k^{\prime}}\neq k}^N\sum_{{n^{\prime}}\neq n}^2(H^{\alpha}_{k^{\prime},n^{\prime}})^2+\frac{1}{2}[(A^{\alpha}_{k,n})^2\otimes I+I\otimes ({A^{\alpha}_{k,n}}^T)^2]\geq0$.\par
By using  $H_{tot}=H_{tot,k,n}^{\alpha}-A^{\alpha}_{k,n}\otimes {A^{\alpha}_{k,n}}^T$ and  $I_{\psi}(A_{k^{\prime}})=\braket{\Delta A_{k^{\prime}}}_{\psi}^2$  with $\ket{\Phi}=\ket{\psi}\ket{\psi^*}$ in  Eq. (\ref{IIIB3-3}),  it  becomes  
\begin{align}
        \sum_{k^{\prime}=1}^N\braket{\Delta A_{k^{\prime}}}_{\psi}^2=\bra{\psi}\braket{{\psi^*}|[H_{tot,k,n}^{\alpha}-A^{\alpha}_{k,n}\otimes {A^{\alpha}_{k,n}}^T]|{\psi}}\ket{\psi^*}.\label{IIIB5-2}
    \end{align}
To prove the minimization part of the right hand side of Eq. (\ref{IIIB5-1}), let us define $\cal{S}_{k,n}^{\alpha}:=\{\ket{\psi}\in \mathcal{H}|\braket{\psi|{A}_{k,n}|\psi}=\alpha\}$, then $\bra{\psi}\braket{\psi^*|A^{\alpha}_{k,n}\otimes {A^{\alpha}_{k,n}}^T|\psi}\ket{\psi^*}=0$ $\forall$ $\ket{\psi}\in \cal{S}_{k,n}^{\alpha}$ and
\begin{align}
 \sum_{k^{\prime}=1}^N\braket{\Delta A_{k^{\prime}}}_{\psi}^2&=\bra{\psi}\braket{\psi^*|H_{tot,k,n}^{\alpha}-A^{\alpha}_{k,n}\otimes {A^{\alpha}_{k,n}}^T|\psi}\ket{\psi^*}\nonumber\\
 &=\bra{\psi}\braket{\psi^*|H_{tot,k,n}^{\alpha}|\psi}\ket{\psi^*}\nonumber\\
 &\geq \epsilon_{0,k,n}^{\alpha}.\label{IIIB5-3}
\end{align}
 Now,  as $\alpha=\braket{\psi|{A}_{k,n}|\psi}$ and  $A_{k,n}$ is a  Hermitian operator,  $\alpha$ can take any values from $[a^{(1)}_{k,n}, a^{(d)}_{k,n}]$.  By varying the  values of $\alpha$, we obtain similar inequalities like Eq. (\ref{IIIB5-3}) and it is obvious that $\cup_{\alpha\in[a^{(1)}_{k,n}, a^{(d)}_{k,n}]}\cal{S}_{k,n}^{\alpha}\equiv\cal{H}$. Thus, if Eq. (\ref{IIIB5-3}) is to be valid for any state $\ket{\psi}$, we must take the minimum of $\epsilon_{0,k,n}^{\alpha}$ over $\alpha\in[a^{(1)}_{k,n}, a^{(d)}_{k,n}]$.\par
To prove the maximization part of the right hand side of Eq. (\ref{IIIB5-1}), note that the left hand side is independent of the index $k$ and $n$, and hence the maximization over $k$ and $n$ is the most preferable choice to have the tightest lower bound in this scenario.
\end{proof}
The above mentioned \emph{Proposition \ref{Proposition 2}} is valid without the transposition operation also and in fact without  transposition operation, \emph{Proposition \ref{Proposition 2}} is same as    the work of Ref. \cite{Giorda-2019}.  \par
\textbf{Example 2.}
 We will  now show that \emph{Corollary \ref{Corollary 2}} can provide better results for the following example compared to Ref. \cite{Giorda-2019}. Consider the operators:
  \begin{equation*}
A_1=
 \begin{pmatrix}
 0 & 1 & 0\\
1 & 0 & i\\
0& -i & 0
\end{pmatrix},
\hspace{6mm}A_2=
\begin{pmatrix}
 1 & 0 & 0\\
0 & 0 & 0\\
0& 0 & -1
\end{pmatrix},
\end{equation*}
  \begin{equation*}
A_3=
 \begin{pmatrix}
 1 & 1 & 0\\
1 & 0 & -1\\
0& -1 & -1
\end{pmatrix},
\hspace{6mm}A_4=
\begin{pmatrix}
 1 & 0 & i\\
0 & 0 & 0\\
-i& 0 & -1
\end{pmatrix}.
\end{equation*}
 We have found that the ground state eigenvalue $\epsilon_0$ of $H_{tot}$ in \emph{Corollary \ref{Corollary 2}} is zero and hence the lower bound of $V_{1234}:=\sum_{k=1}^4\braket{\Delta A_k}_{\psi}^2$  is given by $\epsilon_1(1-\frac{1}{d})=2.32339(1-\frac{1}{3})=1.5489$ [from Eq. (\ref{IIIB4-2})]. We have checked that the strategy provided in \emph{Proposition \ref{Proposition 2}} does not give higher value  of the lower bound of $V_{1234}$ than 1.5489. To compare with the work of Ref.  \cite{Giorda-2019}, the best lower bound of $V_{1234}$  is given by 1.3993 $<$ 1.5489 and hence our strategy gives tighter state-independent lower bound than  than the one given in  Ref. \cite{Giorda-2019} for  this particular example.

  \subsection{State-independent uncertainty relations for channels} \label{IIIC}
  The authors in Ref. \cite{Luo-Sun-2017} have defined a coherence measure of a state with respect to the L$\ddot{u}$ders measurement $\Pi=\{\Pi_i\}$ with $\sum_i\Pi_i=I$ involving the Wigner-Yanase skew information. Namely, they defined the coherence measure of a state $\rho$ as $\cal{C}(\rho|\Pi):=\sum_iI_{\rho}(\Pi_i)$ and showed that $\cal{C}(\rho|\Pi)$ satisfies the reasonable requirements for a coherence measure that is: ($i$) $\cal{C}(\rho|\Pi)\geq 0$ and equality holds iff the state is an incoherent state [here incoherent states are defined as $\cal{I}_L:=\{\sigma: \Pi(\sigma)=\sigma\}$, where $\Pi(\sigma)=\sum_i\Pi_i\sigma\Pi_i$], ($ii$) $\cal{C}(\rho|\Pi)$ is convex in $\rho$, and ($iii$) monotonicity under incoherent operations:  $\cal{C}(\Lambda_I(\rho)|\Pi)\leq \cal{C}(\rho|\Pi)$, where $\Lambda_I$ is an incoherent operation (which cannot create coherence from incoherent states).  Later they have generalized the coherence measure of the state with respect to an arbitrary  quantum channel \cite{Luo-Sun-2018} in the following way.   An arbitrary quantum channel can be defined with the Kraus representation  as
  \begin{align}
  \cal{E}(\rho)=\sum_{i=1}^mK^i\rho {K^i}^{\dagger},\label{IIIC1-1}
  \end{align}
    where \{$K^i\}^{m}_{i=1}$ are Kraus operators with $\sum_i{K^i}^{\dagger}K^i=I$ and $\rho$ is the input state.  Then  the Wigner-Yanase skew information associated with a channel is defined as 
 \begin{align}
 I_{\rho}(\cal{E})=\sum_{i=1}^mI_{\rho}(K^i),\label{IIIC1-2}
 \end{align}
 where $I_{\rho}(K^i)$ is defined in Eq. (\ref{IIA2-1}) for the Kraus operator $K^i$, and  $I_{\rho}(\cal{E})$ is regarded as a bona fide measure for coherence of $\rho$ with respect to the channel  $\cal{E}$ \cite{Luo-Sun-2018}.   The authors have shown that $I_{\rho}(\cal{E})$ satisfies many interesting and desirable properties of a coherence measure. \par
  Now as the measure of coherence depends on orthonormal basis or more generally a channel involving the  Wigner-Yanase skew information, then it is natural to ask if there is some trade-off relation between the coherence measures in different orthonormal bases (or with respect to channels). Singh \emph{et al}. \cite{Singh-Pati-2016} derived such trade-off relation for coherence measures in different orthonormal bases and Fu \emph{et al}. \cite{Fu-Luo-2019} derived state-dependent trade-off relation for coherence measures with respect to two channels based on the Wigner-Yanase skew information.  In the following theorem, we provide state-independent uncertainty relation for coherence measures with respect to a quantum collection of channels based on the Wigner-Yanase skew information defined in Eq. (\ref{IIIC1-2}).\par
 For a collection of  channels \{$\cal{E}_l\}_{l=1}^{N}$, the sum of \{$I_{\rho}(\cal{E}_l)\}^N_{l=1}$ can be rewritten as 
 \begin{align}
 \sum_{l=1}^NI_{\rho}(\cal{E}_l)=\sum_{l=1}^N\sum_{i=1}^{m_l}I_{\rho}(K_{l,1}^i)+I_{\rho}(K_{l,2}^i),\label{IIIC1-2-1}
 \end{align}
 where we have used Eq. (\ref{IIIC1-2}) for each channel $\cal{E}_l$ and taken sum over $l$, and Eq. (\ref{IIIA2-4}) for the operators $K^i_{l}$ which can alway be written as  $K^i_l=K^i_{l,1}\pm i K^i_{l,2}$ with  $K^i_{l,1}$ and $K^i_{l,2}$ being Hermitian operators according to Eq. (\ref{IIIA2-3}).
 \begin{theorem}\label{Theorem 6}
 Let \{$\cal{E}_l\}_{l=1}^{N}$ be a collection  of quantum channels and each channel $\cal{E}_l$ is represented by the Kraus operators \{$K^i_l\}^{m_l}_{i=1}$, and at least two Hermitian operators in  \{\{$K^i_{l,1},K^i_{l,2}\}_{i=1}^{m_l}\}_{l=1}^N$  defined in Eq. (\ref{IIIC1-2-1}) do not share any common eigenstate. Then the sum of \{$I_{\rho}(\cal{E}_l)\}^N_{l=1}$ is lower bounded by 
 \begin{align}
 \sum_{l=1}^NI_{\rho}(\cal{E}_l)\geq \epsilon_0,\label{IIIC1-3}
 \end{align}
 and if $\epsilon_0=0$, then 
 \begin{align}
 \sum_{l=1}^NI_{\rho}(\cal{E}_l)\geq \epsilon_1\left[1-\frac{1}{d}(\Tr\sqrt{\rho})^2\right],\label{IIIC1-4}
  \end{align}
  where $\epsilon_0$ and $\epsilon_1$ are ground state and first excited state eigenvalues, respectively of the operator $H_{tot}=\sum_{l=1}^NF_l$, $F_l=\sum_{i=1}^{m_l}\sum_{n=1}^2(H_{l,n}^i)^2$ and $H^i_{l,n}=\frac{1}{\sqrt{2}}(K_{l,n}^i\otimes I-I\otimes {K_{l,n}^i}^T)$.
 \end{theorem}
 \begin{proof}
 Note that  Eq. (\ref{IIIC1-2-1})  can be rewritten as  
  \begin{align}
\sum_{l=1}^NI_{\rho}(\cal{E}_l)&=\sum_{l=1}^N\sum_{i=1}^{m_l}I_{\rho}(K_{l,1}^i)+I_{\rho}(K_{l,2}^i)\nonumber\\
 &=\sum_{l=1}^N\sum_{i=1}^{m_l}\braket{\Phi|(H_{l,1}^i)^2+(H_{l,2}^i)^2|\Phi}\nonumber\\
 &=\sum_{l=1}^N\braket{\Phi|F_l|\Phi},\label{IIIC1-5}
 \end{align}
 where we have used Eq. (\ref{IIIB-2}) for the Hermitian operators $K_{l,1}^i$ and $K_{l,2}^i$, and $\ket{\tilde{\Phi}^s}=\ket{\tilde{\Phi}^{1-s}}=\ket{\Phi}$ for $s=\frac{1}{2}$.  Here $\ket{{\Phi}}=\sum_{i}\sqrt{\lambda_i}\ket{\lambda_i}\ket{\lambda_i^*}$ is a normalized state vector, and $F_l=\sum_{i=1}^{m_l}\sum_{n=1}^2(H_{l,n}^i)^2$.  Thus Eq. (\ref{IIIC1-5}) becomes
 \begin{align}
\sum_{l=1}^NI_{\rho}(\cal{E}_l)=\braket{\Phi|H_{tot}|\Phi},\label{IIIC1-6}
 \end{align}
 where $H_{tot}=\sum_{l=1}^NF_l$.  Now, as $H_{tot}$ is a positive semidefinite operator, the minimum value of the average of $H_{tot}$ \emph{i.e}., $\braket{\Phi|H_{tot}|\Phi}$ is the ground state eigenvalue $\epsilon_0$ of $H_{tot}$ and hence inequality (\ref{IIIC1-3}) is proved. To prove inequality (\ref{IIIC1-4}), one can use the same arguments as given for the inequality (\ref{IIIB3-2}) in  \emph{Theorem \ref{Theorem 5}} which is that, if $\epsilon_0=0$ then the minimum value of $\braket{\Phi|H_{tot}|\Phi}$ is given by $\epsilon_1\left[1-\frac{1}{d}(\Tr\sqrt{\rho})^2\right]$.
  \end{proof}
{The implication of inequality (\ref{IIIC1-3}) is that  how the coherence of a given quantum state with respect to a channel  is restricted by the coherence of the same quantum state with respect to the other channels. If, for example, we consider two channels, then even if the coherence of the density operator with respect to the first channel is zero, the  coherence of the same density operator with respect to the second channel  is never zero.}\par
  If the initial state is a pure state, then the bound given in Eq. (\ref{IIIC1-3}) can be improved using the same strategy given in \emph{Proposition \ref{Proposition 2}} as skew information becomes variance for a pure state.\par
  \textbf{Example 3.} We consider the phase damping channel $\cal{E}_1(\rho)=\sum_{i=1}^{2}K^i_1\rho {K^i_1}^{\dagger}$: 
  \begin{align*}
  K_1^1=\ketbra{0}{0}+\sqrt{1-p}\ketbra{1}{1}, \hspace{0.6cm} K_1^2=\sqrt{p}\ketbra{1}{1},
  \end{align*}
  and the amplitute damping channel $\cal{E}_2(\rho)=\sum_{i=1}^{2}K^i_2\rho {K^i_2}^{\dagger}$: 
  \begin{align*}
  K_2^1=\ketbra{0}{0}+\sqrt{1-p}\ketbra{1}{1}, \hspace{0.6cm} K_2^2=\sqrt{p}\ketbra{0}{1}.
  \end{align*}
  Now according to \emph{Theorem \ref{Theorem 6}}, the ground state eigenvalue of $H_{tot}$ is given by $\epsilon_0=0$ and hence, we have to find the first excited state eigenvalue of $H_{tot}$ which is given by $\epsilon_1=p$. Thus the sum of the Wigner-Yanase skew informations of two channels $\cal{E}_1$ and $\cal{E}_2$ is lower bounded by:
  \begin{align*}
 I_{\rho}(\cal{E}_1)+I_{\rho}(\cal{E}_2)\geq p\left[1-\frac{1}{2}(\Tr\sqrt{\rho})^2\right],
  \end{align*}   
where we have taken $d=2$ for the qubit.
\subsection{Uncertainty relations  based on generalized skew information}  \label{IIID}
\subsubsection{Generalized skew information}
Recently, Yang \emph{et al.} \cite{Yang-2022}  have shown that there exist a family of skew information like quantities (and they call it as generalized  skew information) in which the well known  Wigner-Yanase  skew information and quantum Fisher information are special cases.  Here we show that for such generalized skew informations of  noncommuting operators, state-dependent and state-independent uncertainty relations exist. We  define the generalized skew information for arbitrary operators  by modifying the definition of generalized skew information  for normal operators given in  Ref. \cite{Yang-2022}.  Along with,  in the special case of the  generalized skew information defined in  Ref. \cite{Yang-2022}, it doesn't become the Wigner-Yanase skew information for arbitrary operator defined in Eq. (\ref{IIA2-1}). To solve both the issues, we define the generalized skew information of an arbitrary operator as
\begin{align}
\cal{I}^{\nu}_{\rho}(A)&=\frac{1}{2}\Tr[\rho(A^{\dagger}A+AA^{\dagger})]\nonumber\\
&-\frac{1}{2}\sum_{i,j}m_{\nu}(\lambda_i,\lambda_j)[|\braket{\lambda_i|A^{\dagger}|\lambda_j}|^2+|\braket{\lambda_i|A|\lambda_j}|^2],\label{IIID1-1}
\end{align}
where $m_{\nu}(\lambda_i,\lambda_j)$ is the generalized mean of  positive real numbers $\lambda_i$ and $\lambda_j$ of equal weighting having the following form
\begin{align}
m_{\nu}(\lambda_i,\lambda_j)=\left(\frac{\lambda_i^{\nu}+\lambda_j^{\nu}}{2}\right)^{1/{\nu}},\label{IIID1-2}
\end{align}
where $-\infty<\nu<0$ and
\begin{align*}
m_{0}(\lambda_i,\lambda_j):&=\underset{\nu\rightarrow 0}{\lim}m_{\nu}(\lambda_i,\lambda_j)=\sqrt{\lambda_i\lambda_j},\\
m_{-\infty}(\lambda_i,\lambda_j):&=\underset{\nu\rightarrow -\infty}{\lim}m_{\nu}(\lambda_i,\lambda_j)=\min\{\lambda_i,\lambda_j\}.
\end{align*} 
One can easily verify that 
\begin{align*}
\cal{I}^{0}_{\rho}(A)&=\Tr(\rho A^2)-\Tr(\sqrt{\rho}A\sqrt{\rho}A),\\
\cal{I}^{-1}_{\rho}(A)&=\frac{1}{2}\sum_{i,j}\frac{(\lambda_i-\lambda_j)^2}{\lambda_i+\lambda_j}|\braket{\lambda_i|A|\lambda_j}|^2
\end{align*}
for a Hermitian operator.  Thus we achieve  the well known informational quantities like the Wigner-Yanase skew information and the  Fisher information as special cases of generalized skew information and these are given by
\begin{align}
\cal{I}^{0}_{\rho}(A)&=I_{\rho}(A),\label{IIID1-3}\\
\cal{I}^{-1}_{\rho}(A)&=\frac{1}{4}F_{\rho}(A),\label{IIID1-4}
\end{align}
where $F_{\rho}(A)$ is the Fisher information defined as $4\cal{I}^{-1}_{\rho}(A)$ \cite{Braunstein-1994,Toth-2014}. It can be shown that $\cal{I}^{\nu}_{\rho}(A)$ defined in Eq. (\ref{IIID1-1}) has the same properties similar to generalized skew information  defined in \cite{Yang-2022}:
\begin{align}
\cal{I}^{0}_{\rho}(A)\leq\cal{I}^{-1}_{\rho}(A)\leq\cdots\leq\cal{I}^{-\infty}_{\rho}(A)\leq V_{\rho}(A).\label{IIID1-5}
\end{align}
Moreover, generalized skew information has\\
($i$) Convexity property:
\begin{align}
\cal{I}^{\nu}_{\rho}(A)\leq\sum_ip_i\cal{I}^{\nu}_{\rho_i}(A),\label{IIID1-6}
\end{align}
where $\rho=\sum_ip_i\rho_i$.\\
($ii$) Additivity property:
\begin{align}
\cal{I}^{\nu}_{\rho_A\otimes\rho_B}(A\otimes I+I\otimes B)=\cal{I}^{\nu}_{\rho_A}(A)+\cal{I}^{\nu}_{\rho_B}(B).\label{IIID1-7}
\end{align}
Having such important properties of generalized skew information, it can be made useful similar to the Wigner-Yanase skew information and the  Fisher information in future.
\subsubsection{Uncertainty relations}
We have already established uncertainty relation for the Wigner-Yanase skew information  describing the impossibility of preparing a density operator such that both the Wigner-Yanase skew informations of incompatible observables are zero in state-dependent and state-independent way. Now, as generalized skew information is lower bounded by the Wigner-Yanase skew information [see Eq. (\ref{IIID1-5})], we can have state-dependent uncertainty relation for generalized skew informations of incompatible observables which is given by
   \begin{align}
       \sqrt{\cal{I}^{\nu}_{\rho}(A)\cal{I}^{\nu^{\prime}}_{\rho}(B)}\geq \frac{\pm \frac{i}{4}[\Tr([A^{\dagger},B]+[A,B^{\dagger}])\rho^s+\mathcal{E}_{AB}^s]}{1+\Omega_{AB}^s-\frac{1}{4}\Tr[(\xi^{\pm}+\eta^{\mp})(I-\rho^{1-s})]}.\label{IIID1-8}
    \end{align}
The derivation is obvious from Eq. (\ref{IIIA2-1}). One can show that the following equality holds for generalized skew informations:
\begin{align}
\cal{I}^{\nu}_{\rho}(A)=\cal{I}^{\nu}_{\rho}(A_1)+\cal{I}^{\nu}_{\rho}(A_2),\label{IIID1-9}
\end{align}
where $A_1$ and $A_2$  are Hermitian operators and  defined in Eq. (\ref{IIIA2-3}). Thus, the generalized skew information of a non-Hermitian operator can be expressed as the sum of the generalized skew informations of two  Hermitian operators. This result will be useful to derive state-independent sum uncertainty relations based on generalized skew information of a collection of non-Hermitian operators.\par
 The state-independent uncertainty relation for $N$ number of  arbitrary operators  is given by
  \begin{align}
        \sum_{k=1}^N\cal{I}^{\nu_k}_{\rho}(A_k)\geq \epsilon_0,\label{IIID1-10}
    \end{align}
    and if $\epsilon_0=0$, then 
    \begin{align}
        \sum_{k=1}^N\cal{I}^{\nu_k}_{\rho}(A_k)\geq \epsilon_1\left[1-\frac{1}{d}(\Tr\sqrt{\rho})^2\right].\label{IIID1-11}
    \end{align}
The inequalities (\ref{IIID1-10}) and (\ref{IIID1-11}) automatically hold due to inequality (\ref{IIID1-5}) and \emph{Theorem \ref{Theorem 5}}. \par

\subsubsection{Uncertainty relations in qubits}
Although in an arbitrary dimensional system, it is not possible in general  to interrelate  different type of generalized skew informations, for example, relation between the Wigner-Yanase skew information ($\nu=0$) and  the Fisher information ($\nu=-1$),   or  more generally, relation between $\cal{I}^{\nu}_{\rho}(A)$ and $\cal{I}^{\nu^{\prime}}_{\rho}(A)$, here we show that in a qubit, it is possible to interrelate   different type of generalized skew informations. Moreover,  uncertainty relations for  different form of generalized  skew informations of incompatible operators, tight uncertainty relations and uncertainty equalities are provided in detail.\par
In a qubit, let the density operator be $\rho=\lambda_1\ketbra{\lambda_1}{\lambda_1}+\lambda_2\ketbra{\lambda_2}{\lambda_2}$, where $\lambda_1$ and $\lambda_2$ are the non-zero eigenvalues with $\lambda_1+\lambda_2=1$. Now the generalized skew information of an arbitrary operator $\sigma$ in a qubit can be expressed as 
\begin{align}
\cal{I}^{\nu}_{\rho}(\sigma)=\left[1-2\left(\frac{1}{2}\Tr\rho^{\nu}\right)^{{1}/{\nu}}\right]\braket{\Delta \sigma}_{\lambda_l}^2,\label{IIID2-1}
\end{align} 
where $l=1$ or $2$. See Appendix \ref{E} for the derivation.  Note that $(\frac{1}{2}\Tr\rho^{\nu})^{{1}/{\nu}}=m_{\nu}(\lambda_1,\lambda_2)$  defined in Eq. (\ref{IIID1-2}) for a qubit. Thus, in a qubit, any two skew informations $\cal{I}^{\nu}_{\rho}(\sigma)$ and $\cal{I}^{\nu^{\prime}}_{\rho}(\sigma)$ of the same operator $\sigma$ are connected by the equality:
\begin{align}
\frac{\cal{I}^{\nu}_{\rho}(\sigma)}{[1-2(\Tr\rho^{\nu}/2)^{{1}/{\nu}}]}=\frac{\cal{I}^{\nu^{\prime}}_{\rho}(\sigma)}{[1-2(\Tr\rho^{\nu^{\prime}}/2)^{{1}/{\nu^{\prime}}}]}.\label{IIID2-2}
\end{align}
For example, the Wigner-Yanase skew information $\cal{I}^{0}_{\rho}(\sigma)=I_{\rho}(\sigma)$ and the Fisher information $\cal{I}^{-1}_{\rho}(A)=\frac{1}{4}F_{\rho}(A)$ are related by $F_{\rho}(\sigma)/I_{\rho}(\sigma)=4[1-4\lambda_1\lambda_2]/[1-2\sqrt{\lambda_1\lambda_2}]$.\par
Now we state the most general form of  uncertainty relation  for different form of generalized skew informations of incompatible operators in a qubit.
\begin{proposition}\label{Proposition 3}
For a collection of  arbitrary  operators \{$\sigma_k\}^{N}_{k=1}$  acting on a two-dimensional Hilbert space,  the following inequality  holds
\begin{align}
\sum_{{k}=1}^N\frac{\cal{I}^{\nu_{k}}_{\rho}(\sigma_{k})}{[1-2\left(\frac{1}{2}\Tr\rho^{\nu_{k}}\right)^{\frac{1}{\nu_{k}}}]}\geq \cal{L}_{N},\label{IIID2-3}
\end{align}
 where  $\cal{L}_N$ is the state-independent  lower bound of  $\sum_{k=1}^N\braket{\Delta \sigma_k}_{\psi}^2$ and $\ket{\psi}$ is any pure state. 
\end{proposition}
 \begin{proof}
 In Eq. (\ref{IIID2-1}), we assume that for each operator, the index $\nu$ is fixed. By taking sum of Eq. (\ref{IIID2-1}) for \{$\sigma_k\}^{N}_{k=1}$, we have 
\begin{align}
\sum_{{k}=1}^N\frac{\cal{I}^{\nu_{k}}_{\rho}(\sigma_{k})}{[1-2\left(\frac{1}{2}\Tr\rho^{\nu_{k}}\right)^{\frac{1}{\nu_{k}}}]}=\sum_{k=1}^N\braket{\Delta \sigma_k}_{\lambda_l}^2.\label{IIID2-4}
\end{align}
Now, if $\cal{L}_N$ is the state-independent lower bound of $\sum_{k=1}^N\braket{\Delta \sigma_k}_{\lambda_l}^2$, then $\cal{L}_N$ is also the lower bound of  $\sum_{k=1}^N\braket{\Delta \sigma_k}_{\psi}^2$, and  Eq. (\ref{IIID2-3}) holds automatically.
 \end{proof}
 Eq. (\ref{IIID2-3}) provides a tighter state-independent uncertainty relation than Eq. (\ref{IIID1-10}) [or Eq. (\ref{IIID1-11}) if $\epsilon_0=0$]. We give an example in the following.\par
\textbf{\emph{Observation 1.}}  For   $\nu_1=\nu_2=\cdots=\nu_N=\nu$, the inequality (\ref{IIID2-3}) becomes
\begin{align}
\sum_{{k}=1}^N\cal{I}^{\nu}_{\rho}(\sigma_{k})\geq \left[1-2\left(\frac{1}{2}\Tr\rho^{\nu}\right)^{\frac{1}{\nu}}\right]\cal{L}_N.\label{IIID2-5}
\end{align}
As $\sum_{{k}=1}^N\cal{I}^{\nu}_{\rho}(\sigma_{k})$ has two lower bounds given in   Eq. (\ref{IIID1-10}) [or Eq. (\ref{IIID1-11}) if $\epsilon_0=0$] and Eq. (\ref{IIID2-5}), we have to take the largest one. For a qubit, we show significant improvement of the lower bound of $\sum_{{k}=1}^N\cal{I}^{\nu}_{\rho}(\sigma_{k})$ given in Eq. (\ref{IIID2-5}) over  the one given in Eq. (\ref{IIID1-10}) in the following example. \par
\textbf{Example 4.} 
Consider the qubit observables: 
  \begin{equation*}
\sigma_1=
 \begin{pmatrix}
 1 & 1-\frac{i}{2} \\
1+\frac{i}{2} & -1
\end{pmatrix},
\hspace{6mm}\sigma_2=
\begin{pmatrix}
 1 &\frac{1}{2}+\frac{i}{2} \\
\frac{1}{2}-\frac{i}{2}&-1 
\end{pmatrix},
\end{equation*}
  \begin{equation*}
\sigma_3=
\frac{1}{2} \begin{pmatrix}
 1 & -1-i \\
-1+i &-1 
\end{pmatrix},
\hspace{6mm}\sigma_4=
\begin{pmatrix}
 -1 & \frac{1}{2}+\frac{i}{2} \\
\frac{1}{2}-\frac{i}{2} & 1 
\end{pmatrix}.
\end{equation*}
{The density operator is considered to be diagonal in computational basis with diagonal elements $0.3$ and $0.7$, respectively.} In this case, $\epsilon_0=0$ [see Eq. (\ref{IIID1-10})] and hence  the lower bound  of $\cal{I}^{\nu}_{1234}:=\sum_{{k}=1}^4\cal{I}^{\nu}_{\rho}(\sigma_{k})$ is given by  $0.1921$ in Eq. (\ref{IIID1-11}) and this remains the same lower bound for any value of $\nu$. Now the lower bound of  $\cal{I}^{\nu}_{1234}$  in Eq. (\ref{IIID2-5}) is found to be  0.1834 for $\nu=0$; 0.3515 for $\nu=-1$;  0.4835 for $\nu=-2$; $\cdots$; 0.8788 for $\nu=-\infty$. Thus a significant improvement of the lower bound of  $\cal{I}^{\nu}_{1234}$ for $\nu=-1,-2,\cdots, -\infty$ can be observed using Eq. (\ref{IIID2-5}). \par
\begin{proposition}\label{Proposition 4}
The sum of the collection of  generalized skew informations \{$\cal{I}_{\rho}^{\nu_k}(\sigma_k)\}^{N_1}_{k=1}$ and  standard deviations \{$\braket{\Delta \Omega_{k^{\prime}}}_{\rho}^{2}\}^{N_2}_{k^{\prime}=1}$ is lower bounded by  
\begin{align}
\sum_{{k}=1}^{N_1}\frac{\cal{I}^{\nu_{k}}_{\rho}(\sigma_{k})}{[1-2\left(\frac{1}{2}\Tr\rho^{\nu_{k}}\right)^{\frac{1}{\nu_{k}}}]}+\sum_{k^{\prime}=1}^{N_2} \braket{\Delta \Omega_{k^{\prime}}}_{\rho}^2\geq \cal{L}_{N_1+N_2},\label{IIID2-6}
\end{align}
where $\cal{L}_{N_1+N_2}$ is the state-independent  lower bound of  $\sum_{k=1}^{N_1}\braket{\Delta \sigma_k}_{\psi}^2+\sum_{k^{\prime}=1}^{N_2}\braket{\Delta \Omega_{k^{\prime}}}_{\psi}^2$ and $\ket{\psi}$ is any pure state. 
\end{proposition}
\begin{proof}
For a qubit, we have $\braket{\Delta \Omega_{k^{\prime}}}_{\lambda_1}=\braket{\Delta \Omega_{k^{\prime}}}_{\lambda_2}$, where $\ket{\lambda_l}$ with $l=1$ or $2$ are eigenstates of the density operator $\rho$ (see Appendix \ref{E}), and hence $\braket{\Delta \Omega_{k^{\prime}}}_{\rho}^2\geq \braket{\Delta \Omega_{k^{\prime}}}_{\lambda_l}^2$ and
\begin{align}
\sum_{k^{\prime}=1}^{N_2}\braket{\Delta \Omega_{k^{\prime}}}_{\rho}^2\geq \sum_{k^{\prime}=1}^{N_2}\braket{\Delta \Omega_{k^{\prime}}}_{\lambda_l}^2.\label{IIID2-7}
\end{align}
Now, together with Eq. (\ref{IIID2-4}) and Eq. (\ref{IIID2-7}),  we have
\begin{align*}
&\sum_{{k}=1}^{N_1}\frac{\cal{I}^{\nu_{k}}_{\rho}(\sigma_{k})}{[1-2\left(\frac{1}{2}\Tr\rho^{\nu_{k}}\right)^{\frac{1}{\nu_{k}}}]}+\sum_{k^{\prime}=1}^{N_2} \braket{\Delta \Omega_{k^{\prime}}}_{\rho}^2\nonumber\\
&=\sum_{k=1}^{N_1}\braket{\Delta \sigma_k}_{\lambda_l}^2+\sum_{k^{\prime}=1}^{N_2}\braket{\Delta \Omega_{k^{\prime}}}_{\lambda_l}^2\geq\cal{L}_{N_1+N_2},
\end{align*}
where $\cal{L}_{N_1+N_2}$ is the state-independent lower bound of $\sum_{k=1}^{N_1}\braket{\Delta \sigma_k}_{\lambda_l}^2+\sum_{k^{\prime}=1}^{N_2}\braket{\Delta \Omega_{k^{\prime}}}_{\lambda_l}^2$.  As  $\cal{L}_{N_1+N_2}$  is also the state-independent lower bound of  $\sum_{k=1}^{N_1}\braket{\Delta \sigma_k}_{\psi}^2+\sum_{k^{\prime}=1}^{N_2}\braket{\Delta \Omega_{k^{\prime}}}_{\psi}^2$,  this completes the proof.
\end{proof}
\textbf{\emph{Observation 2.}} If $N_1=1$, $N_2=1$  and $\nu_1=-1$, then Eq. (\ref{IIID2-6}) becomes $\frac{1}{4(1-4\lambda_1\lambda_2)}{{F}_{\rho}(\sigma_{\vec{a}})}+\braket{\Delta \sigma_{\vec{b}}}^2_{\rho}\geq \frac{1}{4}(1-|\vec{a}\cdot\vec{b}|)$.  Using the unique property of Fisher information being the convex roof of variances given in  Ref. \cite{Chiew-Gessner-2022}, this inequality can be improved as
\begin{align}
 \frac{1}{4}F_{\rho}(\sigma_{\vec{a}})+\braket{\Delta \sigma_{\vec{b}}}^2_{\rho}\geq\frac{1}{4}(1-|\vec{a}\cdot\vec{b}|),\label{IIID2-8}
 \end{align}
 where $\vec{a}$ and $\vec{b}$ are two arbitrary directions for spin operators. See Ref. \cite{Sixia-2013} for the general derivation of how the Fisher information can be expressed as  the convex roof of variances first proved in Ref. \cite{Toth-Petz-2013} for lower dimensions  and conjectured for higher dimensions. The  inequality (\ref{IIID2-8}) holds only for two  observables in two arbitrary directions $\vec{a}$ and $\vec{b}$, and the result can not be extended for multiple observables in arbitrary directions whereas   inequality (\ref{IIID2-6}) holds for  multiple observables in arbitrary directions.  \par
 Now, as any generalized skew information can be expressed in terms of Fisher information [see Eq. (\ref{IIID2-2})], Eq. (\ref{IIID2-8}) can be expressed in terms of generalized skew information which will be an improvement of Eq. (\ref{IIID2-6}) for only two observables:
 \begin{align}
\frac{(1-4\lambda_1\lambda_2)}{[1-2(\Tr\rho^{\nu_k}/2)^{{1}/{\nu_k}}]}\cal{I}^{\nu_{k}}_{\rho}(\sigma_{\vec{a}})+\braket{\Delta \sigma_{\vec{b}}}^2_{\rho}\geq\frac{1}{4}(1-|\vec{a}\cdot\vec{b}|).\label{IIID2-9}
 \end{align}
 {\emph{Remark.}} As Fisher information can be  expressed as  a convex roof of variances, it is possible to express generalized skew information as a convex roof of variances  in a qubit   due to a direct connection between generalized skew information and the Fisher information  [see Eq. (\ref{IIID2-2})]: 
 \begin{align*}
 \cal{I}^{\nu_{k}}_{\rho}(\sigma)=\frac{[1-2(\Tr\rho^{\nu_k}/2)^{{1}/{\nu_k}}]}{(1-4\lambda_1\lambda_2)}4 \underset{p_k, \psi_k}{\emph{inf}}\sum_kp_k\braket{\Delta \sigma}^2_{\psi_k}.
 \end{align*}
\textbf{Uncertainty equalities.} The spin operator along $\vec{n}\in\mathbb{R}^3$ can be written as $\sigma_{\vec{n}}=\frac{1}{2}\vec{n}\cdot\vec{\sigma}$, where $\vec{\sigma}=(\sigma_x,\sigma_y,\sigma_z)$ is a vector of Pauli operators. Any density operator can be characterized in a qubit  by a Bloch vector $\vec{r}\in\mathbb{R}^3$ via $\rho=\frac{1}{2}(I+\vec{r}\cdot\vec{\sigma})$.\par
Note that the standard deviation of $\sigma_{\vec{n}}$ in the given state $\rho$ is 
\begin{align}
\braket{\Delta \sigma_{\vec{n}}}^2_{\rho}=\frac{1}{4}[1-(\vec{n}\cdot\vec{r})^2],\label{IIID3-1}
\end{align}
where we have used $\braket{\sigma_{\vec{n}}^2}_{\rho}=\frac{1}{4}$ and  $\braket{\sigma_{\vec{n}}}_{\rho}=\frac{1}{2}\vec{n}\cdot\vec{r}$ with $|\vec{n}|^2=1$. The purity of $\rho$ is given by 
\begin{align}
\Tr(\rho^2)&=\frac{1}{2}(1+|\vec{r}|^2),\nonumber\\
&=\frac{1}{2}[1+(\vec{n}_1\cdot\vec{r})^2+(\vec{n}_2\cdot\vec{r})^2+(\vec{n}_3\cdot\vec{r})^2],\label{IIID3-2}
\end{align}
where $(\vec{n}_1,\vec{n}_2,\vec{n}_3)$ are three orthogonal directions in $\mathbb{R}^3$. By combining Eqs. (\ref{IIID3-1}) and (\ref{IIID3-2}), we have
\begin{align}
\braket{\Delta \sigma_{\vec{n}_1}}^2_{\rho}+\braket{\Delta \sigma_{\vec{n}_2}}^2_{\rho}+\braket{\Delta \sigma_{\vec{n}_3}}^2_{\rho}=1-\frac{1}{2}\Tr\rho^2.\label{IIID3-3}
\end{align}
\begin{proposition}\label{Proposition 5}
The following equality for generalized skew informations holds:
\begin{align}
\sum_{{k}=1}^3\frac{\cal{I}^{\nu_{k}}_{\rho}(\sigma_{\vec{n}_k})}{[1-2\left(\frac{1}{2}\Tr\rho^{\nu_{k}}\right)^{\frac{1}{\nu_{k}}}]}=\frac{1}{2}.\label{IIID3-4}
\end{align}
\end{proposition}
\begin{proof}
Using Eq. (\ref{IIID3-3}) for pure state in right hand side of  Eq. (\ref{IIID2-4}) for three  spin-1/2 operators in three orthogonal directions $\vec{n}_1$, $\vec{n}_2$ and $\vec{n}_3$, the relation (\ref{IIID3-4}) is proved.
\end{proof}
For $\nu_1=\nu_2=\nu_3=-1$ [which corresponds to Fisher information, see Eq. (\ref{IIID1-4})], the above equation becomes same as of  equation (36) of Ref. \cite{Chiew-Gessner-2022}.\par 
\begin{proposition}\label{Proposition 6}
The following equalities involving  generalized skew informations and  standard deviations hold:
\begin{align}
\Gamma_{\nu} \cal{I}^{\nu}_{\rho}(\sigma_{\vec{n}_1})+\braket{\Delta \sigma_{\vec{n}_2}}^2_{\rho}+\braket{\Delta \sigma_{\vec{n}_3}}^2_{\rho}&=\frac{1}{2},\label{IIID3-5}\\
\Gamma_{\nu} \cal{I}^{\nu}_{\rho}(\sigma_{\vec{n}_1})+\Gamma_{\nu^{\prime}}\cal{I}^{\nu^{\prime}}_{\rho}(\sigma_{\vec{n}_2})+\braket{\Delta \sigma_{\vec{n}_3}}^2_{\rho}&=\frac{1}{2}\Tr\rho^2,\label{IIID3-6}
\end{align}
where $\Gamma_{\nu}=\frac{(1-4\lambda_1\lambda_2)}{[1-2(\Tr\rho^{\nu}/2)^{{1}/{\nu}}]}$.
\end{proposition}
\begin{proof}
In a qubit, it can be shown that  the following relation holds \cite{Toth-2018}: 
\begin{align}
\braket{\Delta \sigma}^2_{\rho}=\frac{1}{4}F_{\rho}(\sigma)+\frac{1}{2}(1-\Tr\rho^2).\label{IIID3-7}
\end{align}
By using Eq. (\ref{IIID2-2}) with $\nu_{k^{\prime}}=-1$ in Eq. (\ref{IIID3-7}), it becomes
\begin{align}
\braket{\Delta \sigma}^2_{\rho}= \Gamma_{\nu}\cal{I}^{\nu}_{\rho}(\sigma)+\frac{1}{2}(1-\Tr\rho^2).\label{IIID3-8}
\end{align}
It is now straightforward to prove Eqs. (\ref{IIID3-5}) and (\ref{IIID3-6}) by using Eq. (\ref{IIID3-8}) for operators $\sigma_{\vec{n}_1}$ and $\sigma_{\vec{n}_2}$ in Eq. (\ref{IIID3-3}).
\end{proof}
\subsubsection{Usefulness of generalized skew informations}
 Generalized skew information includes  both the well known informational quantities the Wigner-Yanase skew information ($\nu=0$) and the Fisher information ($\nu=-1$). The other values of $\nu$ in Eq. (\ref{IIID1-1}) may represent some other type of information which is currently unknown. But the convexity property of generalized skew information can be useful for detection of entanglement \cite{Horodecki2009} (or quantum steering \cite{Wiseman-2007}) in multipartite systems. In particular, as Fisher information $\cal{I}^{-1}_{\rho}(A)$ is upper bounder by   $\cal{I}^{-\infty}_{\rho}(A)$ and both of them are upper bounded by variance $V_{\rho}(A)$, it will be more efficient to use $\cal{I}^{-2}_{\rho}(A)$ or $\cal{I}^{-3}_{\rho}(A)$ or $\cdots$ or  $\cal{I}^{-\infty}_{\rho}(A)$ over $\cal{I}^{-1}_{\rho}(A)$ as the gap between $V_{\rho}(A)$ and   $\cal{I}^{-2}_{\rho}(A)$ or $V_{\rho}(A)$ and   $\cal{I}^{-3}_{\rho}(A)$ or $\cdots$ or $V_{\rho}(A)$ and   $\cal{I}^{-\infty}_{\rho}(A)$ is smaller than  the gap between $V_{\rho}(A)$ and   $\cal{I}^{-1}_{\rho}(A)$. Note that Fisher information has been used previously for entanglement detection of bipartite/multipartite systems \cite{Luca-2009,Toth-2012,Hyllus-2012}. 
 
 \subsection{{The Wigner-Yanase-Dyson skew information via weak value}}\label{IIIE}
 The Wigner-Yanase-Dyson skew information can not be determined in experiment directly by preparing the state and performing  measurement of the concerned observable using the conventional method. Here we provide a scheme based on weak values  which can be implemented in experiment to extract the Wigner-Yanase-Dyson skew information of an unknown  observable. In that case we only need the information of the spectral decomposition of the density operator and the information of the concerned observable is not required at all.  The experimental determination of the Wigner-Yanase-Dyson skew information will be useful in a situation where, for example, one is required to know whether the quantum state contains the information  (or more precisely skew information) with respect to an additive unknown conserved quantity represented by an observable or positive operator-valued measures in the quantum system.  A non-zero value of the Wigner-Yanase-Dyson skew information  obtained  in experiment will  imply that there is a information content in the system's state.  \par
  Compared to the expression of the Wigner-Yanase-Dyson skew information given in Eq. (\ref{IIA2-2}), it has now much simpler expression which is given in Eq. (\ref{IIIB-2}). As Hermitian operators are experimentally measurable, we consider  the Wigner-Yanase-Dyson skew information  $I^s_{\rho}(A)$ for a Hermitian operator $A$ which is unknown. Thus, Eq. (\ref{IIIB-2}) for the Hermitian operator $A$ becomes
\begin{align}
I^s_{\rho}(A)&=\braket{\Tilde{\Phi}^s|H^2_A|\tilde{\Phi}^{1-s}}\nonumber\\
&=\cal{N}_{s}\cal{N}_{1-s}\sum_{i,j=1}^d\braket{{\Phi}^s|H_A|\alpha_i\alpha_j^*}\braket{\alpha_i\alpha_j^*|H_A|{\Phi}^{1-s}}\nonumber\\
&=\cal{N}_{s}\cal{N}_{1-s}\sum_{i,j=1}^d\frac{\braket{{\Phi}^s|H_A|\alpha_i\alpha_j^*}}{\braket{{\Phi}^s|\alpha_i\alpha_j^*}}\frac{\braket{\alpha_i\alpha_j^*|H_A|{\Phi}^{1-s}}}{\braket{\alpha_i\alpha_j^*|{\Phi}^{1-s}}}\cross\nonumber\\
&\hspace{3cm}\braket{{\Phi}^{s}|\alpha_i\alpha_j^*}\braket{\alpha_i\alpha_j^*|{\Phi}^{1-s}},\label{IIIE-1}
\end{align}
where we have used   $\sum_{i,j=1}^d\ketbra{\alpha_i\alpha_j^*}{\alpha_i\alpha_j^*}=I\otimes I$ and \{$\ket{\alpha_i}\}_{i=1}^d$ and \{$\ket{\alpha_i^*}\}_{i=1}^d$ are basis in the first and second Hilbert spaces, respectively,  and  $\ket{\tilde{\Phi}^{1-s}}=\cal{N}_{1-s}\ket{{\Phi}^{1-s}}$, and   $\ket{\tilde{\Phi}^{s}}=\cal{N}_{s}\ket{{\Phi}^{s}}$, where  $\cal{N}_{1-s}$ and $\cal{N}_{s}$ are normalization factors.  Now, we denote  $\frac{\braket{\alpha_i\alpha_j^*|H_A|{\Phi}^{s}}}{\braket{\alpha_i\alpha_j^*|{\Phi}^{s}}}=\braket{H_A}_{{\Phi}^{s}}^{\alpha_i\alpha_j^*}$ as the weak value of the Hermitian operator $H_A$ with the pre- and postselections $\ket{{\Phi}^{s}}$ and $\ket{\alpha_i\alpha_j^*}$, respectively.  Similarly  $\frac{\braket{\alpha_i\alpha_j^*|H_A|{\Phi}^{1-s}}}{\braket{\alpha_i\alpha_j^*|{\Phi}^{1-s}}}=\braket{H_A}_{{\Phi}^{1-s}}^{\alpha_i\alpha_j^*}$ is the weak value of $H_A$ with the pre- and postselections $\ket{{\Phi}^{1-s}}$ and $\ket{\alpha_i\alpha_j^*}$, respectively.  Note that weak values are complex and   experimentally accessible. See Ref. \cite{sahil-sohail-sibasish,AAV-1988,Kofman-Nori,Dressel-2014} to understand how weak values can be measured in experiments in a single quantum system. In this case, as the weak value of an observable  on the joint Hilbert space $\cal{H}\otimes\cal{H}$ is involved,  we provide a detail description  in Appendix \ref{F} on how to obtain it. Finally, by denoting  $\cal{N}_{s}\braket{\alpha_i\alpha_j^*|{\Phi}^{s}}=\tilde{\cal{N}}_{s}^{ij}$   and $\cal{N}_{1-s}\braket{\alpha_i\alpha_j^*|{\Phi}^{1-s}}=\tilde{\cal{N}}_{1-s}^{ij}$ in Eq. (\ref{IIIE-1}), we have
\begin{align}
I^s_{\rho}(A)&=\sum_{i,j=1}^d[{\tilde{\cal{N}}_{s}^{ij}}]^*\tilde{\cal{N}}_{1-s}^{ij}\left[\braket{H_A}_{{\Phi}^{s}}^{\alpha_i\alpha_j^*}\right]^*\braket{H_A}_{{\Phi}^{1-s}}^{\alpha_i\alpha_j^*}.\label{IIIE-2}
\end{align}
As  the Wigner-Yanase-Dyson skew information of the Hermitian operator  $A$ is a real number, we must have the sum of all the imaginary values  (which appear due to the complex nature of the weak values) in the right hand side of Eq. (\ref{IIIE-2}) zero. Thus the correct expression for $I^s_{\rho}(A)$ is 
\begin{align}
I^s_{\rho}(A)&=\cal{R}e\left(\sum_{i,j=1}^d[{\tilde{\cal{N}}_{s}^{ij}}]^*\tilde{\cal{N}}_{1-s}^{ij}\left[\braket{H_A}_{{\Phi}^{s}}^{\alpha_i\alpha_j^*}\right]^*\braket{H_A}_{{\Phi}^{1-s}}^{\alpha_i\alpha_j^*}\right),\label{IIIE-3}
\end{align}
 and in addition, we have 
 \begin{align*}
 \cal{I}m\left(\sum_{i,j=1}^d[{\tilde{\cal{N}}_{s}^{ij}}]^*\tilde{\cal{N}}_{1-s}^{ij}\left[\braket{H_A}_{{\Phi}^{s}}^{\alpha_i\alpha_j^*}\right]^*\braket{H_A}_{{\Phi}^{1-s}}^{\alpha_i\alpha_j^*}\right)=0.
 \end{align*}
 Thus if the weak values of the operator $H_A$ for the pre-selections $\ket{\Phi^s}$ and $\ket{\Phi^{1-s}}$ are determined experimentally when the postselections are $\ket{\alpha_i\alpha_j^*}$ for $i,j=1,2,\cdots,d$, then  the Wigner-Yanase-Dyson skew information of the  Hermitian operator  $A$  is recovered immediately. Note that $\tilde{\cal{N}}_{s}^{ij}$ and $\tilde{\cal{N}}_{1-s}^{ij}$ are known from the information given about the density operator $\rho$ and the basis \{$\ket{\alpha_i}\}_{i=1}^d$.
\section{conclusion}\label{IV}
   We have derived several   state-dependent and state-independent uncertainty relations  based on skew information and standard deviation. We have  derived the sum and product  uncertainty equalities (\ref{IIIA1-1}) and (\ref{IIIA1-2}) based on standard deviation  when the state of the system is an arbitrary density operator, a generalization of the work of Ref. \cite{Maccone-Pati,Yao-2015} where only pure states were considered.  The consequences of the sum and product uncertainty equalities are: ($i$) they can be used to obtain a series of uncertainty inequalities with hierarchical structure, ($ii$) to solve the triviality issue of the RHUR in the sense that even if the average value of the commutation of incompatible observables is zero, the lower bound of the uncertainty relation doesn't become zero.  Product uncertainty equality can be used to derive stronger quantum speed limits for states and observables.   We have derived another product uncertainty equality (\ref{IIIA1-3}) which remains meaningful even if the average value of the commutator is zero while the product uncertainty equality (\ref{IIIA1-2}) becomes useless as in that condition one of the sides of the equality becomes zero. We have also derived sum and product uncertainty equalities for multiple observables (\emph{e.g}., three) and shown that our  uncertainty equalities do not contain covariances of the operators. Having covariances of observables in  uncertainty inequalities  for multiple observables is a difficult task to analyse  them and the lower bounds of those inequalities contain negative terms of product of covariances and thus further simplification of the lower bounds can not be done by deleting the terms of product of covariances (see Appendix \ref{A}). \par
 We have derived  uncertainty equality (\ref{IIIA2-1}) for the Wigner-Yanse-Dyson skew informations of two incompatible operators which contains commutator of those operators. We have shown the consequences of that uncertainty equality  and these are in the following. ($i$) They are useful to derive a series of uncertainty inequalities with hierarchical structure, ($ii$)  We have shown using uncertainty equality that the  commutator term remains explicit and non-zero in the lower  bound   of uncertainty relations based on the Wigner-Yanase skew information and standard deviation   even if the average value of the commutator is zero (see \emph{Corollary \ref{Corollary 1}}), and ($iii$) we have used  the commutator in the uncertainty equality  to derive quantum speed limit for observable (state) [see Eq. (\ref{IIIA2-1-2})]. We have  shown that  skew information (variance) of an arbitrary operator can be expressed as the sum of  skew informations (variances) of two Hermitian operators. This  particular equality  has been used to derive state-independent uncertainty relations based on skew informations  of $N$ number of arbitrary operators.  \par
  We have derived state-independent uncertainty relations for sum of  skew informations of a collection of arbitrary operators. {As the skew information can be used as a resource in various contexts, a non-zero value of it must be guaranteed irrespective of the system’s density operators. A state-independent uncertainty relation based on skew information provides that guarantee in the sense that even if the skew information of an operator is zero, the skew information of another incompatible operator is never zero and the reason is solely due to the noncommutative nature of those operators}. We considered a system of spin-$j$ particles and have shown analytically for spin-$1/2$, $1$, (and spin-$3/2$, 2 numerically) and finally conjectured for large spin number that the lower bound of the  state-independent uncertainty relation for the Wigner-Yanase skew informations of spin operators in three orthogonal directions doesn't depend on  the spin numbers. The condition on the states  which saturates   the lower bound of the state-independent uncertainty relation based on the Wigner-Yanase skew information has been  derived. We have also provided a strategy to check the goodness of the lower bound in our uncertainty relation. For pure states, the results become state-independent uncertainty relations for   sum of  standard deviations  of the collection of those arbitrary operators and this can be considered as a generalization of the work of \cite{Giorda-2019} where only Hermitian operators were considered. In the work \cite{Giorda-2019}, the authors considered that all the operators whose state-independent  sum uncertainty relation is considered must not share any common eigenstate. We have relaxed the condition and shown that at least two operators must not share any common eigenstate. We provided  an example  where we have shown that our state-independent uncertainty relation based on  standard deviation  gives tighter state-independent uncertainty relation  than the work of \cite{Giorda-2019}. The  Wigner-Yanase skew information of a quantum channel has recently been used as a measure of quantum coherence of a density operator with respect to that channel \cite{Luo-Sun-2017}. We have shown that if  the coherence of a density operator  is measured with respect to a collection of different channels, then there exists a state-independent uncertainty  relation for  coherence measures of the density operator  with respect to  that collection of different channels (see Sec. \ref{IIIC}). The implication of the uncertainty relation  is that  how the coherence of the density operator with respect to a channel  is restricted by the coherence of the same density operator with respect to the other channels. \par
We have derived  state-dependent and state-independent  uncertainty relations for   generalized skew informations, a more general version of skew information whose special cases are the  Wigner-Yanase (-Dyson) skew information and the Fisher information \cite{Yang-2022}. Although in an arbitrary dimensional system, it is not possible to interrelate  different type of  generalized skew informations, for example, relation between the Wigner-Yanase skew information ($\nu=0$) and the Fisher information ($\nu=-1$),   or  more generally, relation between $\cal{I}^{\nu}_{\rho}(A)$ and $\cal{I}^{\nu^{\prime}}_{\rho}(A)$, we have shown that it is possible to interrelate   different type of generalized skew informations  in a qubit (see Sec. \ref{IIID}). Moreover,  uncertainty relations for  different form of generalized  skew informations and standard deviations of incompatible operators, tight uncertainty relations, and state-independent uncertainty equalities involving generalized skew informations and standard deviations of spin operators along three orthogonal directions have been derived and discussed in detail. As skew information and standard deviation are important in information theory, the uncertainty relations based on skew information and standard deviation will play fundamental roles in revealing how the values of skew informations and standard deviations of incompatible observables limit each others.\par
Finally, we have provided a scheme based on weak values which can be implemented in experiment to determine the Wigner-Yanase-Dyson skew information of an unknown observable (see Sec. \ref{IIIE}).  The experimental determination of the Wigner-Yanase-Dyson skew information is  useful in a situation where, for example, one is required to know whether the quantum state contains the information  (or more precisely skew information) with respect to an additive unknown conserved quantity represented by an observable or positive operator-valued measures in the quantum system.  A non-zero value of the Wigner-Yanase-Dyson skew information  obtained  in experiment will  imply that there is  information content in the system's state. \par
  As non-Hermitian  operators (\emph{e.g}.,  unitary operators) are becoming very useful in recent years in the context of uncertainty relations \cite{Massar-Spindel,Shrobona-Pati,Bong-2018,Zhao-2024}, our results (which  hold for arbitrary operators) can be used as uncertainty relations for unitary  operators as well.  We have   compared our result of state-independent uncertainty relation based on standard deviations of $N$ number of non-Hermitian operators with the work of Ref. \cite{Zhao-2024} in the context of  entanglement detection in Appendix \ref{G} by showing that their method is  limited for the Kraus operators of quantum channels only whereas our method works for any number of arbitrary operators.\par
 A detail investigation of  speed limits for observables and states [see Eq. (\ref{IIIA2-1-2})] using uncertainty equalities  based on  skew information and standard deviation are left for  future work.
 \section*{Acknowledgments} 
The author acknowledges  Prof. Sibasish Ghosh and Sohail for  fruitful discussions.

\appendix
\onecolumngrid
\section{}\label{A}
\subsection{Proof of Theorem 1: Eq. (\ref{IIIA1-1})}
By using the definition of the standard deviation for  non-Hermitian operators $A$ and $B$ [see Eq. (\ref{IIA1-1})], we have
\begin{align}
\braket{\Delta A}_{\rho}^2+\braket{\Delta B}_{\rho}^2&=\frac{1}{2}\Tr[(A^{\dagger}A+AA^{\dagger})\rho]+\frac{1}{2}\Tr[(B^{\dagger}B+BB^{\dagger})\rho]-(|Tr(A\rho)|^2+|Tr(B\rho)|^2)\nonumber\\
&=\pm \frac{i}{2} \Tr([A^{\dagger},B]\rho)\pm \frac{i}{2} \Tr([A,B^{\dagger}]\rho)+\frac{1}{2}\Tr[(A\mp iB)^{\dagger}(A\mp iB)\rho]+\frac{1}{2}\Tr[(A\pm iB)(A\pm iB)^{\dagger}\rho]\nonumber\\
&\hspace{3.8mm}-\frac{1}{2}|\Tr[(A\mp iB)\rho]|^2-\frac{1}{2}|\Tr[(A\pm iB)\rho]|^2\nonumber\\
&=\pm \frac{i}{2} \Tr([A^{\dagger},B]\rho)\pm \frac{i}{2} \Tr([A,B^{\dagger}]\rho)+\frac{1}{2}\Tr[\left(A\mp iB-\braket{A\mp iB}_{\rho}I\right)^{\dagger}\left(A\mp iB-\braket{A\mp iB}_{\rho}I\right)\rho]\nonumber\\
&\hspace{3.8mm} +\frac{1}{2}\Tr[\left(A\pm iB-\braket{A\pm iB}_{\rho}I\right)\left(A\pm iB-\braket{A\pm iB}_{\rho}I\right)^{\dagger}\rho]\nonumber\\
&=\pm \frac{1}{2}\braket{i([A^{\dagger},B]+[A,B^{\dagger}])}_{\rho}+\frac{1}{2}\braket{(M_{\mp}^{\dagger}M_{\mp}+N_{\pm}N_{\pm}^{\dagger})}_{\rho},\label{A1}
\end{align}
where $M_{\mp}=A\mp iB-\braket{A\mp iB}_{\rho}I$ and $N_{\pm}=A\pm iB-\braket{A\pm iB}_{\rho}I$, and thus Eq. (\ref{IIIA1-1}) is proved.
\subsection{Proof of Theorem 1: Eq. (\ref{IIIA1-2})}
By substituting  $A\rightarrow\frac{A}{\braket{\Delta A}_{\rho}}$ and $B\rightarrow\frac{B}{\braket{\Delta B}_{\rho}}$ in Eq. (\ref{A1}), we have 
\begin{align}
\frac{\braket{\Delta A}^2_{\rho}}{\braket{\Delta A}^2_{\rho}}+\frac{\braket{\Delta B}^2_{\rho}}{\braket{\Delta B}^2_{\rho}}=\pm \frac{1}{2\braket{\Delta A}_{\rho}\braket{\Delta B}_{\rho}}\braket{i([A^{\dagger},B]+[A,B^{\dagger}])}_{\rho}+\frac{1}{2}\braket{(R_{\mp}^{\dagger}R_{\mp}+S_{\pm}S_{\pm}^{\dagger})}_{\rho},\label{A2}
\end{align}
where we have used the fact that $M_{\mp}\rightarrow R_{\mp}$ and $N_{\mp}\rightarrow S_{\pm}$ under the above transformation. Here $R_{\mp}$ and $S_{\pm}$ are defined  in Eq. (\ref{IIIA1-2}). After manipulation of Eq. (\ref{A2}), we arrive at Eq. (\ref{IIIA1-2}).
\subsection{Uncertainty equalities for three observables}
Here we derive sum and product uncertainty equalities for three Hermitian operators. The consequences of the uncertainty equalities  are that a series of inequalities with hierarchical structure can be obtained involving three Hermitian operators. Moreover, the lower bounds of the sum and product uncertainty equalities do not contain the covariances of those Hermitian operators. One can derive uncertainty equalities  for any number of arbitrary operators, but we restrict ourselves for three Hermitian operators.     \par
Let $X_1$, $X_2$ and $X_3$ be Hermitian operators and we denote $X_{ii}=\braket{\Delta X_i}_{\rho}$, $Y_{ij}=\frac{1}{2}\braket{i[X_i,X_j]}_{\rho}$ and $\braket{X_i}=\braket{X_i}_{\rho}$. Then it can be shown using Eq. (\ref{A1}) that the following sum uncertainty equality holds:
 \begin{align}
       X_{11}^2+X_{22}^2+X_{33}^2&=\Big(r_{12}Y_{12}+r_{23}Y_{23}+r_{31}Y_{31}\Big)+\frac{1}{2}\braket{(M_{12}^{\dagger}M_{12}+M_{23}^{\dagger}M_{23}+M_{31}^{\dagger}M_{31})},\label{A3}
  \end{align}
        where $M_{12}=X_1-\braket{X_1}I+i r_{12}(X_2-\braket{X_2}I)$, $M_{23}=X_2-\braket{X_2}I+i r_{23}(X_3-\braket{X_3}I)$, and $M_{31}=X_3-\braket{X_3}I+i r_{31}(X_1-\braket{X_1}I)$. Here, $r_{12}=\pm$, $r_{23}=\pm$ and $r_{31}=\pm$ are chosen such that each of the term within the first bracket in the right hand side of Eq. (\ref{A3}) is positive. \par
To derive product uncertainty equality,  make the  substitutions   $A_1\rightarrow A_1X_{22}X_{33}$, $A_2\rightarrow A_2X_{33}X_{11}$ and $A_3\rightarrow A_3X_{11}X_{22}$ in Eq. (\ref{A3}). Then we have
\begin{align}
 3(X_{11}X_{22}X_{33})^2&=X_{11}X_{22}X_{33}\Big(r_{12}Y_{12}X_{33}+r_{23}Y_{23}X_{11}+r_{31}Y_{31}X_{22}\Big)+\frac{1}{2}\braket{({M^{\prime}_{12}}^{\dagger}{M_{12}^{\prime}}+{M^{\prime}_{23}}^{\dagger}{M_{23}^{\prime}}+{M^{\prime}_{31}}^{\dagger}{M_{31}^{\prime}})},\label{A4}
  \end{align}
where
\begin{align}
M_{12}^{\prime}&=X_{22}X_{33}(X_1-\braket{X_1}I)+i r_{12}X_{33}X_{11}(X_2-\braket{X_2}I),\nonumber\\
M_{23}^{\prime}&=X_{33}X_{11}(X_2-\braket{X_2}I)+i r_{23}X_{11}X_{22}(X_3-\braket{X_3}I),\nonumber\\
M_{31}^{\prime}&=X_{11}X_{22}(X_3-\braket{X_3}I)+i r_{31}X_{22}X_{33}(X_1-\braket{X_1}I).\nonumber
\end{align}
Now  divide both the sides of Eq. (\ref{A4}) by $X_{11}X_{22}X_{33}$, and use $\frac{{M_{12}^{\prime}}^{\dagger}M_{12}^{\prime}}{X_{11}X_{22}X_{33}}=\tilde{M}_{12}^{\dagger}\tilde{M}_{12}$, $\frac{{M_{23}^{\prime}}^{\dagger}M_{23}^{\prime}}{X_{11}X_{22}X_{33}}=\tilde{M}_{23}^{\dagger}\tilde{M}_{23}$ and $\frac{{M_{31}^{\prime}}^{\dagger}M_{31}^{\prime}}{X_{11}X_{22}X_{33}}=\tilde{M}_{31}^{\dagger}\tilde{M}_{31}$, where
\begin{align}
\tilde{M}_{12}&=\sqrt{\frac{X_{22}X_{33}}{X_{11}}}\Big(X_1-\braket{X_1}I\Big)+i r_{12}\sqrt{\frac{X_{33}X_{11}}{X_{22}}}\Big(X_2-\braket{X_2}I\Big),\nonumber\\
\tilde{M}_{23}&=\sqrt{\frac{X_{33}X_{11}}{X_{22}}}\Big(X_2-\braket{X_2}I\Big)+i r_{12}\sqrt{\frac{X_{11}X_{22}}{X_{33}}}\Big(X_3-\braket{X_3}I\Big),\nonumber\\
\tilde{M}_{31}&=\sqrt{\frac{X_{11}X_{22}}{X_{33}}}\Big(X_3-\braket{X_3}I\Big)+i r_{12}\sqrt{\frac{X_{22}X_{33}}{X_{11}}}\Big(X_1-\braket{X_1}I\Big)\nonumber
\end{align}
in Eq. (\ref{A4}), then it becomes
 \begin{align}
 X_{11}X_{22}X_{33}&=\frac{1}{3}\Big(r_{12}Y_{12}X_{33}+r_{23}Y_{23}X_{11}+r_{31}Y_{31}X_{22}\Big)+\frac{1}{6}\braket{({\tilde{M}_{12}}^{\dagger}{\tilde{M}_{12}}+{\tilde{M}_{23}}^{\dagger}{\tilde{M}_{23}}+{\tilde{M}_{31}}^{\dagger}{\tilde{M}_{31}})}.\label{A5}
  \end{align}
  Both the sum and product uncertainty equalities (\ref{A3}) and (\ref{A5}) contain all the possible commutators of the three observables. Both the equalities do not contain the covariances and   remain meaningful \emph{i.e}., non-trivial even if all the average values of the commutators are zero.  \par
  To compare with the work Ref. \cite{Dodonov-2018}, the author derived several uncertainty inequalities for three and four observables (Hermitian operators), and showed that those inequalities do not  contain any covariances of different observables. It was shown that for more than two observables, product uncertainty relations can contain a large number of covariances and to analyse those inequalities with such large number of  covariances  can be very difficult \cite{Dodonov-2018, Qin-2016}. In addition,  the lower bounds of those inequalities contain negative  terms of product of covariances  and thus further simplification of the lower bounds can not be done by deleting  the  terms of product of covariances. See Eq. (7) of Ref. \cite{Dodonov-2018} for a detail description. The uncertainty relations derived in Ref. \cite{Dodonov-2018} are in inequality form and can not be made  arbitrarily  tight. In contrast, our relations (\ref{A4}) and (\ref{A5})  do not contain any covariances of different observables, and a series of uncertainty inequalities with hierarchical structure can be obtained.

\section{}\label{B}
\textbf{Proof of Theorem \ref{Theorem 2}:}\par
The skew information of an arbitrary operator $A$ is defined in Eq. (\ref{IIA2-2}) as 
\begin{align}
    I^s_{\rho}(A)=\frac{1}{2}Tr([\rho^s,A]^{\dagger}[\rho^{1-s},A])=\frac{1}{2}\Tr[(A^{\dagger}A+AA^{\dagger})\rho]-\frac{1}{2}\Tr(\rho^{1-s}A^{\dagger}\rho^sA)-\frac{1}{2}\Tr(\rho^sA^{\dagger}\rho^{1-s}A).\label{B1}
\end{align}
Now consider two positive operators 
\begin{align}
\xi^{\pm}&=\left(\frac{A}{\sqrt{I^s_{\rho}(A)}}\pm i\frac{B}{\sqrt{I^s_{\rho}(B)}}\right)^{\dagger}\rho^s\left(\frac{A}{\sqrt{I^s_{\rho}(A)}}\pm i\frac{B}{\sqrt{I^s_{\rho}(B)}}\right),\label{B2}\\
\eta^{\mp}&=\left(\frac{A}{\sqrt{I^s_{\rho}(A)}}\mp i\frac{B}{\sqrt{I^s_{\rho}(B)}}\right)\rho^s\left(\frac{A}{\sqrt{I^s_{\rho}(A)}}\mp i\frac{B}{\sqrt{I^s_{\rho}(B)}}\right)^{\dagger},\label{B3}
\end{align}
and calculate 
\begin{align}
\frac{1}{2}\Tr(\xi^{\pm})+\frac{1}{2}\Tr(\eta^{\mp})=\frac{\Tr[\frac{1}{2}(A^{\dagger}A+AA^{\dagger})\rho^s]}{I_{\rho}^s(A)}+\frac{\Tr[\frac{1}{2}(B^{\dagger}B+BB^{\dagger})\rho^s]}{I_{\rho}^s(B)}\mp \frac{i}{2}\frac{\Tr([A^{\dagger},B]\rho^s)}{\sqrt{I_{\rho}^s(A)I_{\rho}^s(B)}}\mp \frac{i}{2}\frac{\Tr([A,B^{\dagger}]\rho^s)}{\sqrt{I_{\rho}^s(A)I_{\rho}^s(B)}}.\label{B4}
\end{align}
By substituting $\rho^s=\rho+\sigma$ [where $\sigma$ is a positive operator defined as $\sigma=\rho^s-\rho$] in the first two terms in the right hand side of Eq. (\ref{B4}), it becomes 
\begin{align}
\frac{1}{2}\Tr(\xi^{\pm})+\frac{1}{2}\Tr(\eta^{\mp})=&\frac{\Tr[\frac{1}{2}(A^{\dagger}A+AA^{\dagger})\rho]}{I_{\rho}^s(A)}+\frac{\Tr[\frac{1}{2}(B^{\dagger}B+BB^{\dagger})\rho]}{I_{\rho}^s(B)}+\frac{\Tr(\frac{1}{2}\{A^{\dagger},A\}\sigma)}{I_{\rho}^s(A)}+\frac{\Tr(\frac{1}{2}\{B^{\dagger},B\}\sigma)}{I_{\rho}^s(B)}\nonumber\\
&\mp \frac{i}{2}\frac{\Tr([A^{\dagger},B]\rho^s)}{\sqrt{I_{\rho}^s(A)I_{\rho}^s(B)}}\mp \frac{i}{2}\frac{\Tr([A,B^{\dagger}]\rho^s)}{\sqrt{I_{\rho}^s(A)I_{\rho}^s(B)}}.\label{B5}
\end{align}
Also calculate 
\begin{align}
\frac{1}{2}\Tr(\xi^{\pm}\rho^{1-s})+\frac{1}{2}\Tr(\eta^{\mp}\rho^{1-s})=&\frac{1}{2}\frac{\Tr(\rho^{1-s}A^{\dagger}\rho^sA)}{I_{\rho}^s(A)}\pm\frac{i}{2}\frac{\Tr(\rho^{1-s}A^{\dagger}\rho^sB)}{\sqrt{I_{\rho}^s(A)I^s_{\rho}(B)}}\mp\frac{i}{2}\frac{\Tr(\rho^{1-s}B^{\dagger}\rho^sA)}{\sqrt{I_{\rho}^s(A)I^s_{\rho}(B)}}+\frac{1}{2}\frac{\Tr(\rho^{1-s}B^{\dagger}\rho^sB)}{I_{\rho}^s(A)} \nonumber\\
&+\frac{1}{2}\frac{\Tr(\rho^{1-s}A\rho^sA^{\dagger})}{I_{\rho}^s(A)}\pm\frac{i}{2}\frac{\Tr(\rho^{1-s}A\rho^sB^{\dagger})}{\sqrt{I_{\rho}^s(A)I^s_{\rho}(B)}}\mp\frac{i}{2}\frac{\Tr(\rho^{1-s}B\rho^sA^{\dagger})}{\sqrt{I_{\rho}^s(A)I^s_{\rho}(B)}}+\frac{1}{2}\frac{\Tr(\rho^{1-s}B\rho^sB^{\dagger})}{I_{\rho}^s(A)}.\label{B6}
\end{align}
By taking the difference between Eqs. (\ref{B5}) and (\ref{B6}), we have the following equality:
\begin{align}
\frac{1}{2}\Tr[(\xi^{\pm}+\eta^{\mp})(I-\rho^{1-s})]=2+2\Omega_{AB}^s -\frac{\pm\frac{i}{2}\left[\Tr([A^{\dagger},B]\rho^s)+\Tr([A,B^{\dagger}]\rho^s)+\cal{E}_{AB}^s\right]}{\sqrt{I_{\rho}^s(A)I^s_{\rho}(B)}},\label{B7}
\end{align}
where $\cal{E}_{AB}^s$ and $\Omega_{AB}^s$ are defined in Eq. (\ref{IIIA2-1}). After manipulation, we obtain the uncertainty equality (\ref{IIIA2-1}).\par

\twocolumngrid

\section{}\label{C}
\textbf{The reverse Cauchy-Schwarz inequality:} We follow the same procedure given in  Ref. \cite{Otachel-2018} to prove the \emph{Proposition \ref{Proposition 1}}.\par
($i$) Given the unnormalized state $\ket{\tilde{\Phi}^s}$ and an arbitrary normalized state $\ket{\chi} \in \cal{H}\otimes \cal{H}$ such that $\braket{\tilde{\Phi}^s|\chi}\neq 0$, the  state orthogonal to $\ket{\chi}$ is obtained by the Gram-Schmidt orthogonalization procedure as 
\begin{align}
\ket{\tilde{\chi}^{\perp}_s}=\ket{\tilde{\Phi}^s}-\frac{\braket{\chi|\tilde{\Phi}^s}}{\braket{\chi|\chi}}\ket{\chi},\label{C1}
\end{align}
where $\ket{\tilde{\chi}^{\perp}_s}$ is an unnormalized state. Eq. (\ref{C1}) can be rewritten as 
\begin{align}
c_1\ket{\tilde{\Phi}^s}=\ket{\chi}+\ket{\tilde{\chi}^{\perp}_1},\label{C2}
\end{align}
where $c_1=\frac{1}{\braket{\chi|\tilde{\Phi}^s}}$ and $\ket{\tilde{\chi}^{\perp}_1}=c_1\ket{\tilde{\chi}^{\perp}_s}$. In the same way, given the unnormalized state $H_{tot}\ket{\tilde{\Phi}^{1-s}}=\ket{\tilde{\Phi}^{1-s}_{H_{tot}}}$ and  $\ket{\chi}$ such that $\braket{\tilde{\Phi}^{1-s}_{H_{tot}}|\chi}\neq 0$, the  state orthogonal to $\ket{\chi}$ is obtained by the Gram-Schmidt orthogonalization procedure as 
\begin{align}
\ket{\tilde{\chi}^{\perp}_{1-s}}=\ket{\tilde{\Phi}_{H_{tot}}^{1-s}}-\frac{\braket{\chi|\tilde{\Phi}_{H_{tot}}^{1-s}}}{\braket{\chi|\chi}}\ket{\chi},\label{C1-1}
\end{align}
where $\ket{\tilde{\chi}^{\perp}_{1-s}}$ is an unnormalized state and  Eq. (\ref{C1-1}) can be rewritten as 
\begin{align}
c_2\ket{\tilde{\Phi}^{1-s}_{H_{tot}}}=\ket{\chi}+\ket{\tilde{\chi}^{\perp}_2},\label{C3}
\end{align}
where $c_2=\frac{1}{\braket{\chi|\tilde{\Phi}_{H_{tot}}^{1-s}}}$ and $\ket{\tilde{\chi}^{\perp}_2}=c_2\ket{\tilde{\chi}^{\perp}_{1-s}}$. Note that $\braket{\chi|\tilde{\chi}^{\perp}_i}=0$ for $i=1,2$. \par
Now we consider 
\begin{align}
 ||\ket{\tilde{\chi}^{\perp}_i}||\leq \tau_i||\ket{\chi}||,\label{C4}
 \end{align}
 where $\tau_i\geq 0$  for $i=1,2$. By using Eq. (\ref{C4}) in the Cauchy-Schwarz inequality $Re\braket{\tilde{\chi}^{\perp}_1|\tilde{\chi}^{\perp}_2}\leq |\braket{\tilde{\chi}^{\perp}_1|\tilde{\chi}^{\perp}_2}|\leq ||\ket{\tilde{\chi}^{\perp}_1}||||\ket{\tilde{\chi}^{\perp}_2}||\leq\tau_1\tau_2||\ket{\chi}||^2$, we have 
 \begin{align}
 Re\braket{\tilde{\chi}^{\perp}_1|\tilde{\chi}^{\perp}_2}\geq -\tau_1\tau_2||\ket{\chi}||^2.\label{C5}
 \end{align}
  From Eq. (\ref{C2}) and (\ref{C3}), we have $|c_1|^2||\ket{\tilde{\Phi}^s}||^2=||\ket{\chi}||^2+||\ket{\tilde{\chi}^{\perp}_1}||^2\leq (1+\tau_1^2)||\ket{\chi}||^2$ and $|c_2|^2||\ket{\tilde{\Phi}^{1-s}_{H_{tot}}}||^2=||\ket{\chi}||^2+||\ket{\tilde{\chi}^{\perp}_2}||^2\leq (1+\tau_2^2)||\ket{\chi}||^2$. Hence 
   \begin{align}
   |c_1|^2|c_2|^2||\ket{\tilde{\Phi}^s}||^2||\ket{\tilde{\Phi}^{1-s}_{H_{tot}}}||^2\leq (1+\tau_1^2)(1+\tau_2^2)||\ket{\chi}||^4.\label{C6}
 \end{align}
 Again from Eq. (\ref{C2}) and (\ref{C3}), we have $c_1^*c_2\braket{\tilde{\Phi}^s|\tilde{\Phi}^{1-s}_{H_{tot}}}=||\ket{\chi}||^2+\braket{\tilde{\chi}^{\perp}_1|\tilde{\chi}^{\perp}_2}$ implying  $|c_1|^2|c_2|^2|\braket{\tilde{\Phi}^s|\tilde{\Phi}^{1-s}_{H_{tot}}}|^2=\left(||\ket{\chi}||^2+Re\braket{\tilde{\chi}^{\perp}_1|\tilde{\chi}^{\perp}_2}\right)^2+\left(Im\braket{\tilde{\chi}^{\perp}_1|\tilde{\chi}^{\perp}_2}\right)^2$. Now, by using Eq. (\ref{C5}), it becomes
   \begin{align}
|c_1|^2|c_2|^2|\braket{\tilde{\Phi}^s|\tilde{\Phi}^{1-s}_{H_{tot}}}|^2&\geq\left(||\ket{\chi}||^2+Re\braket{\tilde{\chi}^{\perp}_1|\tilde{\chi}^{\perp}_2}\right)^2\nonumber\\
&\geq \left(||\ket{\chi}||^2-\tau_1\tau_2||\ket{\chi}||^2\right)^2\nonumber\\
&=(1-\tau_1\tau_2)^2||\ket{\chi}||^4.\label{C7}
 \end{align}
 By combining Eqs. (\ref{C6}) and (\ref{C7}), and putting $|\braket{\tilde{\Phi}^s|\tilde{\Phi}^{1-s}_{H_{tot}}}|=\braket{\tilde{\Phi}^s|{H_{tot}}|\tilde{\Phi}^{1-s}}$, $||\ket{\tilde{\Phi}^{1-s}_{H_{tot}}}||=\sqrt{\braket{\Phi^{1-s}|H_{tot}^2|\Phi^{1-s}}}\sqrt{\braket{\tilde{\Phi}^{1-s}|\tilde{\Phi}^{1-s}}}$, where $\ket{\Phi^{1-s}}=\ket{\tilde{\Phi}^{1-s}}/\sqrt{\braket{\tilde{\Phi}^{1-s}|\tilde{\Phi}^{1-s}}}$, we have
   \begin{align}
        |\braket{\Tilde{\Phi}^s|{H_{tot}}|\tilde{\Phi}^{1-s}}| \geq  \frac{(1-\tau_1\tau_2)\Theta_s}{(1+\tau_1^2)(1+\tau_2^2)}\sqrt{\braket{\Phi^{1-s}|H_{tot}^2|\Phi^{1-s}}},\label{C8}
    \end{align}
    where $\Theta_s=\sqrt{\braket{\tilde{\Phi}^{s}|\tilde{\Phi}^{s}}\braket{\tilde{\Phi}^{1-s}|\tilde{\Phi}^{1-s}}}=\sqrt{[\Tr\rho^{2s}][\Tr\rho^{2(1-s)}]}$.  To hold Eq. (\ref{C8}) nontrivially, it must be $\tau_1\tau_2<1$. 
 By using Eq. (\ref{C2}) and (\ref{C3}), and $||\ket{\chi}||=1$, in Eq. (\ref{C4}), we have the conditions on $\tau_1$ and $\tau_2$ as
 \begin{equation*}
 \begin{split}
   \tau_1&\geq \frac{1}{|\braket{\chi|\Phi^s}|^2}-1,\\
    \tau_2&\geq \frac{1}{|\braket{\chi|\Phi^{1-s}_{H_{tot}}}|^2}-1,\\
    s.t\\
    &\tau_1\tau_2<1,
    \end{split}
    \end{equation*}
where $\ket{\Phi^s}=\ket{\tilde{\Phi}^s}/\sqrt{\braket{\tilde{\Phi}^s|\tilde{\Phi}^s}}$,  $\ket{\Phi^{1-s}_{H_{tot}}}=\ket{\tilde{\Phi}^{1-s}_{H_{tot}}}/\sqrt{\braket{\tilde{\Phi}^{1-s}_{H_{tot}}|\tilde{\Phi}^{1-s}_{H_{tot}}}}$.  To make  inequality (\ref{C8}), the tightest one, we have to maximize the right hand side of  inequality (\ref{C8}) over $\tau_1$ and $\tau_2$ and thus inequality (\ref{IIIB1-1}) is proved.\par
($ii$) The left hand side of  inequality (\ref{C8}) remains unchanged  under $s\rightarrow 1-s$, but the right hand side does not. Hence, we will get another lower bound of $|\braket{\Tilde{\Phi}^s|H_{tot}|\tilde{\Phi}^{1-s}}|$ and can be obtained by just substituting  $s\rightarrow 1-s$ in  inequality (\ref{C8}). This is given in the main text in  inequality (\ref{IIIB1-2}).

\section{}\label{D}
Here, our first goal is to find  the state(s) $\ket{\Phi_0}\in \mathcal{H}\otimes\mathcal{H}$ for which   $\braket{{\Phi_0}|H_{tot}^2|{\Phi_0}}=0$ and based on that we will find the minimum value of $\braket{{\Phi^{1-s}}|H_{tot}^2|{\Phi^{1-s}}}$ and $\braket{{\Phi^{s}}|H_{tot}^2|{\Phi^{s}}}$.  As both $H_{tot}$ and $H_{tot}^2$  are positive semidefinite operators,  $\braket{{\Phi_0}|H_{tot}|{\Phi_0}}=0$  for the same $\ket{\Phi_0}$  also. Hence,  for  convenience, we will start with the condition $\braket{{\Phi_0}|H_{tot}|{\Phi_0}}=0$ to find  $\ket{\Phi_0}$. Note that  $\braket{{\Phi_0}|H^2_{k,n}|{\Phi_0}}=0$, $\forall k$, and $n=1,2$ $\implies$ $\braket{{\Phi_0}|H_{tot}|{\Phi_0}}=0$.  It is easy to  show that $\braket{{\Phi_0}|H_{k,n}^2|{\Phi_0}}=0$ if and only if $\ket{\Phi_0}\in Ker(H_{k,n})$ (see Appendix \ref{H}). Now, as each  $A_{k,n}$  is  Hermitian, then   we have  $H_{k,n}=\frac{1}{\sqrt{2}}(A_{k,n}\otimes I-I\otimes A_{k,n}^T)=\sum_{i,j}(a^{(i)}_{k,n}-a^{(j)}_{k,n})\ket{a^{(i)}_{k,n}}\ket{{a^{(j)^*}_{k,n}}}\bra{a^{(i)}_{k,n}}\bra{{a^{(j)^*}_{k,n}}}$, where  \{$\ket{a^{(i)}_{k,n}}\}$ and \{$\ket{{a^{(j)^*}_{k,n}}}\}$ are the eigenstates of $A_{k,n}$ and   $A_{k,n}^T$, respectively. See Appendix \ref{H} for the properties of $A_{k,n}^T$. Thus $span$\{$\ket{a^{(1)}_{k,n}}\ket{{a^{(1)^*}_{k,n}}}$, $\ket{a^{(2)}_{k,n}}\ket{{a^{(2)^*}_{k,n}}}$,$\cdots$,$\ket{a^{(d)}_{k,n}}\ket{{a^{(d)^*}_{k,n}}}\}=Ker(H_{k,n})$. This result remains  same  also for the case of degeneracy of the eigenspectrums of $A_{k,n}$.\par
    Now we consider only two Hermitian operators  $A_{k,n}$ and $A_{k^{\prime},n^{\prime}}$  from the set $\{A_{k,1},A_{k,2}\}_{k=1}^N$, where $n,n^{\prime}=1,2$ and this analysis will make the result for the case of $\{A_{k,1},A_{k,2}\}_{k=1}^N$ easier. We also assume that $A_{k,n}$ and $A_{k^{\prime},n^{\prime}}$ do not share any common eigenstate. It is  obvious that to get  $\braket{{\Phi_0}|H_{k,n}^2|{\phi_0}}+\braket{{\Phi_0}|H_{k^{\prime},n^{\prime}}^2|{\Phi_0}}=||H_{k,n}\ket{\Phi_0}||^2+||H_{k^{\prime},n^{\prime}}\ket{\Phi_0}||^2=0$, one must have $\ket{\Phi_0}\in Ker(H_{k,n})$ as well as $\ket{\Phi_0}\in Ker(H_{k^{\prime},n^{\prime}})$ or equivalently, $\ket{\Phi_0}\in Ker(H_{k,n})\cap Ker(H_{k^{\prime},n^{\prime}})$. Here we assume that $Ker(H_{k,n})\cap Ker(H_{k^{\prime},n^{\prime}})\neq \varnothing$. This suggests that  $\ket{\Phi_0}$ can be written as the linear combination of the elements of $Ker(H_{k,n})$  and individually of the elements of $Ker(H_{k^{\prime},n^{\prime}})$. In other words,
    \begin{align}
\ket{\Phi_0}=\sum_{i=1}^d\alpha_{k,n}^{(i)}\ket{{a_{k,n}^{(i)}}}\ket{{a_{k,n}^{(i)^*}}}=\sum_{j=1}^d\alpha_{k^{\prime},n^{\prime}}^{(j)}\ket{{a_{k^{\prime},n^{\prime}}^{(j)}}}\ket{{a_{k^{\prime},n^{\prime}}^{(j)^*}}}.\label{D1}
  \end{align}
    Now write $\alpha_{k,n}^{(i)}=|\alpha_{k,n}^{(i)}|e^{i\phi^{(i)}_{k,n}}$ and $\alpha_{k^{\prime},n^{\prime}}^{(j)}=|\alpha_{k^{\prime},n^{\prime}}^{(j)}|e^{i\phi^{(j)}_{k^{\prime},n^{\prime}}}$. By putting $e^{i\phi^{(i)}_{k,n}/2}\ket{{a_{k,n}^{(i)}}}=\ket{{\tilde{a}_{k,n}^{(i)}}}$ and $e^{i\phi^{(i)}_{k,n}/2}\ket{{a_{k,n}^{(i)^*}}}=\ket{{\tilde{a}_{k,n}^{(i)^*}}}$, and similarly $e^{i\phi^{(j)}_{k^{\prime},n^{\prime}}/2}\ket{{a_{k^{\prime},n^{\prime}}^{(j)}}}=\ket{{\tilde{a}_{k^{\prime},n^{\prime}}^{(j)}}}$ and $e^{i\phi^{(j)}_{k^{\prime},n^{\prime}}/2}\ket{{a_{k^{\prime},n^{\prime}}^{(j)^*}}}=\ket{{\tilde{a}_{k^{\prime},n^{\prime}}^{(j)^*}}}$ in Eq. (\ref{D1}), we have 
  \begin{align}
\ket{\Phi_0}=\sum_{i=1}^d|\alpha_{k,n}^{(i)}|\ket{{\tilde{a}_{k,n}^{(i)}}}\ket{{\tilde{a}_{k,n}^{(i)^*}}}=\sum_{j=1}^d|\alpha_{k^{\prime},n^{\prime}}^{(j)}|\ket{{\tilde{a}_{k^{\prime},n^{\prime}}^{(j)}}}\ket{{\tilde{a}_{k^{\prime},n^{\prime}}^{(j)^*}}},\label{D2}
    \end{align}
where $\ket{\Phi_0}$ is  now in the Schmidt decomposition form with the basis $\{\ket{{\tilde{a}_{k,n}^{(i)}}}\ket{{\tilde{a}_{k,n}^{(i)^*}}}\}_{i=1}^d$ or $\{\ket{{\tilde{a}_{k^{\prime},n^{\prime}}^{(j)}}}\ket{{\tilde{a}_{k^{\prime},n^{\prime}}^{(j)^*}}}\}_{j=1}^d$. Eq. (\ref{D2}) can equivalently be understood in the form of reduced density operator as
\begin{align}
\sum_{i=1}^d|\alpha_{k,n}^{(i)}|^2\ket{{\tilde{a}_{k,n}^{(i)}}}\!\!\bra{{\tilde{a}_{k,n}^{(i)}}}=\sum_{j=1}^d|\alpha_{k^{\prime},n^{\prime}}^{(j)}|^2\ket{{\tilde{a}_{k^{\prime},n^{\prime}}^{(j)}}}\!\!\bra{{\tilde{a}_{k^{\prime},n^{\prime}}^{(j)}}}.\label{D3}
\end{align}
Now the following case provides the validation of the equality of Eq. (\ref{D3}):\\
 If  there are $d-1$  number of degenerate eigenvalues in $\{|\alpha_{k,n}^{(i)}|\}_{i=1}^d$ and $\{|\alpha_{k^{\prime},n^{\prime}}^{(j)}|\}_{j=1}^d$ \emph{i.e}., for example, $|\alpha_{k,n}^{(1)}|=\cdots=|\alpha_{k,n}^{(d-1)}|$  and $|\alpha_{k^{\prime},n^{\prime}}^{(1)}|=\cdots=|\alpha_{k^{\prime},n^{\prime}}^{(d-1)}|$ (this is the extreme case except for the case where all the eigenvalues are degenerate), then  Eq. (\ref{D3}) becomes
\begin{align}
&|\alpha_{k,n}^{(d)}|^2\ket{{\tilde{a}_{k,n}^{(d)}}}\!\!\bra{{\tilde{a}_{k,n}^{(d)}}}+|\alpha_{k,n}^{(1)}|^2\sum_{i=1}^{d-1}\ket{{\tilde{a}_{k,n}^{(i)}}}\!\!\bra{{\tilde{a}_{k,n}^{(i)}}}\nonumber\\
&=|\alpha_{k^{\prime},n^{\prime}}^{(d)}|^2\ket{{\tilde{a}_{k^{\prime},n^{\prime}}^{(d)}}}\!\!\bra{{\tilde{a}_{k^{\prime},n^{\prime}}^{(d)}}}+|\alpha_{k^{\prime},n^{\prime}}^{(1)}|^2\sum_{j=1}^{d-1}\ket{{\tilde{a}_{k^{\prime},n^{\prime}}^{(j)}}}\!\!\bra{{\tilde{a}_{k^{\prime},n^{\prime}}^{(j)}}}.\label{D4}
\end{align}
Now remember  our initial condition  which is that the operators $\{A_{k,n},A_{k^{\prime},n^{\prime}}\}$ do not share any common eigenstate and  hence $\ket{{\tilde{a}_{k,n}^{(i)}}}\neq \ket{{\tilde{a}_{k^{\prime},n^{\prime}}^{(i)}}}$ implying that the equality in   Eq. (\ref{D4}) doesn't hold. We thus conclude that Eq. (\ref{D3}) or equivalently  Eq. (\ref{D2}) holds only when all the eigenvalues in $\{|\alpha_{k,n}^{(i)}|\}_{i=1}^d$  and $\{|\alpha_{k^{\prime},n^{\prime}}^{(j)}|\}_{j=1}^d$   are degenerate \emph{i.e}., when 
 \begin{align}
\ket{\Phi_0}=\frac{1}{\sqrt{d}}\sum_{i=1}^d\ket{{\tilde{a}_{k,n}^{(i)}}}\ket{{\tilde{a}_{k,n}^{(i)^*}}}=\frac{1}{\sqrt{d}}\sum_{j=1}^d\ket{{\tilde{a}_{k^{\prime},n^{\prime}}^{(j)}}}\ket{{\tilde{a}_{k^{\prime},n^{\prime}}^{(j)^*}}}\label{D5}
    \end{align}
 is a maximally entangled state. Thus, only  maximally entangled state  can  provide $\braket{{\Phi_0}|H_{k,n}^2|{\Phi_0}}+\braket{{\Phi_0}|H_{k^{\prime},n^{\prime}}^2|{\Phi_0}}=0$.  \par
Now we consider $\{A_{k,1},A_{k,2}\}_{k=1}^N$ in which  $A_{k,n}$ and $A_{k^{\prime},n^{\prime}}$ do not share any common eigenstate although the other Hermitian operators can share common eigenstate. Then, the average of $H_{tot}$ in any state $\ket{\Psi}$  is given by   $\braket{\Psi|H_{tot}|\Psi}=\braket{{\Psi}|H_{k,n}^2|{\Psi}}+\braket{{\Psi}|H_{k^{\prime},n^{\prime}}^2|{\Psi}}+\sum_{l\neq (k,k^{\prime})}^N\sum_{r\neq (n,n^{\prime})}\braket{{\Psi}|H_{l,r}^2|{\Psi}}$. Now, let us find the form of $\ket{\Psi}$ such that $\braket{\Psi|H_{tot}|\Psi}=0$.  One can easily check that if  $\ket{\Psi}= \ket{\Phi_0}$, then $\braket{\Psi|H_{tot}|\Psi}=0$.  And if  $\ket{\Psi}\neq\ket{\Phi_0}$ then  $\braket{{\Psi}|H_{k,n}^2|{\Psi}}+\braket{{\Psi}|H_{k^{\prime},n^{\prime}}^2|{\Psi}}\neq 0$.  Thus, we conclude that  $\braket{\Psi|H_{tot}|\Psi}=0$ or $\braket{\Psi|H_{tot}^2|\Psi}=0$ only when $\ket{\Psi}= \ket{\Phi_0}$.\par
 As maximally entangled state leads  $I^s_{\rho}(A_{k,n})=0$,   the state $\ket{\Psi}$ is not allowed to be maximally entangled state or equivalently, the density operator  $\rho=\Tr_B(\ketbra{\Psi}{\Psi})$ can not be taken as  maximally mixed state. Thus whenever $\epsilon_0=0$, we have 
\begin{align}
\braket{\Psi|H_{tot}^2|\Psi}&=\sum_{i=0}^{d^2-1}\epsilon_i^2|\braket{\Psi|\epsilon_i}|^2\nonumber\\
    &=\sum_{i=1}^{d^2-1}\epsilon_i^2|\braket{\Psi|\epsilon_i}|^2\nonumber\\
    &\geq\sum_{i=1}^{d^2-1}\epsilon_1^2|\braket{\Psi|\epsilon_i}|^2\nonumber\\
    &=\epsilon_1^2\sum_{i=1}^{d^2-1}\braket{\Psi|(I-\ketbra{\epsilon_0}{\epsilon_0})|\Phi}\nonumber\\
    &=\epsilon_1^2(1-|\braket{\epsilon_0|\Psi}|^2),\label{D6}
\end{align}
where we have considered $\epsilon_{d^2-1}\geq\epsilon_{d^2-2}\geq\cdots\geq\epsilon_{2}\geq\epsilon_{1}$ without any loss of generality as $\{\epsilon_i\}$ are the eigenvalues of  the positive  semidefinite operator $H_{tot}$ and used $\sum_{i=0}^{d^2-1}\ketbra{\epsilon_i}{\epsilon_i}=I$. \par
 If  $\ket{\Psi}=\ket{\Phi^{1-s}}$ or $\ket{\Psi}=\ket{\Phi^{s}}$,  then by using $\ket{\epsilon_0}=\ket{\Phi_0}$ defined in Eq. (\ref{D5}), we have
\begin{align}
    min\{(1-|\braket{\epsilon_0|\Phi^{1-s}}|^2)\}&=\left[1-\frac{1}{d}\frac{(\Tr\rho^{1-s})^2}{{\Tr\rho^{2(1-s)}}}\right],\nonumber\\
    min\{(1-|\braket{\epsilon_0|\Phi^{s}}|^2)\}&=\left[1-\frac{1}{d}\frac{(\Tr\rho^{s})^2}{{\Tr\rho^{2s}}}\right],\nonumber
\end{align}
and finally we have, 
\begin{align}
    \braket{\Phi^{1-s}|H_{tot}^2|\Phi^{1-s}}&\geq \epsilon_1^2\left[1-\frac{1}{d}\frac{(\Tr\rho^{1-s})^2}{{\Tr\rho^{2(1-s)}}}\right],\nonumber\\
     \braket{\Phi^{s}|H_{tot}^2|\Phi^{s}}&\geq \epsilon_1^2\left[1-\frac{1}{d}\frac{(\Tr\rho^{s})^2}{{\Tr\rho^{2s}}}\right].\nonumber
\end{align}
\par
To prove the inequality (\ref{IIIB3-2}), note that  from Eq. (\ref{IIIB3-3}), $\sum_{k=1}^NI_{\rho}(A_k)= \braket{{\Phi}|H_{tot}|{\Phi}}$ and the minimum value of  $\braket{{\Phi}|H_{tot}|{\Phi}}$ can be found by the  similar derivation of Eq. (\ref{D6}) when  $\epsilon_0=0$ as
\begin{align}
\braket{\Phi|H_{tot}|\Phi}&\geq\epsilon_1(1-|\braket{\epsilon_0|\Phi}|^2)\nonumber\\
&\geq  min \{ \epsilon_1(1-|\braket{\epsilon_0|\Phi}|^2)\}\nonumber\\
&= \epsilon_1\left[1-\frac{1}{d}(\Tr\sqrt{\rho})^2\right].\nonumber
\end{align}

\onecolumngrid
\section{}\label{E}
In a two-dimensional Hilbert space, there are only two non-zero eigenvalues  $\lambda_1$ and  $\lambda_2$ of a density operator, where $\lambda_1+\lambda_2=1$. Now first calculate the generalized skew information  of any  operator  $\sigma$ which acts on a two-dimensional Hilbert space as
\begin{align}
\cal{I}^{\nu}_{\rho}(\sigma)&=\frac{1}{2}\Tr[\rho(\sigma^{\dagger}\sigma+\sigma\sigma^{\dagger})]-\frac{1}{2}\sum_{i,j=1}^2m_{\nu}(\lambda_i,\lambda_j)(|\braket{\lambda_i|\sigma^{\dagger}|\lambda_j}|^2+|\braket{\lambda_i|\sigma|\lambda_j}|^2)\nonumber\\
&=\frac{1}{2}\lambda_1\braket{\lambda_1|(\sigma^{\dagger}\sigma+\sigma\sigma^{\dagger})|\lambda_1}+\frac{1}{2}\lambda_2\braket{\lambda_2|(\sigma^{\dagger}\sigma+\sigma\sigma^{\dagger})|\lambda_2}-\lambda_1|\braket{\lambda_1|\sigma|\lambda_1}|^2-\lambda_2|\braket{\lambda_2|\sigma|\lambda_2}|^2\nonumber\\
&-m_{\nu}(\lambda_1,\lambda_2)(|\braket{\lambda_1|\sigma^{\dagger}|\lambda_2}|^2+|\braket{\lambda_1|\sigma|\lambda_2}|^2)\nonumber\\
&=\lambda_1\braket{\Delta \sigma}_{\lambda_1}^2+\lambda_2\braket{\Delta \sigma}_{\lambda_2}^2-m_{\nu}(\lambda_1,\lambda_2)(|\braket{\lambda_1|\sigma^{\dagger}|\lambda_2}|^2+|\braket{\lambda_1|\sigma|\lambda_2}|^2),\label{E1}
\end{align}
where we used $|\braket{\lambda_2|\sigma^{\dagger}|\lambda_1}|^2=|\braket{\lambda_1|\sigma|\lambda_2}|^2$ and $|\braket{\lambda_2|\sigma|\lambda_1}|^2=|\braket{\lambda_1|\sigma^{\dagger}|\lambda_2}|^2$. Now for a qubit, we can use the identity $\ketbra{\lambda_1}{\lambda_1}=I-\ketbra{\lambda_2}{\lambda_2}$ to evaluate the following: \\

\begin{align}
|\braket{\lambda_1|\sigma^{\dagger}|\lambda_2}|^2+|\braket{\lambda_1|\sigma|\lambda_2}|^2&=\braket{\lambda_1|\sigma^{\dagger}|\lambda_2}\braket{\lambda_2|\sigma|\lambda_1}+\braket{\lambda_1|\sigma|\lambda_2}\braket{\lambda_2|\sigma^{\dagger}|\lambda_1}\nonumber\\
&=\braket{\lambda_1|\sigma^{\dagger}|(I-\ketbra{\lambda_1}{\lambda_1})\sigma|\lambda_1}+\braket{\lambda_1|\sigma|(I-\ketbra{\lambda_1}{\lambda_1})\sigma^{\dagger}|\lambda_1}\nonumber\\
&=\braket{\lambda_1|\sigma^{\dagger}\sigma|\lambda_1}+\braket{\lambda_1|\sigma\sigma^{\dagger}|\lambda_1}-2|\braket{\lambda_1|\sigma|\lambda_1}|^2\nonumber\\
&=2\braket{\Delta \sigma}_{\lambda_1}^2.\label{E2}
\end{align}
It can also be shown that 
\begin{align}
|\braket{\lambda_1|\sigma^{\dagger}|\lambda_2}|^2+|\braket{\lambda_1|\sigma|\lambda_2}|^2=2\braket{\Delta \sigma}_{\lambda_2}^2\label{E3}
\end{align}
and thus from Eqs. (\ref{E2}) and (\ref{E3}), we have
\begin{align}
\braket{\Delta \sigma}_{\lambda_1}^2=\braket{\Delta \sigma}_{\lambda_2}^2.\label{E4}
\end{align}
By using Eqs. (\ref{E2}) and (\ref{E4}) in Eq. (\ref{E1}), it becomes
\begin{align*}
\cal{I}^{\nu}_{\rho}(\sigma)&=[1-2m_{\nu}(\lambda_1,\lambda_2)]\braket{\Delta \sigma}_{\lambda_l}^2,
\end{align*}
where $l=1$ or $2$.

\twocolumngrid
\section{}\label{F}
Let us first denote the system of our interest  by  subsystem A  and the auxiliary  system by subsystem B.  The associated Hilbert spaces for the subsystems A and B are denoted by $\cal{H}_A\equiv\cal{H}$ and  $\cal{H}_B\equiv\cal{H}$, respectively.   Now, remember from Eq. (\ref{IIIE-1}) that $\ket{{\Phi}^{1-s}}=\frac{1}{\cal{N}_{1-s}}\ket{\tilde{\Phi}^{1-s}}\in\cal{H}_A\otimes\cal{H}_B$, and   similarly  $\ket{{\Phi}^{s}}=\frac{1}{\cal{N}_{s}}\ket{\tilde{\Phi}^{s}}\in\cal{H}_A\otimes\cal{H}_B$ are the normalized states. \par
Our goal is now to determine the weak value  $\braket{H_A}_{{\Phi}^{s}}^{\alpha_i\alpha_j^*}=\frac{\braket{\alpha_i\alpha_j^*|H_A|{\Phi}^{s}}}{\braket{\alpha_i\alpha_j^*|{\Phi}^{s}}}=\frac{1}{\sqrt{2}}\left(\braket{A\otimes I}_{{\Phi}^{s}}^{\alpha_i\alpha_j^*}-\braket{I\otimes A^T}_{{\Phi}^{s}}^{\alpha_i\alpha_j^*}\right)$, where we have substituted $H_{A}=\frac{1}{\sqrt{2}}(A\otimes I-I\otimes A^T)$. Now consider the following weak value of $A\otimes I$:
\begin{align}
\braket{A\otimes I}_{{\Phi}^{s}}^{\alpha_i\alpha_j^*}&=\frac{\braket{\alpha_i\alpha^*_j|A\otimes I|{\Phi}^{s}}}{\braket{\alpha_i\alpha_j^*|{\Phi}^{s}}}\nonumber\\
&=\frac{\braket{\alpha_i|A|\tilde{\phi}^j_A}}{\braket{\alpha_i|{\tilde{\phi}}_A^j}}\nonumber\\
&=\frac{\braket{\alpha_i|A|{\phi}^{j}_A}}{\braket{\alpha_i|{{\phi}}^{j}_A}}\nonumber\\
&=\braket{A}_{{\phi}^{j}_A}^{\alpha_i},\label{F1}
\end{align}
where $\ket{\tilde{\phi}^{j}_A}=\braket{\alpha_j^*|\Phi^s}$ is an unnormalized state vector in $\cal{H}_A$ and $\ket{{\phi}^{j}_A}=\ket{\tilde{\phi}^{j}_A}/\sqrt{\braket{\tilde{\phi}^{j}_A|\tilde{\phi}^{j}_A}}$. Thus according to the Eq. (\ref{F1}), the weak value  $\braket{A\otimes I}_{{\Phi}^{s}}^{\alpha_i\alpha_j^*}$  can be obtained in experiment by first measuring  a projective operator $\ketbra{\alpha_j^*}{\alpha_j^*}$ in  subsystem B, and then there will be an immediate state collapse in  subsystem A which is given by $\ket{{\phi}^{j}_A}$, here the  state shared  between  subsystems A and B is $\ket{\Phi^s}$. Now the pre-selection  in the subsystem A is $\ket{{\phi}^{j}_A}$, the operator is $A$ which is unknown, and the postselection is $\ket{\alpha_i}$. Then in this set up, the weak value will be $\braket{A}_{{\phi}^{j}_A}^{\alpha_i}$. Note that, weak value of an unknown observable is experimentally obtainable. \par
Similarly, one can obtain $\braket{I\otimes A^T}_{{\Phi}^{s}}^{\alpha_i\alpha_j^*}$ in experiment in the following way:
\begin{align}
\braket{I\otimes A^T}_{{\Phi}^{s}}^{\alpha_i\alpha_j^*}&=\frac{\braket{\alpha_i\alpha_j^*|I\otimes A^T|{\Phi}^{s}}}{\braket{\alpha_i\alpha_j^*|{\Phi}^{s}}}\nonumber\\
&=\frac{\sum_k\lambda^s_k\braket{\alpha_i|\lambda_k}\braket{\alpha^*_j|A^T|\lambda_k^*}}{\sum_k\lambda^s_k\braket{\alpha_i|\lambda_k}\braket{\alpha^*_j|\lambda_k^*}}\nonumber\\
&=\frac{\sum_k\lambda^s_k\braket{\lambda_k|A|\alpha_j}\braket{\lambda_k^*|\alpha_i^*}}{\sum_k\lambda^s_k\braket{\lambda_k|\alpha_j}\braket{\lambda_k^*|\alpha_i^*}}\nonumber\\
&=\frac{\braket{\Phi^s|A\otimes I|\alpha_j\alpha_i^*}}{\braket{{\Phi}^{s}|\alpha_j\alpha_i^*}}\nonumber\\
&=\left[\frac{\braket{\alpha_j\alpha_i^*|A\otimes I|{\Phi}^{s}}}{\braket{\alpha_j\alpha_i^*|{\Phi}^{s}}}\right]^*\nonumber\\
&=\left[\braket{A\otimes I}_{{\Phi}^{s}}^{\alpha_j\alpha_i^*}\right]^*,\label{F2}
\end{align}
where we have used $\ket{{\Phi}^{s}}=\frac{1}{\cal{N}_{s}}\ket{\tilde{\Phi}^{s}}=\frac{1}{\cal{N}_{s}}\sum_{k}\lambda^s_k\ket{\lambda_k}\ket{\lambda^*_k}$, and $\braket{\alpha^*_j|A^T|\lambda_k^*}=\braket{\lambda_k|A|\alpha_j}$ and $\braket{\alpha_i|\lambda_k}=\braket{\lambda_k^*|\alpha_i^*}$ \cite{proof-1}. \par
Using the similar strategy used to obtain the final expression in Eq. (\ref{F1}), we can express $\braket{I\otimes A^T}_{{\Phi}^{s}}^{\alpha_i\alpha_j^*}$ given in Eq. (\ref{F2})  as follows:
\begin{align}
\braket{I\otimes A^T}_{{\Phi}^{s}}^{\alpha_i\alpha_j^*}&=\left[\braket{A}_{{\phi}^{i}_A}^{\alpha_j}\right]^*,\label{F3}
\end{align}
where $\ket{{\phi}^{i}_A}=\ket{\tilde{\phi}^{i}_A}/\sqrt{\braket{\tilde{\phi}^{i}_A|\tilde{\phi}^{i}_A}}$ and $\ket{\tilde{\phi}^{i}_A}=\braket{\alpha_i^*|\Phi^s}$ is an unnormalized state vector in $\cal{H}_A$. Thus according to the Eq. (\ref{F3}), the weak value  $\braket{I\otimes A^T}_{{\Phi}^{s}}^{\alpha_i\alpha_j^*}$  can be obtained in experiment by first measuring  a projective operator $\ketbra{\alpha_i^*}{\alpha_i^*}$ in  subsystem B, and then there will be an immediate state collapse in  subsystem A which is given by $\ket{{\phi}^{i}_A}$, here the  state shared  between  subsystems A and B is $\ket{\Phi^s}$. Now the pre-selection  in the subsystem A is $\ket{{\phi}^{i}_A}$, the operator is $A$ which is unknown, and the postselection is $\ket{\alpha_j}$. Then in this set up, the weak value will be $\braket{A}_{{\phi}^{i}_A}^{\alpha_j}$.\par
By changing the preselection to $\ket{\Phi^{1-s}}$ in the above description, we can determine the weak value  $\braket{H_A}_{{\Phi}^{1-s}}^{\alpha_i\alpha_j^*}$ also and hence we obtain the full information of $I_{\rho}^s(A)$ according to Eq. (\ref{IIIE-3}).

\onecolumngrid
\section{}\label{G}
The sum uncertainty relation for non-Hermitian operators \{$A_k\}^N_{k=1}$ is given by  Eq. (\ref{IIIB4-1}) as
\begin{align}
        \sum_{k=1}^N\braket{\Delta A_k}_{\psi}^2\geq \epsilon_0.\label{G1}
    \end{align}
This inequality can be useful to detect entanglement of an arbitrary state in a  bipartite system as follows. Consider a generic separable state 
\begin{align}
\rho_{AB}=\sum_ip_i\ketbra{\psi^A_i}{\psi^A_i}\otimes \ketbra{\psi^B_i}{\psi^B_i}.\label{G2}
\end{align}  
Now for a separable state, the sum of the variances of the non-Hermitian operators \{$(A_k\otimes I+I\otimes B_k)\}^{N}_{k=1}$  satisfies the following inequality: 
\begin{align}
\sum_{k=1}^N\braket{\Delta (A_k\otimes I+I\otimes B_k)}^2_{\rho_{AB}}&\geq\sum_{k=1}^N\sum_{i}p_i\braket{\Delta (A_k\otimes I+I\otimes B_k)}^2_{\ket{\psi_i^A}\ket{\psi_i^B}}\nonumber\\
&=\sum_{k=1}^N\sum_{i}p_i(\braket{\Delta A_k}^2_{\psi_i^A}+\braket{\Delta B_k}^2_{\psi_i^B}),\label{G3}
\end{align}
where we have used the concavity property of variance in the first inequality and to achieve the last equality, we have used the fact that for any product state: $\braket{\Delta (A\otimes I+I\otimes B)}^2_{\psi^A\psi^B}=\braket{\Delta A}^2_{\psi^A}+\braket{\Delta B}^2_{\psi^B}$. By using inequality (\ref{G1}) for the non-Hermitian operators  \{$A_k\}^N_{k=1}$ and \{$B_k\}^N_{k=1}$ for the subsystems A and B, respectively, we obtain
\begin{align}
\sum_{k=1}^N\braket{\Delta (A_k\otimes I+I\otimes B_k)}^2_{\rho_{AB}}\geq \epsilon_0^A+\epsilon_0^B.\label{G4}
\end{align}
Thus any separable state must satisfy the inequality (\ref{G4}) and if the inequality (\ref{G4}) is violated for a bipartite state, then the state must be entangled. The inequality (\ref{G4}) for Hermitian operators is well known \cite{Hofmann-2003}.  Thus inequality (\ref{G4}) for non-Hermitian operators is more general and it can detect more class of entangled states than the one given in \cite{Hofmann-2003} by using a suitably chosen  set of  Hermitian operators or  non-Hermitian operators or both. \par
Recently, the authors in \cite{Zhao-2024} have shown that a particular set of non-Hermitian operators namely the Kraus operators of a channel can be used to detect entanglement of a bipartite system using the following inequality:
 \begin{align}
\sum_{k=1}^N\braket{\Delta (E^A_k\otimes I+I\otimes E^B_k)}^2_{\rho_{AB}}\geq 2-F^{max}(\cal{E}^A)-F^{max}(\cal{E}^B),\label{G5}
\end{align}
where $\cal{E}^A$ and $\cal{E}^B$ are  local channels for subsystems A and B, respectively  with the Kraus representations $\cal{E}^A(\rho^A)=\sum_kE^A_k\rho^A{E^A_k}^{\dagger}$ and $\cal{E}^B(\rho^B)=\sum_kE^B_k\rho^B{E^B_k}^{\dagger}$. Here, $F^{max}(\cal{E}^A)=\underset{\phi_A}{max}\braket{\phi^A|\cal{E}^A(\ketbra{\phi^A}{\phi^A})|\phi^A}$ and $F^{max}(\cal{E}^B)=\underset{\phi_B}{max}\braket{\phi^B|\cal{E}^B(\ketbra{\phi^B}{\phi^B})|\phi^B}$. To compare, their protocol inequality (\ref{G5}) is valid only for non-Hermitian operators which are only the Kraus operators. In contrast, our protocol that is inequality (\ref{G4}) is valid for any non-Hermitian operators.

\section{}\label{H}
\begin{lemma}
    For any  positive semidefinite  operator $X$ and $\ket{\psi}\in \mathcal{H}$, where $\cal{H}$ is the Hilbert space
    \begin{align}
        \braket{\psi|X^2|\psi}=0\nonumber
    \end{align}
    if and only if  $\ket{\psi}\in Ker(X)$.
\end{lemma}
\begin{proof}
    If $\ket{\psi}\in Ker(X)$, then $X\ket{\psi}=0\implies X^2\ket{\psi}=0\implies \braket{\psi|X^2|\psi}=0$. \par
    Now, consider $\braket{\psi|X^2|\psi}=0  \implies ||X\ket{\psi}||^2=0
        \implies X\ket{\psi}=0\implies \ket{\psi}\in Ker(X)$.  
\end{proof}
\textbf{Properties of $Y^T$:} If $Y$ is a Hermitian operator, then  the spectral decomposition of $Y^T$ is given by 
\begin{align}
    Y^T&=\left(\sum_{i=1}^dy_i\ketbra{y_i}{y_i}\right)^T\nonumber\\
    &=\sum_{i=1}^dy_i\left(\ketbra{y_i}{y_i}\right)^T.\label{H1}
\end{align}
Now, let  $\ket{y_i}=\sum_k\alpha^k_i\ket{k}$ in the computational basis \{$\ket{k}\}_k$, then we define  $\ket{y_i^*}=\sum_k{\alpha^k_i}^*\ket{k}$.  It can be shown that $\left(\ketbra{y_i}{y_i}\right)^T=\sum_{k,k^{\prime}}{\alpha^{k^{\prime}}_i}^*\alpha^k_i(\ketbra{k}{k^{\prime}})^T=\sum_{k,k^{\prime}}{\alpha^{k^{\prime}}_i}^*\alpha^k_i\ketbra{k^{\prime}}{k}=\sum_{k,k^{\prime}}{\alpha^{k}_i}^*\alpha^{k^{\prime}}_i\ketbra{k}{k^{\prime}}=\ketbra{y_i^*}{y_i^*}$. Thus Eq. (\ref{H1}) becomes
\begin{align*}
    Y^T=\sum_{i=1}^dy_i\ketbra{y_i^*}{y_i^*}.
\end{align*}
 \par
Note that \{$\ket{y_i^*}\}^d_{i=1}$ are orthonormal vectors:  $\braket{y^*_i|y^*_j}=\braket{y_j|y_i}=\delta_{ij}$ \cite{proof-1} and the above equation represents the spectral decomposition of $Y^T$.

\twocolumngrid


\begin{thebibliography}{}

\bibitem{robertson}
H. P. Robertson, \href{https://journals.aps.org/pr/abstract/10.1103/PhysRev.34.163}{Phys. Rev. 34, 163 (1929)}.

\bibitem{Busch-2014}
Paul Busch, Pekka Lahti, and Reinhard F. Werner, \href{https://link.aps.org/doi/10.1103/PhysRevA.89.012129}{Phys. Rev. A 89, 012129 (2014)}.

\bibitem{Bagchi-2016}
Shrobona Bagchi and Arun Kumar Pati, \href{https://link.aps.org/doi/10.1103/PhysRevA.94.042104}{Phys. Rev. A 94, 042104 (2016)}.


\bibitem{Mondal-2017}
Debasis Mondal, Shrobona Bagchi, and Arun Kumar Pati, \href{https://link.aps.org/doi/10.1103/PhysRevA.95.052117}{Phys. Rev. A 95, 052117 (2017)}.


\bibitem{Huang-2012}
Yichen Huang, \href{https://link.aps.org/doi/10.1103/PhysRevA.86.024101}{Phys. Rev. A 86, 024101 (2012)}.



\bibitem{Guise-2018}
Hubert de Guise, Lorenzo Maccone, Barry C. Sanders, and Namrata Shukla, \href{https://link.aps.org/doi/10.1103/PhysRevA.98.042121}{Phys. Rev. A 98, 042121 (2018)}.



\bibitem{Abbott-2016}
A. A. Abbott, P. L. Alzieu, M. J. Hall, and C. Branciard, \href{https://doi.org/10.3390/math4010008}{Mathematics 4, 8 (2016)}.



\bibitem{Werner-2015}
L. Dammeier, R. Schwonnek, and R. F. Werner, \href{https://iopscience.iop.org/article/10.1088/1367-2630/17/9/093046/meta}{New J. Phys. 17, 093046 (2015)}.



\bibitem{Werner-2017}
René Schwonnek, Lars Dammeier, and Reinhard F. Werner, \href{https://link.aps.org/doi/10.1103/PhysRevLett.119.170404}{Phys. Rev. Lett. 119, 170404 (2017)}.


\bibitem{Karol-2019}
Konrad Szymański and Karol Życzkowski,  \href{https://iopscience.iop.org/article/10.1088/1751-8121/ab4543/pdf}{J. Phys. A: Math. Theor. 53 015302 (2020)}.





\bibitem{maassen-uffink}
Hans Maassen and J. B. M. Uffink,
\href{https://journals.aps.org/prl/abstract/10.1103/PhysRevLett.60.1103}{Phys.Rev.Lett.60, 1103 (1988)}.

\bibitem{sahil-sohail-sibasish}
Sahil, Sohail and  Sibasish Ghosh, \href{https://journals.aps.org/pra/abstract/10.1103/PhysRevA.108.032206}{Phys. Rev. A 108, 032206 (2023)}.

\bibitem{wigner-yanase}
E.P. Wigner and M.M. Yanase, \href{https://doi.org/10.1073/pnas.49.6.910}{Proc. Natl. Acad. Sci. USA 49 (1963) 910–918}.

\bibitem{Luo-07-05}
Luo, S. \href{https://doi.org/10.1007/s11232-007-0054-8}{Theor Math Phys 151, 693–699 (2007)}; 
Shunlong Luo, \href{https://journals.aps.org/pra/abstract/10.1103/PhysRevA.72.042110}{Phys. Rev. A 72, 042110 (2005)}.

\bibitem{Alain-1978}
Alain Connes and Erling Størmer, \href{https://doi.org/10.1016/0022-1236(78)90085-X}{J. Funct. Anal. 28, 187 (1978).}

\bibitem{Luo-2003}
Shunlong Luo, \href{https://journals.aps.org/prl/abstract/10.1103/PhysRevLett.91.180403}{Phys. Rev. Lett. 91, 180403 (2003)}.

\bibitem{Girolami-Luo}
Davide Girolami, \href{https://journals.aps.org/prl/abstract/10.1103/PhysRevLett.113.170401}{Phys. Rev. Lett. 113, 170401 (2014)}; Shunlong Luo, Yuan Sun, \href{https://journals.aps.org/pra/abstract/10.1103/PhysRevA.96.022130}{Phys. Rev. A 96, 022130 (2017)}.

\bibitem{Marvian-2016}
I. Marvian and R. W. Spekkens, \href{https://www.nature.com/articles/ncomms4821}{Nat. Commun. 5, 3821 (2014)}; Iman Marvian, Robert W. Spekkens, and Paolo Zanardi, 
\href{https://journals.aps.org/pra/abstract/10.1103/PhysRevA.93.052331}{Phys. Rev. A 93, 052331 (2016)}.

\bibitem{Luo-Zhang}
 S. Luo and Q. Zhang, \href{https://ieeexplore.ieee.org/document/1317121}{EEE Trans. Inform. Theory 50 (2004) 1778–1782}.
 
 \bibitem{Chen-2016}
Bin Chen, Shao-Ming Fei, and Gui-Lu Long, \href{https://link.springer.com/article/10.1007/s11128-016-1274-3}{Quantum Inf. Process. 15, 2639 (2016)}.


\bibitem{Cai-2021}
L. Cai, \href{https://doi.org/10.1007/s11128-021-03008-0}{Quantum Inf. Process. 20, 72 (2021)}.


 
 \bibitem{Yanagi-2005}
 K. Yanagi, S. Furuichi, and K. Kuriyama, \href{https://ieeexplore.ieee.org/document/1542436/citations#citations}{EEE Trans. Inform. Theory 51 (2005) 4401–4404}.
 

\bibitem{Furuichi-2009}
S. Furuichi \emph{et al.},  \href{https://www.sciencedirect.com/science/article/pii/S0022247X0900184X?ref=cra_js_challenge&fr=RR-1}{J. Math. Anal. Appl. 356 (2009) 179–185}.

\bibitem{Yanagi-2010}
Kenjiro Yanagi, \href{https://iopscience.iop.org/article/10.1088/1742-6596/201/1/012015}{2010 J. Phys.: Conf. Ser. 201 012015}; Kenjiro Yanagi, \href{https://www.sciencedirect.com/science/article/pii/S0022247X09008166}{J. Math. Anal. Appl. 365 (2010) 12–18}.



\bibitem{Busch-2007}
  {P. Busch, T. Heinonen and P. Lahti},
 \href {https://www.sciencedirect.com/science/article/pii/S0370157307003481?via%3Dihub} {Phys.Rep.452,155(2007)}.



\bibitem{Lahti-1987}
  {P. J. Lahti and M. J. Maczynski}, \href {https://doi.org/10.1063/1.527822}{ J. Math. Phys. (N.Y.) 28, 1764 (1987)}.



\bibitem{Hofmann-2003}
 {Holger F. Hofmann and Shigeki Takeuchi}, \href{https://doi.org/10.1103/PhysRevA.68.032103} { Phys. Rev. A 68, 032103 (2003)}.



\bibitem{Guhne-2004}
{ O. Guhne}, \href{https://doi.org/10.1103/PhysRevLett.92.117903}{ Phys. Rev. Lett. 92, 117903 (2004)}.



\bibitem{Fuchs-1996}
{ C. A. Fuchs and A. Peres}, \href{https://doi.org/10.1103/PhysRevA.53.2038}{ Phys. Rev. A 53, 2038 (1996)}.


\bibitem{mjwhall-2005}
  {Michael J. W. Hall}, \href{https://doi.org/10.1007/s10714-005-0131-y}{ Gen. Relativ. Gravit. 37, 1505 (2005)}.




\bibitem{Maccone-Pati}
  {Lorenzo Maccone and  Pati K. Arun}, \href{https://doi.org/10.1103/PhysRevLett.113.260401} {Phys. Rev. Lett. 113,  260401 (2014)}.


\bibitem{Yao-2015}
  {Yao Yao, Xing Xiao, Xiaoguang Wang, and C. P. Sun}, \href{https://doi.org/10.1103/PhysRevA.91.062113}{   Phys. Rev. A 91, 062113 (2015)}.




\bibitem{Bagchi-Pati-2023}
Shrobona Bagchi, Dimpi Thakuria, and Arun K. Pati, \href{https://doi.org/10.3390/e25071046}{Entropy 2023, 25(7), 1046}.



\bibitem{pati-brij-sahil-SLB}
Arun K. Pati, Brij Mohan, Sahil, and Samuel L. Braunstein, \href{https://doi.org/10.48550/arXiv.2305.03839}{arXiv:2305.03839v2}.




\bibitem{Hall-2001}
Michael J. W. Hall, \href{https://link.aps.org/doi/10.1103/PhysRevA.64.052103}{Phys. Rev. A 64, 052103 (2001)}.




\bibitem{Massar-Spindel}
Serge Massar and Philippe Spindel, \href{https://link.aps.org/doi/10.1103/PhysRevLett.100.190401}{Phys. Rev. Lett. 100, 190401 (2008)}.


\bibitem{Shrobona-Pati}
Shrobona Bagchi and Arun Kumar Pati, \href{https://link.aps.org/doi/10.1103/PhysRevA.94.042104}{Phys. Rev. A 94, 042104 (2016)}.



\bibitem{Bong-2018}
  {KW. Bong, N. Tischler, R. B. Patel, S. Wollmann, G. J. Pryde, and M. J. W. Hall}, \href{https://doi.org/10.1103/PhysRevLett.120.230402}{Phys. Rev. Lett. 120, 230402 (2018)}.



\bibitem{Yu-2019}
Bing Yu, Naihuan Jing, and Xianqing Li-Jost, \href{https://link.aps.org/doi/10.1103/PhysRevA.100.022116}{Phys. Rev. A 100, 022116 (2019)}.



\bibitem{Zhao-2024}
Xinzhi Zhao, Xinglei Yu, Wenting Zhou, Chengjie Zhang, Jin-Shi Xu, Chuan-Feng Li, and Guang-Can Guo, \href{https://link.aps.org/doi/10.1103/PhysRevLett.132.070203}{Phys. Rev. Lett. 132, 070203 (2024)}.




\bibitem{Mandelstam-Tamm}
Leonid Mandelstam and IG Tamm,  \href{https://link.springer.com/chapter/10.1007/978-3-642-74626-0_8}{J. Phys. (USSR) 9, 249 (1945)}.



\bibitem{Deffner-Campbell}
Sebastian Deffner and Steve Campbell,  \href{https://iopscience.iop.org/article/10.1088/1751-8121/aa86c6/meta}{J. Phys. A: Math. Theor. 50 453001 (2017)}.




\bibitem{Brij-Pati-observable}
Brij Mohan and Arun Kumar Pati, \href{https://link.aps.org/doi/10.1103/PhysRevA.106.042436}{Phys. Rev. A 106, 042436 (2022)}.





\bibitem{Giorda-2019}
Paolo Giorda, Lorenzo Maccone, and Alberto Riccardi, \href{https://journals.aps.org/pra/abstract/10.1103/PhysRevA.99.052121}{Phys. Rev. A 99, 052121 (2019)}.



\bibitem{Luo-Sun-2017}
Shunlong Luo and Yuan Sun, \href{https://journals.aps.org/pra/abstract/10.1103/PhysRevA.96.022136}{Phys. Rev. A 96, 022136 (2017)}.




\bibitem{Luo-Sun-2018}
Shunlong Luo and Yuan Sun, \href{https://journals.aps.org/pra/abstract/10.1103/PhysRevA.98.012113}{Phys. Rev. A 98, 012113 (2018)}.



\bibitem{Singh-Pati-2016}
U. Singh, A. K. Pati, and M. N. Bera, \href{https://www.mdpi.com/2227-7390/4/3/47}{Mathematics 4, 47 (2016)}.



\bibitem{Fu-Luo-2019}
Fu, S., Sun, Y.  Luo, S., \href{https://doi.org/10.1007/s11128-019-2371-x}{Quantum Inf Process 18, 258 (2019)}.


\bibitem{Yang-2022}
Ma-Cheng Yang and Cong-Feng Qiao, \href{https://journals.aps.org/pra/abstract/10.1103/PhysRevA.106.052401#:~:text=A%20transparent%20proof%20of%20convexity,counterpart%20of%20the%20uncertainty%20relation.}{Phys. Rev. A 106, 052401 (2022)}.


\bibitem{Leblond-1976}
J.-M. Levy-Leblond, \href{https://www.sciencedirect.com/science/article/abs/pii/0003491676902839?via%3Dihub}{Ann. Phys. NY 101, 319 (1976)}.



\bibitem{Hall-Pati-2016}
Michael J. W. Hall, Arun Kumar Pati, and Junde Wu, \href{https://journals.aps.org/pra/abstract/10.1103/PhysRevA.93.052118}{Phys. Rev. A 93, 052118 (2016)}.


\bibitem{Lieb-1973}
E. H. Lieb, \href{https://www.sciencedirect.com/science/article/pii/000187087390011X}{Adv. Math. 11, 267 (1973)}.



\bibitem{Bhatia-Davis-1993}
R. Bhatia, C. Davis, \href{https://epubs.siam.org/doi/abs/10.1137/0614012?journalCode=sjmael}{SIAM J. Matrix Anal. Appl. 14 (1993) 132–136}; Fuad Kittaneh and Yousef Manasrah, \href{https://www.sciencedirect.com/science/article/pii/S0022247X09007057#section-cited-by}{J. Math. Anal. Appl. 361 (2010) 262–269}.




\bibitem{Luo-2005}
  Shunlong Luo, \href{https://doi.org/10.1103/PhysRevA.72.042110}{Phys. Rev. A 72, 042110 (2005)}.



\bibitem{proof-0}
\emph{Proof:}  Let  $\ket{\lambda_i}=\sum_k\alpha^k_i\ket{k}$ in the computational basis \{$\ket{k}\}_k$, then we define  $\ket{\lambda_i^*}=\sum_k{\alpha^k_i}^*\ket{k}$. Now, it can be shown that  $\braket{\lambda^*_i|\lambda^*_j}=\sum_{k,k^{\prime}}\alpha^k_i{\alpha^{k^{\prime}}_j}^*\braket{k|k^{\prime}}=\sum_{k}\alpha^k_i{\alpha^k_j}^*$. Also $\braket{\lambda_j|\lambda_i}=\sum_{k,k^{\prime}}{\alpha^{k^{\prime}}_j}^*{\alpha^{k}_i}\braket{k^{\prime}|k}=\sum_{k}\alpha^k_i{\alpha^k_j}^*$. Thus  $\braket{\lambda^*_i|\lambda^*_j}=\braket{\lambda_j|\lambda_i}=\delta_{ij}$.


\bibitem{proof-1}
\emph{Proof:} For  generic states $\ket{\psi}$, $\ket{\phi}\in \cal{H}$, we define $\ket{\psi^*}=\sum_k\alpha_k^*\ket{k}$ and $\ket{\phi^*}=\sum_k\beta^*_k\ket{k}$ in the computation basis \{$\ket{k}\}$. Now calculate $\braket{\phi|X|\psi}=\sum_l\beta_l^*\bra{l}\left(\sum_{ij}x_{ij}\ketbra{i}{j}\right)\sum_k\alpha_k^*\ket{k}=\sum_{ij}\beta^*_i\alpha_jx_{ij}$ and  $\braket{\psi^*|X^T|\phi^*}=\sum_l\alpha_l\bra{l}\left(\sum_{ij}x_{ji}\ketbra{i}{j}\right)\sum_k\beta_k^*\ket{k}=\sum_{ij}\beta^*_j\alpha_ix_{ji}=\sum_{ij}\beta^*_i\alpha_jx_{ij}$. Hence, we have $\braket{\psi^*|X^T|\phi^*}=\braket{\phi|X|\psi}$. If $X=I$ \emph{i.e}., identity operator, then it becomes $\braket{\psi^*|\phi^*}=\braket{\phi|\psi}$.




\bibitem{AAV-1988}
  { Y. Aharonov, D. Z. Albert, and L. Vaidman}, \href{https://doi.org/10.1103/PhysRevLett.60.1351}{  Phys. Rev. Lett. 60, 1351 (1988)}.




\bibitem{Kofman-Nori}
  { A.G. Kofman, Sahel Ashhab, and Franco Nori},  \href{https://www.sciencedirect.com/science/article/pii/S0370157312002050} {  Physics Reports 520 (2012) 43–133}.



\bibitem{Dressel-2014}
Justin Dressel, Mehul Malik, Filippo M. Miatto, Andrew N. Jordan, and Robert W. Boyd,
\href{https://link.aps.org/doi/10.1103/RevModPhys.86.307}{Rev. Mod. Phys. 86, 307 (2014)}.




\bibitem{Otachel-2018}
Zdzisław Otachel, \href{https://doi.org/10.1016/j.jmaa.2017.10.038}{J. Math. Anal. Appl. 458 (2018) 1409–1426}.




\bibitem{Wolf-2012}
M. M. Wolf, \href{https://mediatum.ub.tum.de/doc/1701036/document.pdf}{Quantum Channels and Operations Guided Tour (2012)} [see Lemma 3.1].



\bibitem{Braunstein-1994}
Samuel L. Braunstein and Carlton M. Caves, \href{https://journals.aps.org/prl/abstract/10.1103/PhysRevLett.72.3439}{Phys. Rev. Lett. 72, 3439 (1994)}.


\bibitem{Toth-2014}
Geza Toth and Iagoba Apellaniz, \href{https://iopscience.iop.org/article/10.1088/1751-8113/47/42/424006/meta}{J. Phys. A: Math. Theor. 47 424006 (2014)}.


\bibitem{Chiew-Gessner-2022}
Shao-Hen Chiew and Manuel Gessner, \href{https://journals.aps.org/prresearch/abstract/10.1103/PhysRevResearch.4.013076}{Phys. Rev. Research 4, 013076 (2022)}.


\bibitem{Sixia-2013}
Sixia Yu, \href{https://doi.org/10.48550/arXiv.1302.5311}{arXiv:1302.5311v1}.


\bibitem{Toth-Petz-2013}
Géza Tóth and Dénes Petz, \href{https://journals.aps.org/pra/abstract/10.1103/PhysRevA.87.032324}{Phys. Rev. A 87, 032324 (2013)}.


\bibitem{Toth-2018}
 Geza Toth, \href{https://doi.org/10.48550/arXiv.1701.07461}{arXiv:1701.07461v5}.





 \bibitem{Horodecki2009}
 Ryszard Horodecki \emph{et al.}, \href{https://journals.aps.org/rmp/abstract/10.1103/RevModPhys.81.865}{Rev. Mod. Phys. 81, 865 (2009)}; O. Gühne, G. Tóth, \href{https://www.sciencedirect.com/science/article/pii/S0370157309000623}{Physics Reports 474 (2009) 1–75}.




\bibitem{Wiseman-2007}
H. M. Wiseman, S. J. Jones, and A. C. Doherty, \href{https://link.aps.org/doi/10.1103/PhysRevLett.98.140402}{Phys. Rev. Lett. 98, 140402 (2007)}; Roope Uola, Ana C.S. Costa, H. Chau Nguyen,  and Otfried Guhne, \href{https://link.aps.org/doi/10.1103/RevModPhys.92.015001}{Rev. Mod. Phys. 92, 015001 (2020)}.





\bibitem{Luca-2009}
Luca Pezzé and Augusto Smerzi, \href{https://journals.aps.org/prl/abstract/10.1103/PhysRevLett.102.100401}{Phys. Rev. Lett. 102, 100401 (2009)}.



\bibitem{Toth-2012}
Géza Tóth, \href{https://journals.aps.org/pra/abstract/10.1103/PhysRevA.85.022322}{Phys. Rev. A 85, 022322 (2012)}.



\bibitem{Hyllus-2012}
Philipp Hyllus \emph{et al.}, \href{https://journals.aps.org/pra/abstract/10.1103/PhysRevA.85.022321}{Phys. Rev. A 85, 022321 (2012)}.



\bibitem{Dodonov-2018}
V. V. Dodonov, \href{https://link.aps.org/doi/10.1103/PhysRevA.97.022105}{Phys. Rev. A 97, 022105 (2018)}.


\bibitem{Qin-2016}
H.-H. Qin, S.-M. Fei, and X. Li-Jost, \href{https://doi.org/10.1038/srep31192}{Sci Rep 6, 31192 (2016)}.

\end{thebibliography}
\end{document}